\documentclass[journal]{IEEEtran}
\usepackage[numbers,sort&compress,square]{natbib}   % used for the advanced citer
\usepackage{graphicx}
\usepackage{subfig}
\usepackage{amsmath,amssymb}
\usepackage{array}
\usepackage{amsthm}
\usepackage{mathrsfs}
\usepackage{color}
\usepackage[ruled,linesnumbered]{algorithm2e}
\SetKw{KwDownTo}{down to}
\usepackage{hyperref}

\newtheorem{theorem}{Theorem}

\newtheorem{lemma}[theorem]{Lemma}

\theoremstyle{definition}
\newtheorem{define}[theorem]{Definition}

\newtheorem{remark}[theorem]{Remark}

\usepackage{mathtools}

\begin{document}

\title{Timely Information Updating for Mobile Devices Without and With ML Advice}

\author{\IEEEauthorblockN{Yu-Pin Hsu and Yi-Hsuan Tseng}
	
\thanks{A preliminary version of this work appeared in the Proc. of IEEE ISIT, 2019~\cite{tseng2019online}. Yu-Pin Hsu is with the Department of Communication Engineering, National Taipei University, New Taipei City 237303, Taiwan (e-mail: yupinhsu@mail.ntpu.edu.tw). Yi-Hsuan Tseng was with the Department of Communication Engineering, National Taipei University, New Taipei City 237303, Taiwan  (email: alicetseng1006@gmail.com).} 
	
}

\maketitle

\begin{abstract}
This paper investigates an information update system in which a mobile device monitors a physical process and sends status updates to an access point (AP). A fundamental trade-off arises between the timeliness of the information maintained at the AP and the update cost incurred at the device. To address this trade-off, we propose an \emph{online algorithm} that determines when to transmit updates using only available observations. The proposed algorithm asymptotically achieves the \emph{optimal competitive ratio} against an adversary that can simultaneously manipulate multiple sources of uncertainty, including the operation duration,  information staleness,  update cost, and update opportunities. Furthermore, by incorporating machine learning (ML) advice of unknown reliability into the design, we develop an ML-augmented algorithm that asymptotically attains the \emph{optimal consistency-robustness trade-off}, even when the adversary can additionally corrupt the ML advice. The optimal competitive ratio scales linearly with the range of  update costs, but is unaffected by other sources of uncertainty. Moreover,  an optimal   competitive online algorithm exhibits a threshold-like response to the ML advice: it either fully trusts  or completely ignores the ML advice, as partially trusting the advice cannot improve the consistency without severely degrading the robustness. Extensive simulations in stochastic settings further validate the theoretical findings in the adversarial environment. 
\end{abstract}


\section{Introduction}\label{section:intro}

In recent years, the demand for \emph{timely} information has surged across diverse  systems. 
In  Internet-of-Things (IoT) networks (e.g., unmanned aerial vehicles deployed for disaster response~\cite{gupta2015survey}), 
each IoT device is equipped with sensors (e.g., GPS, radar, and temperature sensors) that continuously monitor its surroundings. 
These sensors generate status updates about physical processes and transmit them to a central controller. 
By aggregating such updates, the controller has a real-time view of the environment, thereby enabling intelligent decision-making. 
Similarly, in location-based smartphone applications (e.g., navigation and gaming~\cite{karki2020characterizing}), 
users frequently report their locations to a central server so that the service can respond in real time. 
In both cases, a central entity relies on timely status updates from mobile devices to  perform time-sensitive inference tasks.

To quantify the timeliness of information maintained at a central entity, 
Kaul, Yates, and Gruteser introduced the \emph{age of information} metric in~\cite{kaul2012real}, defined as the time elapsed since the most recently received update was generated. Under this definition, the information at the central entity linearly ages with time until it is updated. In addition to the linear aging function, more general nonlinear aging functions \cite{kosta2017age} have also been analyzed.  These functions further characterize the quality of an update, e.g.,  capturing how quickly the information held by the central entity deviates from the true status, or  representing the penalty associated with using outdated information in decision-making. In this paper, we consider   general aging functions. 
%They showed that throughput-optimal or delay-optimal designs may fail to minimize AoI, 
%thereby motivating system designs that explicitly account for information timeliness.

While frequent updates reduce the age of information at a central entity, 
they also incur substantial update costs (e.g., energy consumption and bandwidth utilization) at  local devices. 
Such  costs are particularly significant for resource-limited mobile devices (e.g., battery-powered and bandwidth-constrained IoT devices or smartphones).  We therefore investigate the fundamental trade-off between  the age of information at the central entity and the update cost at the mobile device. Specifically, this paper considers an information update system in which a mobile device monitors a physical process and reports its latest status to a nearby access point (AP). 
To balance information timeliness at the AP with resource consumption at the device, 
a scheduling algorithm that determines when to transmit updates is crucial. 
Our goal is to design such an algorithm to minimize the total cost over the operation duration, where the total cost jointly accounts for an \emph{age cost} (representing the AP’s information staleness) 
and an \emph{update cost} (representing the device’s resource expenditure).

The scheduling problem is complicated by several forms of uncertainty in mobile networks:  
1) The \emph{operation duration} is uncertain, e.g., the runtime of a location-based application depends on how long a smartphone user keeps the application active.  
2) The \emph{age increment} may vary over time, e.g., a location update
becomes stale more quickly when the device moves at a higher speed.
3) The \emph{update cost} is also time-varying, e.g., user mobility  causes fluctuations in energy consumption.
4) The device’s \emph{update opportunities} may also be intermittent. Such cases include sporadic update arrivals (e.g., due to misalignment between the update generation periods and transmission slots) and transmission constraints that prevent the device from sending in certain slots (e.g., due to a power-saving policy~\cite{lin2022survey} or uplink scheduling decisions imposed by the AP~\cite{takeda2020understanding}).

Most prior works modeled uncertainties using stationary stochastic processes, 
e.g., employing an M/M/1 queueing model in~\cite{kaul2012real} to represent the update arrival and service processes. 
However, such assumptions are often unrealistic, e.g., when a device moves arbitrarily so that the service time no longer follows an exponential distribution. Even if such models fit reality reasonably well, the operation duration may be too short for the process to converge to stationarity. Moreover, practical stochastic models for operation duration or the information aging are often unclear. Such non-stationary uncertainties pose the central challenge for our scheduling design. Under non-stationarity, a scheduling algorithm cannot rely on future knowledge and must instead operate solely based on past and present observations,  as an \textit{online}  algorithm. 
%Furthermore, while update arrivals up to the current time are observable, the current update cost (e.g., energy expenditure) is unknown prior to making a transmission decision but is revealed only after the update is sent. 
%This limited prior knowledge further complicates the scheduling problem.

Our first contribution is the design and analysis of online scheduling algorithms that operate under observable information. 
Specifically, the proposed algorithm requires knowledge only of the current age increment of the status held by the AP and also  whether an update opportunity is currently available,  without relying on any prior knowledge of the operation duration or the entire sequences of status aging,  update costs, and update opportunities. Let $R$ denote the ratio between the maximum and minimum update cost. Our main result establishes that, asymptotically in the large update cost regime, the proposed algorithm achieves a total cost at most $\frac{e^{1/R}}{e^{1/R}-1}=\mathcal{O}(R)$ (known as the \textit{competitive ratio}) times the minimum total cost attained by an optimal offline algorithm with  complete knowledge of all uncertainties. This competitive ratio turns out to be optimal. Thus, we can observe that the optimal competitive ratio scales linearly with the range of update costs, while remaining independent of all other sources of uncertainty.

The above guarantee holds uniformly over all uncertainty instances, including the worst-case scenario.  Such worst-case analysis can be overly pessimistic, since in practice future events often follow patterns that would be predicted using machine learning (ML). Motivated by this, Lykouris and Vassilvitskii \cite{lykouris2018competitive} proposed incorporating (potentially imperfect) ML advice into online algorithms, as an approach that goes beyond the worst-case analysis.  A central challenge in this setting is that the reliability of the ML advice is generally unknown. Thus, the design goal is twofold: 1) when the ML advice is accurate, the algorithm should perform well; 2) when the advice is unreliable, the algorithm should still provide performance guarantees. However, it is impossible to achieve both properties simultaneously. For example, an algorithm that blindly trusts the ML advice  performs excellently under accurate ML advice but can suffer arbitrarily poor performance when the ML advice is wrong. Hence, the objective is to optimally balance these two properties.

Our second contribution is to integrate untrusted ML advice (specifically,  advice on the next update time) into our online scheduling framework. We introduce a hyperparameter $\lambda \in (0,1]$ to control the level of trust in the ML advice, where a smaller $\lambda$ places greater reliance on the ML advice. Our main result is  that the proposed  algorithm achieves the following asymptotic trade-off: 1) it  achieves a total cost at most $\frac{\lambda e^{\lambda/R}}{e^{\lambda/R}-1}=\mathcal{O}(R)$ (known as the \textit{consistency}) times the cost of blindly following the ML advice; 2) it also achieves a total cost at most  $\frac{e^{\lambda/R}}{e^{\lambda/R}-1}=\mathcal{O}(\frac{R}{\lambda})$ (known as the \textit{robustness}) times the minimum cost achieved by an optimal offline algorithm. Again, this balance depends only on the ratio $R$ and turns out to be optimal. We can observe that partially trusting the ML advice with $\lambda \in (0,1)$  cannot leverage the ML advice, as it yields (asymptotically) no improvement in the consistency. Thus, an optimal online algorithm  displays an essentially threshold-type behavior with respect to ML advice, either fully adopting it or ignoring it altogether.

\section{Related work}

Extensive research has been conducted on analyzing and minimizing the age of information in diverse system settings. For example, Costa et al. \cite{costa2016age} derived closed-form expressions for the average age in single-source systems; then,  Yates and Kaul \cite{yates2018age} extended the analysis to multi-source scenarios. Building on the foundational work, numerous system design strategies have been proposed to minimize the age, including scheduling algorithms \cite{kadota2018scheduling}, resource allocation schemes \cite{park2020centralized}, and sampling strategies \cite{ornee2019sampling}. Beyond solely minimizing the age, several trade-off problems have also been investigated, such as age–throughput trade-off \cite{mankar2021throughput,wang2025understanding} and age–energy trade-off \cite{nath2018optimum,gu2019timely}. A comprehensive survey of these efforts is provided in \cite{yates2021age}.

Most prior age-related works assume stationary stochastic processes to model uncertainties (e.g., \cite{yates2018age,kadota2018scheduling,park2020centralized,ornee2019sampling, mankar2021throughput,wang2025understanding, nath2018optimum,gu2019timely}). Since such assumptions can be overly optimistic, several works have examined how non-stationary (adversarial) environments affect information timeliness from different perspectives. Examples include adversarial ON/OFF channels \cite{tseng2019online,sinha2022optimizing}, adversarial update arrivals \cite{saurav2023online}, and adversarial aging functions \cite{lin2025optimal,tripathi2021online}. Recent work \cite{liu2025learning} further incorporated ML advice into online algorithms for adversarial ON/OFF channels.

To the best of our knowledge, there is no unified design and analysis
framework capable of handling an adversary that simultaneously controls multiple sources of uncertainty as in our model, where the adversary can jointly manipulate the operation duration, information aging, update cost, and update opportunities. This gap is critical, since mobile networks inherently involve several forms of non-stationary uncertainty, and it is also technically challenging because the adversary is so
powerful. Particularly, the impact of adversarially varying update costs (beyond simple ON/OFF channel models) on age-driven design has not been explored in the existing literature, and our results reveal that it is the most critical source of uncertainty affecting performance.

We address these challenges gradually. 
Sections~\ref{section:system}–\ref{section:onine} focus on the first three types of uncertainty introduced in Section~\ref{section:intro}, 
and Section~\ref{section:non-saturated} further generalizes the results to incorporate the fourth type of uncertainty.

%Extending the analysis beyond ON/OFF channel models is especially relevant in mobile networks, where users may move arbitrarily and the device s also can dynamically adjust their transmission power.

%To tackle this challenging problem, we leverage  online algorithm design techniques for linear programs   \cite{buchbinder2009design,bamas2020primal}. By establishing an appropriate mapping,  we show that our proposed linear program formulation  generalizes the classical online Transmission Control Protocol (TCP) acknowledgment (ACK) problem \cite{buchbinder2009design,bamas2020primal}. In particular, our setting extends theirs to scenarios where the ACK channel can alternate between ON and OFF states, and transmitting an ACK during an ON slot incurs an adversarially chosen cost. This extension is both practically relevant (e.g., noisy wireless networks) and theoretically challenging, since (i) the adversary simultaneously controls multiple sources of uncertainty, and (ii) the online algorithm can make ACK transmission decisions only in ON slots.

\section{System overview} \label{section:system}

\begin{figure}
	\centering
	\includegraphics[width=.45\textwidth]{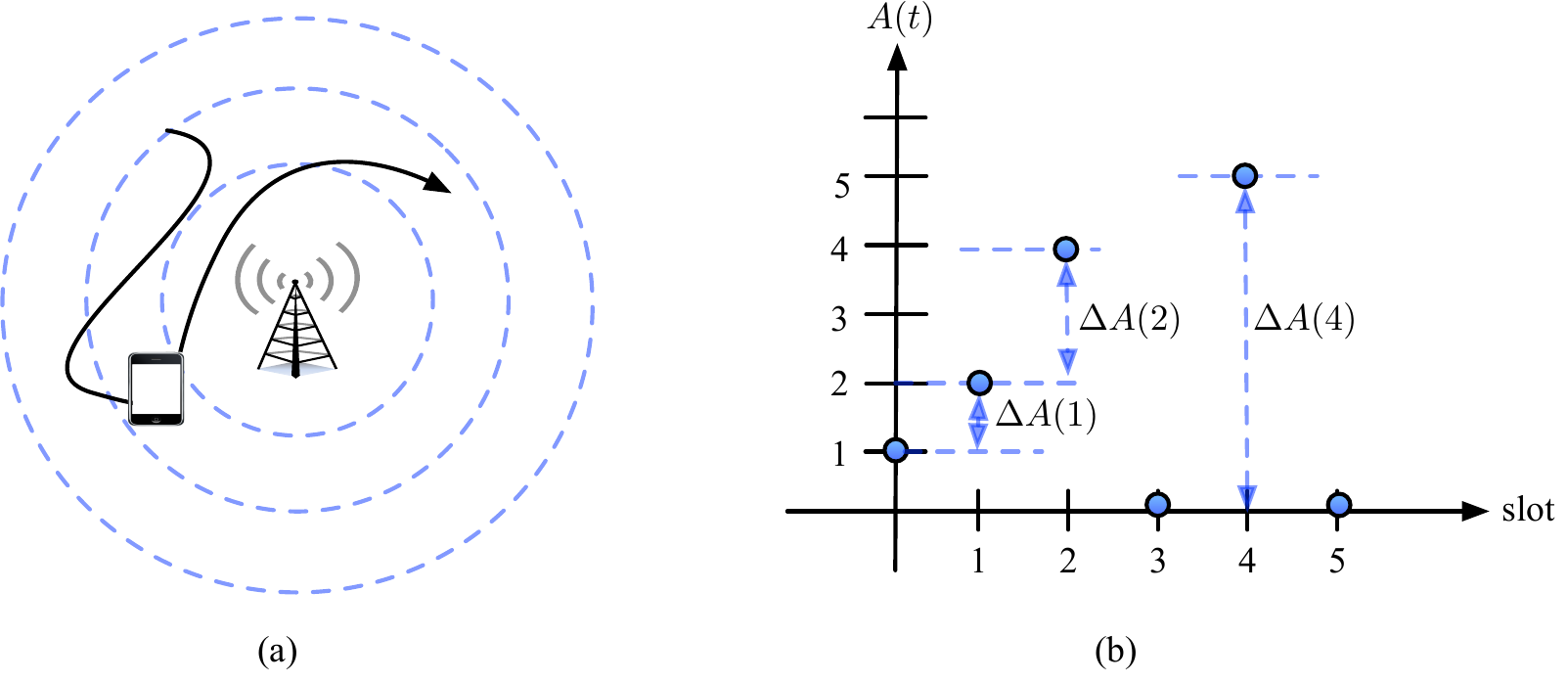}
	\caption{An example network model: (a) a mobile device updating an AP; (b)	the age of information at the AP when the device sends updates in	slots~3 and~5. }
	\label{fig:model}
\end{figure}

As illustrated in Fig.~\ref{fig:model}(a), we consider an information update system in which a mobile device monitors a physical process and reports its latest status to a nearby access point (AP). The system operates in discrete time slots indexed by \mbox{$t=1,2,\cdots,T$}, where $T$ represents the total operation duration.

%This study primarily focuses on a single mobile user. By adopting this simplified setting, we aim to clearly convey our innovations for addressing several types of uncertainties that lack stationary probabilistic models. {\color{red}For the case with multiple users, refer to Theorem~\ref{theorem:competitive-ratio} for an asymptotic analysis and Section~\ref{} for further discussion.}
We begin with a scenario in which the device always has an update packet at the beginning of every slot and is also permitted to transmit it in every slot. Then, for each slot $t$, the device decides whether to transmit the update. Let $d(t) \in \{0,1\}$ denote the  device’s transmission decision, where $d(t) = 1$ if the device transmits in slot~$t$, and $d(t) = 0$ otherwise. If the device transmits at the beginning of a slot, the update is delivered by the end of that slot.
%; that is, the slot duration corresponds to the time required (with high probability)  to deliver an update packet using maximum transmit power. 
In Section~\ref{section:non-saturated}, we extend the model to more general scenarios in which the device cannot transmit updates in certain slots, i.e., under intermittent update opportunities.

%To deliver a clear insight of our scheduling design for the unpredictable mobile device,  this paper focuses on the bottleneck between the device and the user (through wireless networks), but neglects the transmission time between the information source and the device (through wired networks), i.e., the device can obtain an update immediately at the beginning of a slot. 

\subsection{Age of information}

If the device decides to transmit an update at the beginning of slot~$t$, then the age of information at the AP is reset to zero at the end of slot~$t$, indicating that the AP has received the latest update. Otherwise, the age increases by an amount $\Delta A(t)$ to reflect the continued staleness of the information at the AP. The value of $\Delta A(t)$ can vary across slots~$t$. Let $A(t)$ denote the age of information maintained by the AP at the end of slot~$t$.  As illustrated in Fig.~\ref{fig:model}(b), the evolution of $A(t)$ across time slots is given by
\begin{align}
	A(t) = 
	\begin{cases}
		A(t-1) + \Delta A(t), & \text{if } d(t) = 0, \\
		0,                 & \text{if } d(t) = 1,
	\end{cases}
	\label{eq:age-dynamic}
\end{align}
with initial age $A(0) = A_0$, where $A_0$ is specified by the AP during the initial connection. We define the \emph{age increment sequence} as $\boldsymbol{\Delta A} = (\Delta A(1), \Delta A(2), \cdots, \Delta A(T))$.

%We remark that, according to the foundational work in~\cite{.}, the AoI is conventionally set to one upon successful delivery, to account for the one-slot delay incurred by the transmitted update. Under that convention, the AoI at each slot~$t$ corresponds to our model's $A(t) + 1$. However, as we explain in Remark~\ref{remark:age-answer}, our model in Eq.~\eqref{eq:age-dynamic} is sufficient and appropriate for the optimization problem at hand.

\subsection{Problem formulation} \label{subsection:problem}
While transmitting updates in every slot minimizes the age of information at the AP, it also incurs substantial resource consumption (e.g., energy and bandwidth) at the device. To capture this trade-off, we introduce two cost metrics: the \emph{age cost} and the \emph{update cost}. Specifically, we assume that each unit of age in a slot incurs a cost of one unit; thus, the age cost in slot~$t$ is given by $A(t)$. In addition, if the device  transmits an update in slot~$t$, it incurs an update cost denoted by $C(t)$. For instance, $C(t)$ can be modeled as the product of a unit  cost $C_u(t)$ and the transmission energy $\mathcal{E}(t)$  by $C(t) =C_u(t)  \mathcal{E}(t)$. Here, $\mathcal{E}(t)$ depends on the instantaneous channel condition between the device and the AP, which may fluctuate due to user mobility. Meanwhile, $C_u(t)$ may also vary over time, e.g., depending on the device’s remaining energy or, when multiple packets are present, different unit costs can be assigned to prioritize certain packets over others. The form $C_u(t)\mathcal{E}(t)$ can also be interpreted more generally as a unit cost multiplied by  resource expenditure (e.g., energy or bandwidth). However, for clarity, in the remainder of this paper we focus on the energy example. We define the \emph{update cost sequence} as $\mathbf{C} = (C(1), \cdots, C(T))$.

To balance the age cost and the update cost, the device needs a scheduling algorithm defined as $\boldsymbol{\pi} = (d(1), \cdots, d(T))$. The total cost incurred by a scheduling algorithm $\boldsymbol{\pi}$ depends on the sources of uncertainty, including the operation duration~$T$, the age increment sequence $\boldsymbol{\Delta A}$, and the update cost sequence $\mathbf{C}$. We represent this uncertainty instance as $\mathcal{I} = \{T, \boldsymbol{\Delta A}, \mathbf{C}\}$.
Given an instance $\mathcal{I}$ and a scheduling algorithm $\boldsymbol{\pi}$, the total cost is defined as the sum of age and update costs:
\begin{align}
	J(\mathcal{I}, \boldsymbol{\pi}) = \sum_{t=1}^T \Big( C(t)  d(t) + A(t) \Big).
	\label{eq:cost}
\end{align}
Our objective is to design a scheduling algorithm $\boldsymbol{\pi}$ that minimizes the total cost $J(\mathcal{I}, \boldsymbol{\pi})$.

\subsection{Scheduling  algorithm classification}

A scheduling algorithm is referred to as an \emph{offline scheduling algorithm} if it has prior knowledge of the entire instance~$\mathcal{I}$. Such algorithms are generally impractical in real-world systems due to their reliance on future information. Thus, this paper focuses on a more realistic design, where the device has access only to the historical or current information but lacks knowledge of future values. 

Because of the potential unavailability of real-time channel information, the current update cost $C(t)$ may not be  known at the beginning of slot~$t$. Therefore, we design scheduling algorithms that do not rely on instantaneous channel knowledge. Instead, the algorithms use only the maximum and minimum possible update costs, denoted by $C_M$ and $C_m$. 
%The value of $C_M$ can be determined, for example, as the product of the maximum unit cost, the maximum transmission power, and the slot length. 
The values of $C_M$ and $C_m$ can be estimated from the historically observed worst and best channel conditions (e.g., those observed by the AP within its service region and reported by the AP). 
If these values are unavailable, see Remark~\ref{remark:cost} for a slight modification of the proposed algorithms that preserves some performance guarantee.

A scheduling algorithm is called a \emph{(cost-agnostic) online scheduling algorithm} if it requires only the constants $C_M$ and $C_m$ and the realized age increments up to the current slot. For simplicity, we will omit the term \emph{cost-agnostic}, with the understanding that all references to online scheduling algorithms refer to the cost-agnostic setting. Without access to the complete uncertainty instance, an online algorithm is generally unable to achieve the minimum total cost attainable by an optimal offline scheduling algorithm. Given an instance~$\mathcal{I}$, let $\mathrm{OPT}(\mathcal{I})$ denote the minimum total cost achieved by  an optimal offline algorithm. We evaluate the performance of online algorithms in terms of their \emph{competitiveness}~\cite{buchbinder2009design} relative to the offline optimum, defined as follows.

\begin{define}
A scheduling algorithm~$\boldsymbol{\pi}$ is said to be \mbox{\textbf{$\gamma$-competitive}} if $\frac{J(\mathcal{I}, \boldsymbol{\pi})}{\mathrm{OPT}(\mathcal{I})} \leq \gamma$, for all possible instances~$\mathcal{I}$.
\end{define}

That is, a $\gamma$-competitive online scheduling algorithm guarantees that the resulting total cost is at most $\gamma$ times the offline minimum cost, \emph{regardless of the instance~$\mathcal{I}$}. Our goal is to design an online  scheduling algorithm that achieves the smallest possible competitive ratio~$\gamma$.

Note that in addition to the competitive ratio, another common performance metric is \emph{regret} \cite{lattimore2020bandit}. Regret measures the additive performance gap between an \emph{online learning approach} and the best offline algorithm restricted to  fixed decision rules. In contrast, the competitive ratio compares the performance of an online algorithm against the best offline algorithm without such restrictions. Moreover, online learning approaches typically provide regret guarantees only as the decision horizon grows to infinity. Such guarantees are less suitable for our setting, since a  device may monitor and transmit updates only for a short and unpredictable duration.

Moreover, with the advancement of machine learning (ML) techniques, it is increasingly feasible to leverage ML  to provide scheduling advice. A scheduling algorithm is referred to as an \emph{online scheduling algorithm with ML} if it additionally has access to ML advice. Let $\mathcal{M}(t) \in \{0, 1\}$ denote the decision advised by ML at slot~$t$. We define the sequence $\boldsymbol{\mathcal{M}} = (\mathcal{M}(1), \cdots, \mathcal{M}(T))$ as the \emph{ML advice}. Because such a scheduling algorithm can adapt its actions based on the ML advice, its decision sequence $\boldsymbol{\pi}$ is allowed to be a function of~$\boldsymbol{\mathcal{M}}$.

This paper considers the setting where the ML advice may be  untrusted. While following perfect ML advice (which can minimize the total cost) yields the minimum total cost, blindly trusting  imperfect ML advice (which cannot minimize the total cost) may lead to poor performance. Moreover, we assume that the reliability of the ML advice is unknown a priori. In this context,  we characterize online scheduling algorithms with ML in terms of two metrics introduced in~\cite{lykouris2018competitive}: \emph{consistency} and \emph{robustness}. Consistency quantifies performance relative to the ML advice, while robustness guarantees performance in the worst case. These notions are formally defined as follows.

\begin{define}
An online scheduling algorithm~$\boldsymbol{\pi}$ with ML advice is said to be \textbf{$\alpha$-consistent} if $\frac{J(\mathcal{I}, \boldsymbol{\pi})}{J(\mathcal{I}, \boldsymbol{\mathcal{M}})} \leq \alpha$, for all possible instances~$\mathcal{I}$ and ML advice~$\boldsymbol{\mathcal{M}}$; it is said to be \textbf{$\beta$-robust} if $\frac{J(\mathcal{I}, \boldsymbol{\pi})}{\mathrm{OPT}(\mathcal{I})} \leq \beta$, 
for all possible instances~$\mathcal{I}$ and advice~$\boldsymbol{\mathcal{M}}$.
\end{define}

In other words, an $\alpha$-consistent and $\beta$-robust online scheduling algorithm ensures that 1) when the ML advice is perfect, the resulting total cost is at most $\alpha$ times the offline minimum cost, and  2) when the advice is arbitrary or even adversarial, the cost remains within a factor of $\beta$ of the offline minimum cost. An algorithm that fully trusts the ML advice may achieve near-optimal consistency, but this often comes at the expense of robustness. Therefore, there exists an inherent trade-off between consistency~$\alpha$ and robustness~$\beta$. The goal of this paper is to design an online scheduling algorithm with ML that achieves the optimal consistency-robustness trade-off, namely, to minimize the consistency~$\alpha$ for any fixed robustness~$\beta$.

\section{Linear program formulation  for offline scheduling} \label{section:offline}
The main challenge in designing our scheduling algorithm arises from several forms of uncertainty that are impractical to model using stationary stochastic processes and also limited current observations. To address these challenges, we leverage online algorithm design techniques based on \textit{offline} optimal linear programming~\cite{buchbinder2009design,bamas2020primal}. 

However, casting our offline scheduling problem as an  linear program (LP) is non-trivial due to the non-linear nature of the age cost. For example, consider a scenario where the device  transmits an update in slot~$1$ and schedules the next update in slot~$x$. If the age increases linearly by one unit per slot until the next update, the cumulative age from slot~$1$ to slot~$x$ is given by $\sum_{t=1}^{x} t = x(x+1)/2$, which grows quadratically with the decision variable~$x$. See~\cite{arafa2017age} for a concrete example illustrating this behavior. This quadratic growth implies that the total age cost $\sum_{t=1}^{T} A(t)$ in Eq.~\eqref{eq:cost} includes \emph{non-linear terms}, thereby complicating direct LP formulation.

To overcome this issue, we introduce a transformation of the age evolution into an equivalent \emph{virtual queueing system}, described in Section~\ref{subsection:virtual-queue}. This transformation facilitates an LP formulation for the offline scheduling problem, as presented in Section~\ref{subsection:primal-dual}. The resulting LP formulation serves as the foundation for the design and analysis of our online  scheduling algorithms: Section~\ref{section:onine} develops an online scheduling algorithm without ML, Section~\ref{section:ml} incorporates ML advice into the scheduling process, and Section~\ref{section:non-saturated} extends the model to intermittent update opportunities.

\subsection{Virtual queueing system} \label{subsection:virtual-queue}

\begin{figure}[t]
	\centering
	\includegraphics[width=.4\textwidth]{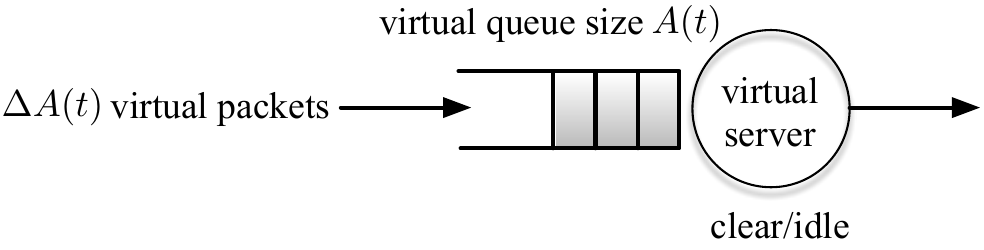}
	\caption{Virtual queueing system.}
	\label{fig:virtual}
\end{figure}

Without loss of generality, we assume that the age increment $\Delta A(t)$ is an integer for all $t$. If this is not the case, we can multiply both $C(t)$ and $\Delta A(t)$ by a common constant so that every $\Delta A(t)$ becomes integer-valued. Such scaling does not alter the optimal solution to the objective in Eq.~(\ref{eq:cost}). Based on this assumption, we  introduce a virtual queue that mirrors the evolution of the integer-valued age.

We construct a virtual queueing system (shown in Fig.~\ref{fig:virtual}) consisting of a virtual server, a virtual queue, and virtual packet arrivals. The virtual system operates in the same discrete time slots as the real mobile network. Initially, the virtual queue contains $A_0$ virtual packets. At the beginning of each slot~$t$, $\Delta A(t)$ virtual packets arrive at the virtual queue. If the device decides to transmit an update  in slot~$t$ (in the actual network), then the virtual server \emph{clears} the virtual queue at the end of the slot.  Otherwise, the virtual server remains \emph{idle} and the virtual packets accumulate. As a result, the virtual queue size evolves as follows: it resets to zero if $d(t) = 1$ (i.e., the virtual server clears the virtual queue, corresponding to an update in the actual network), or increases by $\Delta A(t)$ if $d(t) = 0$ (i.e., the virtual server idles, corresponding to no update in the actual network).  This evolution exactly mirrors the age dynamics in Eq.~\eqref{eq:age-dynamic}. Therefore, we use the same notation $A(t)$ to denote the virtual queue size at the end of slot~$t$.

We index the virtual packets by $1, 2, \cdots$ according to their arrival times, and let $T_i$ denote the slot in which virtual packet~$i$ arrives. For each virtual packet~$i$, we use a binary variable $z_i(t) \in \{0, 1\}$ to indicate whether it remains in the virtual queue at the end of slot~$t$, where $z_i(t) = 1$ if it is still present, and $z_i(t) = 0$ otherwise. Using this notation, the virtual queue size  at the end of slot~$t$ can be expressed as $A(t) = \sum_{i: T_i \leq t} z_i(t)$,
which  counts the number of virtual packets that have arrived by slot~$t$ and remain in the virtual queue. This representation allows us to express the age $A(t)$ as a linear function of the binary variables $z_i(t)$. Substituting this expression into Eq.~\eqref{eq:cost}, we can rewrite the total cost $J(\mathcal{I}, \pi)$ as the following linear function:
\begin{align}
	J(\mathcal{I}, \pi) = \sum_{t=1}^T \left( C(t) d(t) + \sum_{i: T_i \leq t} z_i(t) \right).
	\label{eq:cost-2}
\end{align}
 This linear expression facilitates the formulation of an LP in the next section. Moreover, by Eq.~(\ref{eq:cost-2}), the update cost $C(t)$ can also be interpreted as a \emph{clearing cost} incurred when the virtual queue is cleared in slot~$t$, while holding a virtual packet for one slot incurs a unit \emph{holding cost}.

\subsection{LP formulation} \label{subsection:primal-dual}
We note that the clear/idle behavior in the virtual queueing system directly corresponds to sending/withholding an acknowledgment (ACK) to clear all received packets in the Transmission Control Protocol (TCP). From this perspective, we can leverage prior studies for the classic online TCP ACK problem \cite[Chapter~12]{buchbinder2009design} to formulate our \emph{offline} optimal scheduling problem as an integer program with a linear objective and constraint functions:
\begin{subequations}	\label{integer-program}
	\begin{align}
	\min_{x(t),z_i(t)} & \hspace{.2cm} \sum_{t=1}^T \left(C(t)  x(t) + \sum_{i: T_i \leq t} z_{i}(t)\right) \label{integer-program:objective}\\
	\text{s.t.} & \hspace{.2cm} z_i(t) + \sum_{\tau = T_i}^{t} x(\tau) \geq 1, 	\nonumber\\
	& \hspace{1cm}\text{for all $i$ such that $T_i \leq t$ and for all $t$};
	\label{integer-program:const1}\\
	& \hspace{.2cm} x(t),\, z_i(t) \in \{0,1\}, \quad \text{for all $i$ and $t$.} \label{integer-program:const2} 
	\end{align}
\end{subequations}
In this integer program, we introduce a variable $x(t)$ to denote whether the device  transmits an update in slot~$t$. That is, $x(t)$ in the integer program is exactly the decision variable $d(t)$ introduced earlier. The reason for redefining this variable is  that we will relax $x(t)$ to take a real value between 0 and~1, which prevents immediate interpretation as a transmission decision. Later in Section~\ref{section:onine}, we show how to convert a fractional solution for $x(t)$ into a randomized scheduling decision for $d(t)$.  Moreover, the constraint in Eq.~\eqref{integer-program:const1} ensures that each virtual packet~$i$ arriving by slot~$t$ either remains in the virtual queue at the end of slot~$t$ (i.e., $z_i(t) = 1$ in the first term of Eq.~\eqref{integer-program:const1}) or has been cleared  by slot~$t$ (i.e., there exists a slot $\tau \in \{T_i, \cdots, t\}$ such that $x(\tau)=1$ in the second term of Eq.~\eqref{integer-program:const1}).

By relaxing the integrality constraint~\eqref{integer-program:const2} to allow continuous variables, we obtain the following LP:
\begin{subequations}	\label{lp}
	\begin{align}
	\min_{x(t), z_i(t)} & \hspace{.2cm} \sum_{t=1}^T \left(C(t)x(t) + \sum_{i: T_i \leq t} z_{i}(t)\right) \label{lp:objective}\\
	\text{s.t.} & \hspace{.2cm} z_i(t) + \sum_{\tau = T_i}^{t} x(\tau) \geq 1, \nonumber\\
	& \text{\hspace{1cm}\text{for all $i$ such that $T_i \leq t$ and for all $t$};} 	\label{lp:const1}\\
	&\hspace{.2cm} x(t),\, z_i(t) \geq 0, \quad \text{for all $i$ and $t$.} \label{lp:const2} 
	\end{align}
\end{subequations}

Next, Section~\ref{section:onine}  proposes an online algorithm to compute a feasible solution to  LP~\eqref{lp} without relying on ML, while Section~\ref{section:ml} extends this approach by incorporating ML advice. 

\begin{remark}\label{remark:tcp1}
Before solving our online problem, we remark that, via the virtual queue transformation, our formulation generalizes the classical online TCP ACK problem~\mbox{\cite[Chapter~12]{buchbinder2009design}} and its learning-augmented variant~\cite{bamas2020primal}.
In the classical TCP ACK setting, each ACK incurs a constant cost. In contrast, our objective in~Eq.~\eqref{integer-program:objective} allows the ACK (i.e., clearing) cost to vary across slots. Later, in Section~\ref{section:non-saturated}, we further generalize the problem to scenarios where the ACK channel alternates between ON and OFF states, and transmitting an ACK during an ON slot incurs an adversarially chosen cost. This generalization is practically relevant in noisy wireless environments; moreover, it poses theoretical challenges, since the adversary simultaneously controls multiple sources of uncertainty and an online algorithm is restricted to  transmitting an ACK only in ON slots. As we will show, our generalized setting also yields a fundamentally different optimal competitive ratio that depends solely on the cost range (i.e., time-varying costs are the dominant factor). When augmented with ML advice, our generalized setting  exhibits behavior that differs qualitatively from the classic learning-augmented TCP model; in particular, it yields a threshold-like optimal trust rule when the cost range is large.
\end{remark}

\section{Online scheduling algorithm design without ML} \label{section:onine}
This section develops an online scheduling algorithm \emph{without} ML by leveraging  LP~\eqref{lp}. Section~\ref{subsection:lp-alg} introduces an online algorithm that can compute a feasible solution to the proposed LP in an online fashion.  Based on this fractional solution, Section~\ref{subsection:online-alg} proposes a randomized online scheduling algorithm without ML.

%{\color{red}The primal-dual  algorithm  produces the solution to the primal program in the online fashion, i.e., it produces a solution for $d(t)$ and $z_i(t)$ for all $i$ in slot $t$ according to  $p(t)$ only and will not change the  solution after slot~$t$.}

%{\color{red} In contrast, the primal-dual algorithm can produce a solution to the dual program in the offline fashion. That is because the solution  to the dual program  is  used for analyzing the primal objective value in Eq.~(\ref{lp:objective})  by the primal-dual algorithm, but not for the online scheduling design.}

\subsection{Online LP algorithm} \label{subsection:lp-alg}
We propose Alg.~\ref{lp-alg}, referred to as the \emph{online LP algorithm}, which computes a feasible solution to LP~\eqref{lp}. All variables are initialized to zero in Line~\ref{lp-alg:initial}. At the beginning of each slot~$t$, the algorithm iteratively adjusts the variables for all virtual packets that have arrived by slot $t$, as specified in Line~\ref{lp-alg:for}.

The underlying idea is that in each slot~$t$, our scheduling algorithm that will be proposed in Section~\ref{subsection:online-alg} makes a probabilistic decision: to set $d(t) = 1$ with some probability or $d(t) = 0$ otherwise. The probability is governed by the current value of $x(t)$, which is determined in Line~\ref{lp-alg:x} of Alg.~\ref{lp-alg}.
Accordingly, $x(t)$ can be interpreted as the probability of clearing the virtual queue in slot~$t$. In this context, the cumulative sum $\sum_{\tau = T_i}^t x(\tau)$ represents the cumulative clearing probability (up to slot~$t$) for virtual packet~$i$.

With this interpretation, the condition in Line~\ref{lp-alg:condition} checks whether virtual packet~$i$ has already been cleared by slot~$t$. If $\sum_{\tau = T_i}^t x(\tau) \geq 1$, virtual packet~$i$ is considered cleared, and no further processing is required. Otherwise, the virtual packet may still remain in the virtual queue and its associated variables should be adjusted. As shown in Line~\ref{lp-alg:for}, for each such packet, Line~\ref{lp-alg:x} increases the value of $x(t)$. That is, the more virtual packets remain in the virtual queue, the higher the resulting clearing probability.

Moreover, the idea behind Line~\ref{lp-alg:x} is that it  adjusts the cumulative clearing probability $\sum_{\tau = T_i}^t x(\tau)$ as follows:
\begin{align}
\sum_{\tau = T_i}^t x(\tau) 
&\leftarrow \sum_{\tau = T_i}^t x(\tau) + \text{increment of $x(t)$ in Line~\ref{lp-alg:x} }\nonumber\\
&\quad =\left(1 + \frac{1}{C_M}\right) \sum_{\tau = T_i}^t x(\tau) + \frac{1}{\theta C_M}, \label{eq:cum-x-update}
\end{align}
which increases the cumulative clearing probability by a multiplicative factor of $1 + (1/C_M)$ and an additive factor of $1/(\theta C_M)$. The constant $\theta$ is chosen as in Line~\ref{lp-alg:constant} so that the algorithm asymptotically achieves the minimum achievable competitive ratio (as stated  in Lemma~\ref{lemma:converse}).
 The appearance of $C_M$ in the denominators reflects that a larger update cost reduces the rate at which the clearing probability  increases. In addition, Line~\ref{lp-alg:z} sets $z_i(t) = 1 - \sum_{\tau = T_i}^t x(\tau)$ to ensure that the constraint in Eq.~\eqref{lp:const1} is satisfied, so that Alg.~\ref{lp-alg} produces a feasible solution to  LP~\eqref{lp}.

Note that Alg.~\ref{lp-alg} operates in an online manner, as it requires only the constants $C_M$ and $C_m$, and the knowledge of  virtual arrivals up to the current slot (which corresponds to the age increment sequence up to the current slot), without relying on any future information.

\begin{algorithm}[t]
	\SetAlgoLined 
	\SetKwFunction{Union}{Union}\SetKwFunction{FindCompress}{FindCompress} \SetKwInOut{Input}{input}\SetKwInOut{Output}{output}
%	\Input{$c$.}
%	\Output{$d(t)$.}
%	

	\tcc{Initialize all variables as follows:}
	$x(t)$, $z_i(t)$ $\leftarrow 0$ for all  $i$ and  $t$\;  \label{lp-alg:initial}
 
 						$\theta \leftarrow (1+\frac{1}{C_M})^{C_m}-1$\; 	\label{lp-alg:constant}

	\tcc{At the beginning of   slot $t=1, \cdots, T$, adjust all variables as follows:}

		\ForEach{virtual packet such that $T_i \leq t$\label{lp-alg:for}}{	
				\If{$\sum_{\tau=T_i}^{t}x(\tau)<1$  \label{lp-alg:condition}}{			
				$z_{i}(t) \leftarrow 1- \sum_{\tau=T_i}^{t}x(\tau)$\;  \label{lp-alg:z}

				$x(t)\leftarrow x(t)+ \frac{1}{C_M}\sum_{\tau=T_i}^{t} x(\tau)+\frac{1}{\theta C_M}$\; 	 \label{lp-alg:x}	
				
				}
			}\label{lp-alg:forend}

	\caption{Online LP algorithm without ML}
	\label{lp-alg}
\end{algorithm}

\subsection{Analysis of Alg.~\ref{lp-alg}}\label{subsection:analysis}

%In this section, we establish the competitiveness of the proposed Alg.~\ref{lp-alg}. As a first step, we analyze the feasibility of the solution it produces.
%
%\begin{lemma} \label{lemma:primal-feasible}
%Alg.~\ref{lp-alg} produces a feasible solution to  primal program~\eqref{lp}.
%\end{lemma}
%
%\begin{proof}
%(Sketch) The constraint in Eq.~(\ref{lp:const1}) is satisfied due to the update in Line~\ref{lp-alg:z}. See Appendix~\ref{appendix:lemma:primal-feasible} for details. 
%\end{proof}
In this section, we analyze the objective value in Eq.~\eqref{lp:objective} computed by Alg.~\ref{lp-alg}. Unlike prior studies  \cite{buchbinder2009design,bamas2020primal} that analyze online algorithms for LPs using primal–dual techniques, our analysis exploits  structural properties of Alg.~\ref{lp-alg} and its relation to an optimal offline scheduling algorithm. An advantage of our approach is that it provides a unified analysis framework for all proposed LP algorithms (including Algs.~\ref{lp-alg}, \ref{lp-alg-ml}, and \ref{lp-alg-non}) without the need to construct separate dual solutions for different scenarios.

Let $P(t) = \{i : T_i \leq t\}$ denote the set of virtual packets that have arrived by slot~$t$. The following two lemmas characterize properties of this set. Here, when a virtual packet satisfies the condition in Line~\ref{lp-alg:condition} and thus triggers the operation in Line~\ref{lp-alg:x}, we say that it \emph{activates}. For clarity and continuity, we move most detailed proofs of this paper to the appendices in the supplemental material.

\begin{lemma}\label{lemma:sum-x-bound}
For a fixed slot~$t$, after the virtual packets in $P(t)$ have activated $n$ times since slot~$t$, 
the value computed by Alg.~\ref{lp-alg} satisfies
\begin{align*}
	\sum_{\tau = T_i}^{\infty} x(\tau) \geq \frac{\left(1 + \frac{1}{C_M}\right)^n - 1}{\theta},
\end{align*}
for all $i \in P(t)$.
\end{lemma}

\begin{proof}
(Sketch) We prove by induction on $n$. See Appendix~\ref{appendix:lemma:sum-x-bound} for details. 
\end{proof}

Lemma~\ref{lemma:sum-x-bound} immediately implies the following result.

\begin{lemma}\label{lemma:max-iteration}
For a fixed slot~$t$, the virtual packets in $P(t)$ can activate at most $\lceil C_m \rceil$ times since slot~$t$.
\end{lemma}

\begin{proof}
Fix a slot~$t$. By Lemma~\ref{lemma:sum-x-bound} and the choice of 
$\theta=(1+(1/C_M))^{C_m}-1$ defined in Line~\ref{lp-alg:constant}, 
once the packets in $P(t)$ have activated $\lceil C_m \rceil$ times, we obtain
\begin{align*}
\sum_{\tau=T_i}^{\infty} x(\tau) 
&\;\geq\; 
\frac{\left(1 + \frac{1}{C_M}\right)^{\lceil C_m \rceil} - 1}
     {\left(1 + \frac{1}{C_M}\right)^{C_m}-1} \;\geq\; 1,
\end{align*}
for all $i \in P(t)$, which implies that the condition in Line~\ref{lp-alg:condition} no longer holds. 
Hence, the virtual packets in $P(t)$ can activate at most $\lceil C_m \rceil$ times.
\end{proof}

Leveraging Lemma~\ref{lemma:max-iteration}, we are now ready to analyze
the objective value in Eq.~\eqref{lp:objective} achieved by
Alg.~\ref{lp-alg} in the following theorem. The theorem also
characterizes the asymptotic behavior when the update cost scales
linearly with the energy consumption, i.e., $C(t) = C_u \mathcal{E}(t)$
for a constant unit cost $C_u$. The asymptotic regime $C_u \to \infty$
models scenarios with severely constrained  resources. In this regime,
the competitive ratio depends only on the ratio between the maximum and
minimum update cost, denoted by $R = C_M/C_m$. Let $\mathcal{E}_M$ and
$\mathcal{E}_m$ denote the maximum and minimum per-update energy
consumption, respectively. Then, when $C(t) = C_u \mathcal{E}(t)$, the
same ratio can also be written as $R = \mathcal{E}_M / \mathcal{E}_m$.

\begin{theorem} \label{theorem:competitive-ratio}
The objective value in Eq.~\eqref{lp:objective} computed by Alg.~\ref{lp-alg} at the end of slot~$T$ is bounded above by
\begin{align*}
	\left(1 + \frac{1}{C_m}\right)\left(1 + \frac{1}{(1 + \frac{1}{C_M})^{C_m} - 1}\right)  \mathrm{OPT}(\mathcal{I}),
\end{align*}
for all possible instances~$\mathcal{I}$. Moreover, as the unit  cost $C_u$ scales to infinity, the ratio  with respect to the optimum approaches $\frac{e^{1/R}}{e^{1/R} - 1}$.
\end{theorem}

\begin{proof}
(Sketch)
Fix an instance~$\mathcal{I}$. 
Suppose that an optimal offline scheduling algorithm clears the virtual queue in slots 
$t_1, \cdots, t_n$, performing a total of $n$ clearing operations. 
Let $t_0=0$ and $t_{n+1}=T$. 
We divide the timeline into $n+1$ periods, where period~$k$ consists of slots 
$t_{k-1}+1$ through $t_k$. 

Let $J^*(k)$ denote the cost incurred by an optimal offline scheduling algorithm in period~$k$.  Let $H^*(k)$ denote the holding cost incurred by the optimal offline scheduling
algorithm for all virtual packets arriving in period~$k$.
Consider a fixed $k\in \{1, \cdots, n\}$. Including the additional clearing cost in slot~$t_k$, we have $J^*(k)=H^*(k)+C(t_k)$.

Similarly, let $J(k)$ denote the increment of the objective value in Eq.~(\ref{lp:objective}) by Alg.~\ref{lp-alg},
according to the activations of all virtual packets that arrive in period~$k$. Note that one activation increases the objective value by 
\begin{align}
   & C(t)\left( \frac{1}{C_M} \sum_{\tau=T_i}^t x(\tau) + \frac{1}{\theta C_M} \right) 
     + \left( 1 - \sum_{\tau=T_i}^t x(\tau) \right) \nonumber \\
   \leq \;& 1 + \frac{1}{\theta} \quad \text{(since $C(t) \leq C_M$).}    \label{proof:eq:increment}
\end{align}

Next, we count the number of activations made by the virtual packets arriving in period~$k$. First,  $H^*(k)$ exactly counts the number of iterations of Line~\ref{lp-alg:for} 
from slot $t_{k-1}+1$ through slot $t_k-1$ in period~$k$ for the virtual packets  arriving  during this period.
Second, by Lemma~\ref{lemma:max-iteration}, the virtual packets arriving in period~$k$ 
can activate at most $\lceil C_m \rceil$ additional times from slot~$t_k$ onward. Hence, they can activate at most $H^*(k)+\lceil C_m \rceil$ times in total. 

Thus, we obtain
\begin{align}
    J(k)
      &\le \left(1+\frac{1}{\theta}\right) \big(H^*(k)+\lceil C_m\rceil\big) \label{eq:theorem:competitive-ratio-1}\\
      &\le \left(1+\frac{1}{\theta}\right)\frac{\lceil C_m\rceil}{C_m}\,\big(H^*(k)+C(t_k)\big) \nonumber\\
      &\le \left(1+\frac{1}{C_m}\right) \left(1+\frac{1}{\theta}\right) J^*(k). \nonumber
\end{align}
The inequality also holds for $k=n+1$. 
Thus, the  objective value computed by Alg.~\ref{lp-alg} satisfies
\[
  \sum_{k=1}^{n+1} J(k)
  \le \left(1+\frac{1}{C_m}\right) \left(1+\frac{1}{\theta}\right)
  \sum_{k=1}^{n+1} J^*(k).
\]
Substituting $\mathrm{OPT}(\mathcal{I}) = \sum_{k=1}^{n+1} J^*(k)$ and 
$\theta = (1+(1/C_M))^{C_m} - 1$ completes the proof. 
See Appendix~\ref{appendix:theorem:competitive-ratio} for details. 
\end{proof}

Next, we also use Lemma~\ref{lemma:max-iteration}  to analyze  the computational complexity of Alg.~\ref{lp-alg}, as stated in the following lemma. Here, we use $\Delta A_M$ to denote the maximum value of $\Delta A(t)$ for all possible $t$.

\begin{lemma}\label{lemma:max-num-packet}
At the end of any slot, at most $2\left\lceil \sqrt{\Delta A_M  C_m} \right\rceil$ virtual packets satisfy the condition in Line~\ref{lp-alg:condition} of Alg.~\ref{lp-alg}. 
\end{lemma}

\begin{proof}
See Appendix~\ref{appendix:lemma:max-num-packet} for details. 
\end{proof}

According to Lemma~\ref{lemma:max-num-packet}, at the beginning of slot~$t$, at most $2\left\lceil \sqrt{\Delta A_M  C_m} \right\rceil + \Delta A_M$ virtual packets may satisfy the condition in Line~\ref{lp-alg:condition}. Therefore,  Line~\ref{lp-alg:for} needs at most  $2\left\lceil \sqrt{\Delta A_M  C_m} \right\rceil + \Delta A_M$ iterations. The computational complexity of the online algorithm scales quadratically with the minimum update cost.

%	\begin{enumerate}
%		\item If the value of  $C(t)$ at the beginning of each slot~$t$ can be estimated (through channel estimation techniques), then we define $C_M(t)=\max_{\tau \leq t} C(\tau)$ and $C_m(t)=\min_{\tau \leq t} C(\tau)$ be the maximum and minimum update cost, respectively, up to the present slot~$t$. For each slot~$t$,  we  substitute  $C_M$ and $C_m$ used in  Alg.~\ref{online-alg} by $C_M(t)$ and $C_m(t)$ . Following the analysis to Alg.~\ref{lp-alg-ml}, we can get the same competitive ratio. 	

\subsection{Randomized online scheduling algorithm} \label{subsection:online-alg}

\begin{algorithm}[t]
	\SetAlgoLined 
	\SetKwFunction{Union}{Union}\SetKwFunction{FindCompress}{FindCompress} \SetKwInOut{Input}{input}\SetKwInOut{Output}{output}
	%	\Input{$c$.}
	%	\Output{$d(t)$.}
	%	
		\tcc{Initialize all variables  as follows:}
	$x_{\text{pre-sum}}, x_{\text{sum}},  x(t)\leftarrow 0$  for all  $t$\;  \label{online-alg:initial}

		$\theta \leftarrow (1+\frac{1}{C_M})^{C_m}-1$\; 
		
	Choose a random number $u \in [0,1)$ with the continuous uniform distribution\;  \label{online-alg:random}
	\tcc{At the beginning of slot~$t=1,  \cdots, T$, do as follows:}

		\ForEach{virtual packet  such that $T_i \leq t$
 \label{online-alg:for}}{	
          \If{$\sum_{\tau=T_i}^t x(\tau) < 1$}{
					
				$x(t)\leftarrow x(t)+ \frac{1}{C_M}\sum_{\tau=T_i}^{t} x(\tau)+\frac{1}{\theta C_M}$\;  	 \label{online-alg:x}			}
			
		}
		$x_{\text{pre-sum}} \leftarrow x_{\text{sum}}$\;   \label{online-alg:pre-sum}
		$x_{\text{sum}} \leftarrow x_{\text{sum}}+\min\{x(t),1\}$\;  \label{online-alg:sum}

	\uIf{$x_{\text{pre-sum}} \leq u < x_{\text{sum}}$\label{online-alg:condition}}{
	$d(t) \leftarrow 1$\;  \label{online-alg:tx}
	$u \leftarrow u+1$\;   \label{online-alg:u+1}
	} 
	\Else{
		$d(t) \leftarrow 0$\;  \label{online-alg:idle}
	}
    \label{online-alg:end}

	\caption{Randomized online scheduling  algorithm without ML.}
	\label{online-alg}
\end{algorithm}

Leveraging the fractional-to-probabilistic conversion technique proposed in~\cite[Chapter 12]{buchbinder2009design}, this section presents a randomized online scheduling algorithm in Alg.~\ref{online-alg}, which converts the fractional solution $x(t)$ generated by Alg.~\ref{lp-alg} into a probabilistic transmission decision.

Alg.~\ref{online-alg} adjusts the variable $x(t)$ in Line~\ref{online-alg:x} using the same rule as in Alg.~\ref{lp-alg}. In addition, we introduce two auxiliary variables: $x_{\text{pre-sum}}$ and $x_{\text{sum}}$. The variable $x_{\text{pre-sum}}$ records the cumulative sum of $\min\{x(t), 1\}$ up to slot~$t - 1$ (Line~\ref{online-alg:pre-sum}), while $x_{\text{sum}}$ records the cumulative sum up to slot~$t$ (Line~\ref{online-alg:sum}). In Line~\ref{online-alg:random}, Alg.~\ref{online-alg} selects a uniform random number $u \in [0, 1)$. Then, according to Lines~\ref{online-alg:condition} and~\ref{online-alg:u+1}, if there exists $k \in \mathbb{N}$ such that $u + k \in [x_{\text{pre-sum}}, x_{\text{sum}})$, then the device decides to transmit an update (Line~\ref{online-alg:tx}); otherwise, the device idles (Line~\ref{online-alg:idle}). The idea behind Alg.~\ref{online-alg} mirrors the classical technique for sampling from a distribution using its cumulative distribution function. In particular, by the uniform randomness of~$u$, the probability of transmitting an update (equivalently, clearing the virtual queue) in slot~$t$ is exactly $\min\{x(t), 1\}$, and the cumulative transmission probability by slot~$t$ is $\min\left\{\sum_{\tau = T_i}^{t} x(\tau), 1\right\}$.

% \begin{remark}
% A naive approach would be to independently transmit with probability $\min\{x(t),1\}$ in each slot~$t$, using the value from Alg.~\ref{lp-alg}. However, such an approach leads to a cumulative clearing probability for each virtual packet that is strictly less than $\min\{\sum_{\tau = T_i}^t x(\tau),1\}$. As a result, the expected total cost could exceed that of the primal solution, potentially violating the competitive ratio in Theorem~\ref{theorem:competitive-ratio}.
% \end{remark}

Because of the randomness of Alg.~\ref{online-alg}, we evaluate its performance in terms of the \emph{expected} competitive ratio.

\begin{theorem} \label{theorem:online-alg}
The expected competitive ratio of Alg.~\ref{online-alg} is
\[
\left(1 + \frac{1}{C_m}\right) \left(1 + \frac{1}{(1 + \frac{1}{C_M})^{C_m} - 1}\right).
\]
Moreover, as the unit  cost $C_u$ scales to infinity, the ratio approaches $\frac{e^{1/R}}{e^{1/R} - 1}$.
\end{theorem}

\begin{proof}
 Fix an instance~$\mathcal{I}$. Following  \cite{tseng2019online}, 
 we can show that the expected clearing cost in each slot~$t$ under Alg.~\ref{online-alg} is upper bounded by the value of $C(t)x(t)$ as computed by Alg.~\ref{lp-alg}. Similarly, the expected number of virtual packets present in slot~$t$ under Alg.~\ref{online-alg} is upper bounded by $z_i(t)$ as computed by Alg.~\ref{lp-alg}. Therefore, the expected total cost in Eq.~\eqref{eq:cost-2} incurred by Alg.~\ref{online-alg}  is bounded above by the  objective value in Eq.~(\ref{lp:objective}) computed by Alg.~\ref{lp-alg}. The result then follows directly from Theorem~\ref{theorem:competitive-ratio}.   
\end{proof}

Next, we show that the competitive ratio of Alg.~\ref{online-alg}  is optimal by establishing a matching lower bound as follows.

\begin{lemma}\label{lemma:converse}
No online algorithm can achieve a competitive ratio  smaller than $\frac{e^{1/R}}{e^{1/R} - 1}$.
\end{lemma}

\begin{proof}
(Sketch) Consider the initial age $A_0=1$, a fixed age increment sequence $\boldsymbol{\Delta A} = (0,  \cdots, 0)$, and a fixed update cost sequence $\mathbf{C} = (C_m, C_m R, C_m R, \cdots, C_m R)$. Only the operation duration $T$ is unknown to the device. The age increment sequence captures a versioned monitoring scenario, where the age remains unchanged until a new version is generated, as in the \emph{age of version} metric \cite{abolhassani2021fresh}. Despite this single source of unknown uncertainty, we show in Appendix~\ref{appendix:lemma:converse} that as $C_m \to \infty$, no online scheduling algorithm can achieve a competitive ratio smaller than $(e^{1/R})/(e^{1/R} - 1)$.
\end{proof}

Through the lower bound on the competitive ratio in Lemma~\ref{lemma:converse} and the matching achievability scheme proposed in Alg.~\ref{online-alg}, we establish that the optimal competitive ratio against an adversary that can jointly manipulate the operation duration, the age increment, and the update cost is $(e^{1/R})/(e^{1/R}-1)$. This  matches the result in the classic online TCP \cite{buchbinder2009design} for $R=1$ (without cost variation). Moreover, this ratio  is $\mathcal{O}(R)$ for large $R$, scaling linearly with the update cost range $R$, while is unaffected by all other sources of uncertainty.  Thus, when the cost fluctuates, it becomes fundamentally harder for any online algorithm to balance timeliness and update cost. Moreover, see Fig.~\ref{fig:cr} for the competitive ratio at  finite values of  $R$, which also appears approximately linear in $R$.

\begin{figure}[!t]
\centering
\includegraphics[width=.3\textwidth]{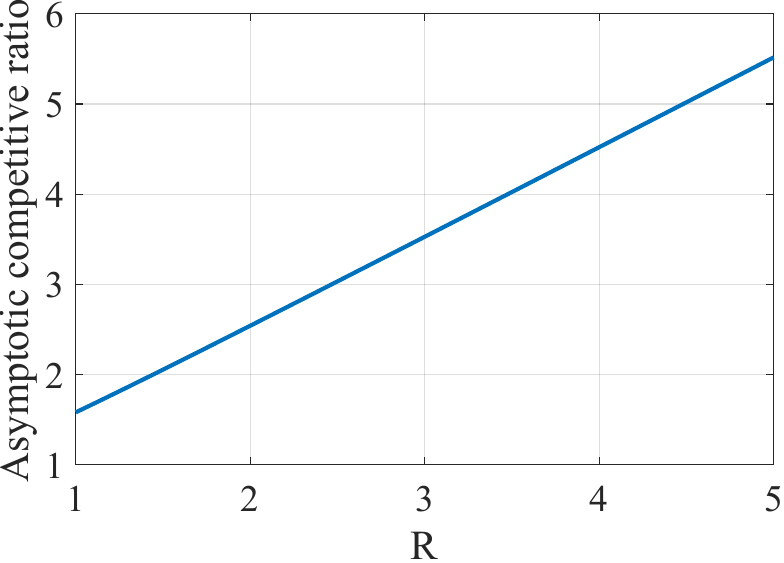}
\caption{Asymptotic competitive ratio $(e^{1/R})/(e^{1/R}-1)$ for finite values of $R$.}
\label{fig:cr}
\end{figure}

\begin{remark}\label{remark:cost}
This remark discusses how Alg.~\ref{lp-alg} can be adapted to scenarios in which
the bounds $C_M$ and $C_m$ on the update cost are not known in advance. In this
case, we propose to periodically estimate the update cost using channel estimation
techniques (e.g., see \cite{liu2025learning}). Let $T_{\text{est}}$ denote the estimation period and assume that the update cost $C(t)$ is measured in slots $t = 1,\, T_{\text{est}}+1,\, 2T_{\text{est}}+1,\, \cdots$. For each slot~$t$, by $C_M(t)=\max_{0\le n\le (t-1)/T_{\text{est}}} C(nT_{\text{est}}+1)$ and $C_m(t)=\min_{0\le n\le (t-1)/T_{\text{est}}} C(nT_{\text{est}}+1)$ we define the maximum and minimum observed costs up to slot~$t$, respectively. In Alg.~\ref{online-alg} (and the corresponding Alg.~\ref{lp-alg}), we replace
$C_M$ and $C_m$ by $C_M(t)$ and $C_m(t)$, respectively. This remains an online algorithm, as no future information is used. Let $\theta(t)= (1 + (1/(C_M(t)))^{C_m(t)} - 1$. Then, the increment of the objective value in Eq.~\eqref{proof:eq:increment} becomes
\begin{align*}
&C(t)\!\left( \frac{1}{C_M(t)} \sum_{\tau=T_i}^t x(\tau)
+ \frac{1}{\theta(t)\, C_M(t)} \right)
+ \left(1 - \sum_{\tau=T_i}^t x(\tau)\right) \\
&\qquad\le
\frac{\max\{ C(t), C_M(t)\}}{C_M(t)}
\left(1 + \frac{1}{\theta(t)}\right).
\end{align*}
Let $\Delta C_M = \max_t (C(t) - C(t-1))$ denote the maximum per-slot variation of the update cost. Since $C_M(t)$ is the maximum observed cost up to  slot~$\lfloor (t-1)/T_{\text{est}}\rfloor T_{\text{est}}+1$, and there are at most $T_{\text{est}}-1$ additional slots between slot~$\lfloor (t-1)/T_{\text{est}}\rfloor T_{\text{est}}+1$ and slot~$t$, we have
\[
\max\{C(t),C_M(t)\} \le C_M(t) + (T_{\text{est}}-1)\Delta C_M.
\]
Moreover, $C_M(t)\le C_M$ and $C_m(t)\ge C_m$, so $\theta(t)\ge\theta$, where
$\theta$ denotes the value used in the original algorithms. Thus, the
increment above is bounded by
\[
\frac{C_M(t) + (T_{\text{est}}-1)\Delta C_M}{C_M(t)}
\left(1 + \frac{1}{\theta}\right).
\]
If the device can estimate $C(t)$ at the beginning of each slot (i.e., $T_{\text{est}}=1$), then the factor above equals $1 + (1/\theta)$, achieving the same competitive ratio as in
Theorems~\ref{theorem:competitive-ratio} and~\ref{theorem:online-alg}. Otherwise, suppose that the
update cost also takes the form $C(t)=C_u \mathcal{E}(t)$. Let
$\Delta \mathcal{E}_M = \max_t (\mathcal{E}(t)-\mathcal{E}(t-1))$. Then, we can express the bound by
\[
\left(1+(T_{\text{est}}-1)\frac{\Delta \mathcal{E}_M}{\mathcal{E}_m}\right)
\left(1 + \frac{1}{\theta}\right).
\]
As $C_u\to\infty$, the modified online algorithm achieves the  asymptotic competitive ratio of $
\mathcal{O}(e^{1/R}/(e^{1/R}-1))$,
which matches the order of the original competitive ratio   (but with an inflated multiplicative pre-constant   $1+(T_{\text{est}}-1)(\Delta \mathcal{E}_M/ \mathcal{E}_m$). The same fix can apply to all remaining algorithms proposed later, yielding the same pre-constant.
\end{remark}

\section{Online scheduling algorithm design with ML}\label{section:ml}
This section extends the proposed online scheduling algorithm by incorporating ML that can suggest the next transmission time (i.e., to clear the virtual queue).  We focus on the online LP algorithm design as in Section~\ref{subsection:lp-alg}, since the resulting fractional solution can also be converted into a randomized scheduling algorithm as in Section~\ref{subsection:online-alg}.

\subsection{Online LP algorithm with ML}
This section extends the online LP algorithm in Alg.~\ref{lp-alg} by incorporating  ML advice~$\boldsymbol{\mathcal{M}}$ with unknown reliability, as described in Alg.~\ref{lp-alg-ml}. 
The key idea underlying Alg.~\ref{lp-alg-ml} is as follows. A new variable $T_{ML}$ is introduced in Line~\ref{lp-alg-ml:initial-ML} and is adjusted in Line~\ref{lp-alg-ml:initial-ML-update} whenever the ML advice suggests  clearing the virtual queue. Hence, the value of $T_{ML}$ represents the most recent slot when the ML advice recommended a clearing. Since the ML advice may be imperfect, the device does not blindly follow it. Instead, Alg.~\ref{lp-alg} modulates its response based on whether a virtual packet~$i$ has been suggested for clearing by the ML advice at the beginning of slot~$t$:
\begin{itemize}
 
      \item If the ML advice \emph{has already} recommended clearing the virtual packet~$i$ (checked via Line~\ref{lp-alg-ml:constant-cond2}), Alg.~\ref{lp-alg-ml} raises the clearing probability more aggressively by setting a smaller value for~$\theta$ in Line~\ref{lp-alg-ml:constant2}. We refer to this as a \emph{fast step}. 
      \item   Conversely, if the ML advice has \emph{not yet} recommended clearing the virtual packet~$i$ by slot~$t$ (checked via Line~\ref{lp-alg-ml:constant-cond}), Alg.~\ref{lp-alg-ml} raises the clearing probability more conservatively by setting a larger value for the constant $\theta$ in Line~\ref{lp-alg-ml:constant1}. We refer to this as a \emph{slow step}.
\end{itemize}
The adjustment of $\theta$ is governed by a hyperparameter \mbox{$\lambda \in (0,1]$}, which reflects the  device's level of trust in the ML advice. A smaller value of $\lambda$ corresponds to greater confidence in the ML advice and leads to closer alignment with it, whereas a larger value represents caution and yields more robust behavior. This trade-off between the consistency and the robustness with respect to $\lambda$ will be further discussed in Section~\ref{subsection:robust-consistency}.

%Furthermore, as implied by Lemma~\ref{lemma:sum-x-bound}, setting a smaller~$\theta$ results in more updates to the elements in $Y(t)$ before the condition in Line~\ref{lp-alg-ml:condition} becomes false. If we were to set $y_i(t) = 1$ as in Alg.~\ref{lp-alg}, the dual constraint in Eq.~\eqref{dual-program:const1} could become infeasible, even if $C_m$ is an integer. To avoid this, Alg.~\ref{lp-alg-ml} assigns a smaller value to $y_i(t)$, ensuring that the dual constraint remains approximately feasible.

\begin{algorithm}[t]
	\SetAlgoLined 
	\SetKwFunction{Union}{Union}\SetKwFunction{FindCompress}{FindCompress} \SetKwInOut{Input}{input}\SetKwInOut{Output}{output}
%	\Input{$\lambda$.}
%	\Output{$d(t)$.}
%	

	\tcc{Initialize all variables  as follows:}
	$x(t)$, $z_i(t)$  $\leftarrow 0$ for all  $i$ and  $t$\;  	\label{lp-alg-ml:initial}
	$T_{ML} \leftarrow 0$\;  	\label{lp-alg-ml:initial-ML}

	\tcc{At the beginning of slot~$t=1,  \cdots, T$, adjust all variables as follows:}

		\If{$\mathcal{M}(t)=1$}{
		$T_{ML} \leftarrow t$\;  \label{lp-alg-ml:initial-ML-update}

		}

		\ForEach{virtual packet $i$ such that $T_i \leq t$\label{lp-alg-ml:for1}}{	
				\If{$\sum_{\tau=T_i}^{t}x(\tau)<1$  \label{lp-alg-ml:condition}}{
				
						$z_{i}(t) \leftarrow 1- \sum_{\tau=T_i}^{t}x(\tau)$  \label{lp-alg-ml:z1}

				\uIf{$T_i < T_{ML}$ \label{lp-alg-ml:constant-cond2}}{
				 $\theta \leftarrow (1+\frac{1}{C_M})^{C_m \lambda}-1$\; 	\label{lp-alg-ml:constant2}

				}
				\Else{\label{lp-alg-ml:constant-cond}
				$\theta \leftarrow (1+\frac{1}{C_M})^{C_m / \lambda}-1$\; 	\label{lp-alg-ml:constant1}	
							
				}

				$x(t)\leftarrow x(t)+ \frac{1}{C_M}\sum_{\tau=T_i}^{t} x(\tau)+\frac{1}{\theta C_M}$\;  	 \label{lp-alg-ml:e1}	
				}
			}\label{lp-alg-ml:for1end}

	\caption{Online LP algorithm with ML}
	\label{lp-alg-ml}
\end{algorithm}

\subsection{Analysis of Alg.~\ref{lp-alg-ml}} \label{subsection:robust-consistency}
In this section, we establish the robustness and consistency of the proposed Alg.~\ref{lp-alg-ml}. To this end, we begin by analyzing the set $P(t)$ under Alg.~\ref{lp-alg-ml} in the following two lemmas, 
analogous to Lemmas~\ref{lemma:sum-x-bound} and~\ref{lemma:max-iteration}, respectively. Here, if a virtual packet activates and performs a slow or fast step, 
we say that it \emph{activates a slow or fast step}, respectively.   Moreover, we denote $\theta_s = (1 + (1/C_M))^{C_m / \lambda} - 1$ and $\theta_f = (1 +(1/C_M))^{C_m \lambda} - 1$.

\begin{lemma} \label{lemma:sum-x-bound-ml}
For a fixed slot~$t$, after the virtual packets in $P(t)$ have activated $N_s$ slow steps and $N_f$ fast steps since slot~$t$, 
the value computed by Alg.~\ref{lp-alg-ml} satisfies
\begin{align*}
\sum_{\tau=T_i}^{\infty} x(\tau) \geq & \,\frac{(1 + \frac{1}{C_M})^{N_s} - 1}{\theta_s}  \left(1 + \frac{1}{C_M} \right)^{N_f} \\
& + \frac{(1 + \frac{1}{C_M})^{N_f} - 1}{\theta_f},
\end{align*}
for all $i$ such that $T_i \leq t$.
\end{lemma}
\begin{proof}
(Sketch) We prove by induction on $N_f$. See Appendix~\ref{appendix:lemma:sum-x-bound-ml} for details. 
 	\end{proof}

\begin{lemma}\label{lemma:max-iteration-ml}
For a fixed slot~$t$, the virtual packets in $P(t)$ can activate $N_s$ slow steps and $N_f$ fast steps since slot~$t$, subject to the condition $N_s \lambda + N_f \leq   C_m+1$.
\end{lemma}

\begin{proof}
(Sketch) Using Lemma~\ref{lemma:sum-x-bound-ml},  we show that when $N_s \lambda + N_f \geq C_m$, then $\sum_{\tau=T_i}^{\infty} x(\tau) \geq 1$ for all $i \in P(t)$. See Appendix~\ref{appendix:lemma:max-iteration-ml} for details. 
\end{proof}

Leveraging Lemma~\ref{lemma:max-iteration-ml}, we are ready to  analyze the  objective value in Eq.~\eqref{lp:objective} computed by Alg.~\ref{lp-alg-ml}. The next two theorems analyze its robustness and consistency, respectively. 

\begin{theorem} \label{theorem:robustness}
The  objective value in Eq.~\eqref{lp:objective} computed by Alg.~\ref{lp-alg-ml} at the end of slot~$T$ is bounded above by
\[
\left(1 + \frac{1}{C_m} \right) \left( 1 + \frac{1}{(1 + \frac{1}{C_M})^{C_m \lambda} - 1} \right)  \mathrm{OPT}(\mathcal{I}),
\]
for all possible  instances~$\mathcal{I}$ and ML advice~$\boldsymbol{\mathcal{M}}$. Moreover, as the unit  cost $C_u$ scales to infinity,  the ratio  with respect to the optimum approaches $\frac{e^{\lambda / R}}{e^{\lambda / R} - 1}$.
\end{theorem}

\begin{proof}
(Sketch)  We follow the proof of Theorem~\ref{theorem:competitive-ratio} and show that $J(k) \leq (1 + (1/C_m))(1 + (1/\theta_f)) J^*(k)$ for all $k$, under Alg.~\ref{lp-alg-ml}.
Then, applying the same reasoning as in the proof of Theorem~\ref{theorem:competitive-ratio} 
and substituting the definition of $\theta_f$ completes the proof.
 See  Appendix~\ref{appendix:theorem:robustness} for details. 
\end{proof}

\begin{theorem}\label{theorem:consistency}
The  objective value in Eq.~\eqref{lp:objective} computed by Alg.~\ref{lp-alg-ml} at the end of slot~$T$ is bounded above by
\begin{align*}
&\max\left\{
1 + \frac{1}{\left(1 + \frac{1}{C_M} \right)^{C_m/ \lambda} - 1},  \right.\\
&\quad \quad\left.\frac{\lceil C_m \lambda \rceil}{C_m} \cdot \left(1 + \frac{1}{\left(1 + \frac{1}{C_M} \right)^{C_m \lambda} - 1} \right)
\right\}
\cdot J(\mathcal{I}, \boldsymbol{\mathcal{M}})
\end{align*}
for all possible instances~$\mathcal{I}$ and ML advice~$\boldsymbol{\mathcal{M}}$. Moreover, as the unit update cost $C_u \to \infty$, the  ratio with respect to the optimum converges to $\frac{\lambda e^{ \lambda/R}}{e^{ \lambda/R} - 1}$.
\end{theorem}

\begin{proof}
(Sketch) Fix an instance $\mathcal{I}$ and  ML advice~$\boldsymbol{\mathcal{M}}$. 
We follow the proof of Theorem~\ref{theorem:competitive-ratio}, with minor modifications. 
Redefine $t_k$ for $k \in \{1, \cdots, n\}$ as the slot when $\boldsymbol{\mathcal{M}}$ clears the virtual queue 
for the $k$-th time. These redefined time points determine a new set of periods, replacing those used in the proof of Theorem~\ref{theorem:competitive-ratio}.

Let $J_{\boldsymbol{\mathcal{M}}}(k)$ denote the cost incurred by $\boldsymbol{\mathcal{M}}$ in period~$k$. 
%Then, the total cost in Eq.~(\ref{eq:cost-2}) incurred by $\boldsymbol{\mathcal{M}}$ is 
%$\sum_{k=1}^{n+1} J_{\boldsymbol{\mathcal{M}}}(k)$. 
Let $J(k)$ be the increment of the objective value in Eq.~(\ref{lp:objective}) by Alg.~\ref{lp-alg-ml}, according to the slow and fast step 
activations of all virtual packets arriving in period~$k$. We show that 
\begin{align*}
J(k) \leq  \max \left\{1+\frac{1}{\theta_s}, \frac{\lceil C_m \lambda \rceil}{C_m}\left(1+\frac{1}{\theta_f}\right)\right\} J_{\boldsymbol{\mathcal{M}}}(k),
\end{align*}
for all $k$. Thus, the total objective value computed by Alg.~\ref{lp-alg-ml} satisfies
\[
\sum_{k=1}^{n+1} J(k)\;\leq\;\max \left\{1+\frac{1}{\theta_s}, \frac{\lceil C_m \lambda \rceil}{C_m}\left(1+\frac{1}{\theta_f}\right)\right\} 
\sum_{k=1}^{n+1} J_{\boldsymbol{\mathcal{M}}}(k).
\]
Substituting $J(\mathcal{I}, \boldsymbol{\mathcal{M}})=\sum_{k=1}^{n+1} J_{\boldsymbol{\mathcal{M}}}(k)$ 
and the definitions of $\theta_s$ and $\theta_f$ proves the theorem. 
See Appendix~\ref{appendix:theorem:consistency} for details.
\end{proof}

 Next, we show that the results in the above two theorems characterize the optimal consistency-robustness trade-off. 
 
\begin{lemma}\label{lemma:converse-ml}
A $\frac{\lambda e^{ \lambda/R}}{e^{ \lambda/R}-1}$-consistency scheduling algorithm has a robustness  of at least $\frac{e^{ \lambda/R}}{e^{ \lambda/R}-1}$.
\end{lemma}
\begin{proof}
(Sketch)
Using the same  instance as in the proof of Lemma~\ref{lemma:converse}, we can establish the result. See Appendix~\ref{appendix:lemma:converse-ml} for details. 
\end{proof}

Combining the lower bound in Lemma~\ref{lemma:converse-ml} with the matching achievability scheme in Alg.~\ref{lp-alg-ml}, we establish that the optimal consistency–robustness trade-off is characterized by the pair of $(\lambda e^{\lambda/R})/(e^{\lambda/R}-1)$ and $(e^{\lambda/R})/(e^{\lambda/R}-1)$. This  matches the result for the classic online TCP  with ML \cite{bamas2020primal} when $R=1$ (without cost variation). 

Note that in many prior studies on ML-augmented online algorithms
(e.g., \cite{bamas2020primal, liu2025learning}), the consistency approaches~1 as the trust parameter $\lambda \to 0$, i.e., setting $\lambda \to 0$ forces the algorithm to rely entirely on the ML advice. In contrast, in our setting the consistency approaches $R$ as $\lambda \to 0$, which exceeds~1 whenever $R > 1$. This difference is explained as follows. The  robustness becomes unbounded as $\lambda \to 0$. Because this trade-off is optimal, the robustness is infinite whenever the consistency falls below $R$. This implies that robustness guarantees collapse even when an online algorithm follows  ML advice only partially (so that
its consistency remains below $R$). Thus, in our setting, taking $\lambda \to 0$ does not force the
algorithm to rely fully on the ML advice. Instead, it identifies the limit in which the algorithm becomes as consistent with the ML advice as possible while still maintaining a finite robustness guarantee.

Thus far, we have characterized the optimal consistency–robustness trade-off as the trust level in the ML advice varies. Tuning the trust level leads to different competitive ratios. We next discuss how to determine an optimal trust level that minimizes the competitive ratio when the value of $R$ is large.  
For large $R$, the consistency \mbox{$(\lambda e^{\lambda/R})/(e^{\lambda/R}-1)$} varies from \mbox{$(e^{1/R})/(e^{1/R}-1) \approx R$} at $\lambda=1$ to $R$ as \mbox{$\lambda \rightarrow 0$}. Hence, the consistency is nearly identical for all $\lambda \in (0,1]$. In contrast, for large $R$, the robustness $(e^{\lambda/R})/(e^{\lambda/R}-1)=\mathcal{O}(R/\lambda)$ degrades as $\lambda$ decreases.  Then, considering all possible $\lambda \in (0,1]$ (representing partial or no trust in  the ML advice), an optimal online scheduling algorithm $\pi^*$ that minimizes the competitive ratio satisfies the following performance guarantees:
\begin{align*}
J(\mathcal{I}, \pi^*) \leq \min\big\{
R \cdot J(\mathcal{I}, \boldsymbol{\mathcal{M}}),\, \frac{R}{\lambda}\cdot \mathrm{OPT}(\mathcal{I}) \big\},
\end{align*}
for all possible $\mathcal{I}$, $\boldsymbol{\mathcal{M}}$, and all  $\lambda \in (0, 1]$. The minimum is attained at $\lambda = 1$, leading  $J(\mathcal{I}, \pi^*) \leq   R\cdot \mathrm{OPT}(\mathcal{I})$.
In addition, considering full trust in the ML advice, we also have
$J(\mathcal{I}, \pi^*) \leq J(\mathcal{I}, \boldsymbol{\mathcal{M}})$. 
Suppose that the ML advice $\boldsymbol{\mathcal{M}}$ satisfies the following reliability guarantee:
\mbox{$\frac{J(\mathcal{I}, \boldsymbol{\mathcal{M}})}{OPT(\mathcal{I})}
\leq \zeta(\boldsymbol{\mathcal{M}})$} for all possible~$\mathcal{I}$. Then, we have
\begin{align*}
J(\mathcal{I}, \pi^*)
\leq
\min \big\{ R,\, \zeta(\boldsymbol{\mathcal{M}}) \big\}
\cdot OPT(\mathcal{I}),		
\end{align*}
for all possible $\mathcal{I}$ and $\boldsymbol{\mathcal{M}}$.
That is, for large $R$, regardless of the ML reliability
$\zeta(\boldsymbol{\mathcal{M}})$, the optimal response to ML advice (for minimizing the competitive ratio) is threshold-like: the algorithm should either fully trust the ML advice if \mbox{$\zeta(\boldsymbol{\mathcal{M}}) \leq R$} or completely ignore it otherwise. 

Moreover, see Fig.~\ref{fig:c_r_tradeoff} for the trade-off at finite values of $R$. Here, we also observe a dramatic degradation in robustness resulting from even a small
change in consistency. This indicates that the threshold structure nearly holds as well.

\begin{figure}[!t]
    \centering
    \subfloat[$R=1$]{
        \includegraphics[width=.16\textwidth]{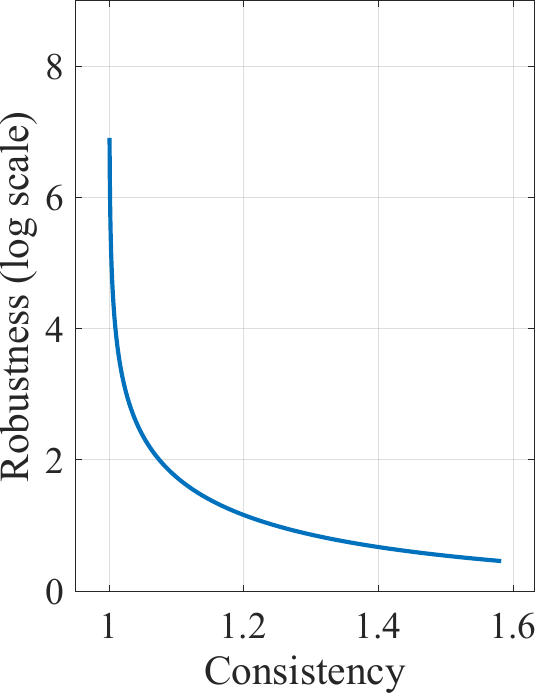}} 
            \subfloat[$R=2$]{
        \includegraphics[width=.16\textwidth]{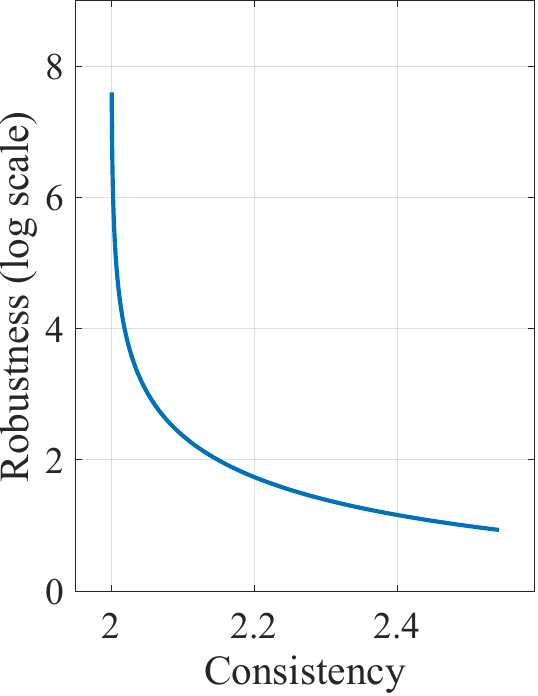}} 
            \subfloat[$R=3$]{
        \includegraphics[width=.16\textwidth]{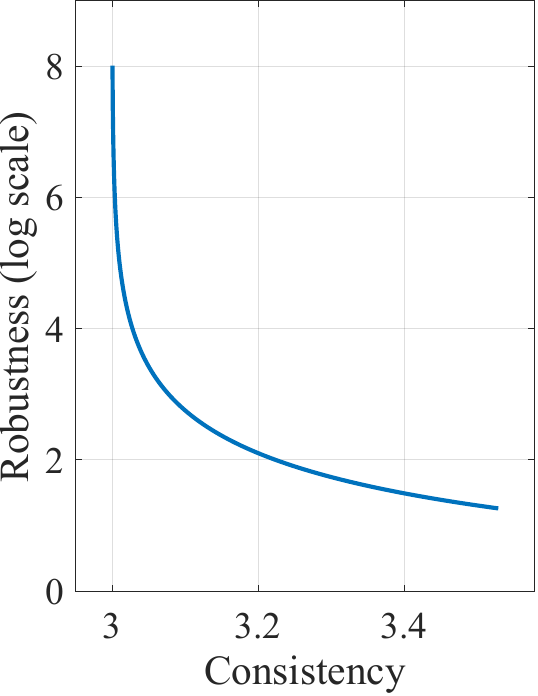}}
\caption{Robustness $(e^{ \lambda/R})/(e^{ \lambda/R} - 1)$ (in log scale) versus consistency $(\lambda e^{ \lambda/R})/(e^{ \lambda/R} - 1)$.}
    \label{fig:c_r_tradeoff}
\end{figure}

\section{Intermittent update opportunities}\label{section:non-saturated}
We extend our framework to scenarios where the device may be unable to update in certain slots (e.g., when no update is generated or when the device cannot transmit). 
Let $U(t)$ indicate whether the device is able to update in slot~$t$, where $U(t) = 1$ if it can  and $U(t) = 0$ otherwise. 
%We adopt a \emph{bufferless} model: if an update packet is not transmitted in its arrival slot, it is discarded. Such a model has been shown to provide similar average AoI performance as buffered systems in many settings (e.g., see~\cite{.}), since transmitting stale packets typically contributes little to timeliness improvement.
Let $\mathbf{U} = (U(1), \cdots, U(T))$ denote the update opportunity sequence. 
We then redefine the uncertainty instance as $\mathcal{I} = \{T, \boldsymbol{\Delta A}, \mathbf{C}, \mathbf{U}\}$ to incorporate this additional source of uncertainty.

To model this, we augment the virtual queueing system described in Section~\ref{subsection:virtual-queue} with a virtual ON/OFF channel. Specifically, if $U(t) = 1$, the virtual channel is ON and the virtual server is allowed to clear virtual packets; if $U(t) = 0$, the virtual channel is OFF and the virtual server must idle.
This leads to the following revised LP:
\begin{subequations}	\label{lp-non}
	\begin{align}
	\min_{x(t), z_i(t)} & \hspace{.2cm}\sum_{t=1}^T \left(C(t)  x(t) + \sum_{i:T_i \leq t} z_{i}(t)\right) \label{lp-non:objective}\\
	\text{s.t.} & \hspace{.2cm} z_i(t) + \sum_{\tau = T_i}^{t} U(\tau) x(\tau) \geq 1,\nonumber\\
	& 	 \text{\hspace{1cm}\text{for all $i$ such that $T_i \leq t$ and for all $t$};} 	 \label{lp-non:const1}\\
	&\hspace{.2cm} x(t), z_i(t) \geq 0 \quad \text{for all $i$ and $t$}. \label{lp-non:const2}
	\end{align}
\end{subequations}
Here, Eq.~(\ref{lp-non:const1}) differs from Eq.~(\ref{lp:const1}) because a virtual packet is cleared only when the virtual channel is ON.

Next, we generalize Algs.~\ref{lp-alg} (without ML) and~\ref{lp-alg-ml} (with ML) in Sections~\ref{subsection:non-saturated} and~\ref{subsection:non-saturated-ml}, respectively, to handle scenarios with intermittent update opportunities.

\subsection{Without ML advice} \label{subsection:non-saturated}

This section extends Alg.~\ref{lp-alg}, as described in Alg.~\ref{lp-alg-non}. 
The key design change is that $x(t)$ is adjusted only when the virtual channel is ON, i.e., when $U(t) = 1$ (Line~\ref{lp-alg-non:has-update}). Furthermore, unlike Alg.~\ref{lp-alg}, which  adjusts $x(t)$ only for the current slot~$t$,  Alg.~\ref{lp-alg-non} also considers all prior virtual OFF slots that occurred since the previous virtual ON slot. Concretely, for each such prior virtual OFF slot (Line~\ref{lp-alg-non:all-update}), if the constraint in Eq.~(\ref{lp-non:const1}) still holds (Line~\ref{lp-alg-non:for2}), 
the algorithm keeps increasing $x(t)$ (Line~\ref{lp-alg-non:x}).  This reflects the intuition that virtual packets held longer in the queue (due to consecutive virtual OFF periods)  should have higher clearing probabilities once the virtual channel becomes ON. To implement this logic, the algorithm maintains a pointer $\hat{t}$ (Line~\ref{lp-alg-non:all-update}) 
to denote the starting slot of this multiple increment procedure. 
This pointer is adjusted in Line~\ref{lp-alg-non:ts} when the condition in Line~\ref{lp-alg-non:ts-condition1} holds. The pointer identifies the slot immediately following the previous virtual ON slot, which is either the current slot (if the current virtual channel is ON) or the start of the current virtual OFF period (otherwise).

Note that if a virtual packet arrives in a virtual OFF slot, it must remain in the virtual queue until the next virtual ON slot. 
This limitation applies to all scheduling algorithms (including an offline optimal algorithm). 
Thus, the multiple increment mechanism applies only to those virtual packets that arrived before the previous virtual ON slot (Line~\ref{lp-alg-non:ts-condition}). 
For virtual packets that arrive after the previous virtual ON slot, Alg.~\ref{lp-alg-non} adjusts $x(t)$ only once (Line~\ref{lp-alg-non:else:x}).

\begin{algorithm}[t]
	\SetAlgoLined 
	\SetKwFunction{Union}{Union}\SetKwFunction{FindCompress}{FindCompress} \SetKwInOut{Input}{input}\SetKwInOut{Output}{output}
%	\Input{$c$.}
%	\Output{$d(t)$.}
%	

	\tcc{Initialize all variables  as follows:}
	$x(t)$, $z_i(t)$  $\leftarrow 0$ for all  $i$ and  $t$\;  \label{lp-alg-non:initial}
 
 						$\theta \leftarrow (1+\frac{1}{C_M})^{C_m}-1$\; 	\label{lp-alg-non:constant}		
 $\hat{t} \leftarrow 1$\; 
		
	\tcc{At the beginning of slot~$t=1,  \cdots, T$, adjust all variables as follows:}
	
		\If{$t>1$ and $U(t-1) = 1$ \label{lp-alg-non:ts-condition1}}{
    $\hat{t} \leftarrow t$\;  \label{lp-alg-non:ts}
}

		\ForEach{virtual packet $i$ such that $T_i \leq t$}{		
			\If{$\sum_{\tau=T_i}^{t}x(\tau)<1$\label{lp-alg-non:for}}{
				$z_{i}(t) \leftarrow 1- \sum_{\tau=T_i}^{t}x(\tau)$\;  \label{lp-alg-non:z}
				
				\If{$U(t)=1$\label{lp-alg-non:has-update}}{	
				\uIf{$T_i < \hat{t}$ \label{lp-alg-non:ts-condition}}{
					\For{$t'=t$ \KwDownTo $\hat{t}$ \label{lp-alg-non:all-update}}{
						\If{$\sum_{\tau=T_i}^{t}x(\tau)<1$ \label{lp-alg-non:for2}}{
							$x(t)\leftarrow x(t)+ \frac{1}{C_M}\sum_{\tau=T_i}^{t} x(\tau)	
							  +\frac{1}{\theta C_M}$\;  	 \label{lp-alg-non:x}
						}
						}
					}
					\Else{
					\If{$\sum_{\tau=T_i}^{t}x(\tau)<1$ \label{lp-alg-non:else}}{
							$x(t)\leftarrow x(t)+ \frac{1}{C_M}\sum_{\tau=T_i}^{t} x(\tau)	+\frac{1}{\theta C_M}$\;  	 \label{lp-alg-non:else:x}
						}
						}
				}
				}
				
			}\label{lp-alg-non:forend}

	\caption{Online LP algorithm without ML for intermittent update opportunities.}
	\label{lp-alg-non}
\end{algorithm}

We next analyze the performance of Alg.~\ref{lp-alg-non} and show that it can  achieve the same asymptotic competitive ratio as stated in Theorem~\ref{theorem:competitive-ratio}. To that end, we  present a lemma that bounds the increment of $\sum_{\tau = T_i}^{\infty} x(\tau)$ under the multiple increment mechanism in Alg.~\ref{lp-alg-non}.  Here, when a virtual packet satisfies the condition in Line~\ref{lp-alg-non:for2} or~\ref{lp-alg-non:else} and thus triggers the operation in Line~\ref{lp-alg-non:x} or~\ref{lp-alg-non:else:x}, we say that it activates.

\begin{lemma}\label{lemma:x-upper-bound}
For a fixed slot~$t$, after the virtual packets in $P(t)$ have activated $n$ times since slot~$t$, 
the value of $\sum_{\tau=T_i}^{\infty} x(\tau)$ computed by Alg.~\ref{lp-alg-non} 
increases (relative to the beginning of slot~$t$) by at most
\begin{align*}
\left(1+\frac{1}{\theta}\right)\left[\left(1+\frac{1}{C_M}\right)^n-1\right],	
\end{align*}
for all $i \in P(t)$.
\end{lemma}

\begin{proof}
(Sketch) We prove by induction on $n$. See Appendix~\ref{appendix:lemma:x-upper-bound} for details. 
\end{proof}

Using Lemma~\ref{lemma:x-upper-bound}, we are ready to analyze Alg.~\ref{lp-alg-non} in the next result.  Here, let $T_{\text{OFF}}$ denote the maximum number of consecutive virtual  OFF slots up to and including the next virtual ON slot.

\begin{theorem} \label{theorem:competitive-ratio-non}
The  objective value in Eq.~\eqref{lp-non:objective} computed by Alg.~\ref{lp-alg-non} at the end of slot~$T$ is bounded above by
\begin{align*}
&\left[ \left(1 + \frac{1}{C_m} \right)^{1+2 \lceil \sqrt{\Delta A_M  C_m} \rceil T_{\text{OFF}}} \left(1 + \frac{1}{(1 + \frac{1}{C_M})^{C_m} - 1}\right)\right.\\
&\quad \left.+ \frac{2 \lceil \sqrt{\Delta A_M  C_m} \rceil T_{\text{OFF}}}{C_m} \right] \mathrm{OPT}(\mathcal{I}),	
\end{align*}
for all possible  instances~$\mathcal{I}$. Moreover, if $\Delta A_M$ and $T_{\text{OFF}}$ is finite, then as the unit  cost $C_u$ scales to infinity, the ratio  with respect to the optimum approaches $\frac{e^{1/R}}{e^{1/R} - 1}$.
\end{theorem}

\begin{proof}
(Sketch)  We follow the proof of Theorem~\ref{theorem:competitive-ratio}.  Fix a period $k\in\{1,\dots,n\}$ and a virtual ON slot $t$ in that period. We bound the increment of the objective value in Eq.~\eqref{lp-non:objective} in slot~$t$  by considering how  Alg.~\ref{lp-alg-non} adjusts $x(t)$ and the matching $z$-variables, where we use the notation $\sum_{\tau=0}^{t}x(\tau) \big|_{\text{condition}}$ to represent the value of $\sum_{\tau=T_i}^{t}$ under the specific condition:
\begin{enumerate}
	\item If a virtual packet $i$  arrives before $\hat{t}$ and activates in iteration $t'$ of Line~\ref{lp-alg-non:all-update} in slot~$t$, then Line~\ref{lp-alg-non:x} increases $x(t)$ by
\begin{align*}
\frac{1}{C_M}\left(\sum_{\tau=T_i}^t x(\tau)\bigg|_{\text{before the activation}}\right) + \frac{1}{\theta C_M}.	
\end{align*}
However, the paired $z_i(t')$ was already set to be \mbox{$1-\sum_{\tau=T_i}^{t'}x(\tau)$}  in a slot \mbox{$t' \leq t$}. Because $x(\tau)$ does not change (for all possible $\tau$) over the virtual OFF period until slot~$t$, we have 
\begin{align*}
z_i(t')=1-\left(\sum_{\tau=T_i}^t x(\tau)\bigg|_{\text{start of slot~$t$}}\right).	
\end{align*}

Hence, the increment of the objective value  due to the adjustment of $x(t)$ from virtual packet~$i$ in iteration~$t'$ and the paired $z_i(t')$ is
\begin{align*}
&C(t)\left[\frac{1}{C_M}\left(\sum_{\tau=T_i}^t x(\tau)\Big|_{\text{before the activation}}\right) + \frac{1}{\theta C_M}\right]\nonumber\\
&\quad+ 1-\left(\sum_{\tau=T_i}^t x(\tau)\bigg|_{\text{start of slot~$t$}}\right)\nonumber\\
&\le 1+\frac{1}{\theta} \nonumber\\
&+ \left(\sum_{\tau=T_i}^t x(\tau)\bigg|_{\text{before the activation}}\right) - \left(\sum_{\tau=T_i}^t x(\tau)\bigg|_{\text{start of slot~$t$}}\right).
\end{align*}

By Lemma~\ref{lemma:max-num-packet}, 
we can show that there are at most $2 \left\lceil \sqrt{\Delta A_M C_m} \right\rceil T_{\text{OFF}}$ activations from the start of slot~$t$ until the considered activation. By Lemma~\ref{lemma:x-upper-bound}, we have
\begin{align*}
&\left(\sum_{\tau=T_i}^t x(\tau)\big|_{\text{before the activation}}\right) - \left(\sum_{\tau=T_i}^t x(\tau)\bigg|_{\text{start of slot~$t$}}\right)\\
& \le\left(1+\frac{1}{\theta}\right)\left[\left(1+\frac{1}{C_M}\right)^{2 \left\lceil \sqrt{\Delta A_M C_m} \right\rceil T_{\text{OFF}}}-1\right],	
\end{align*}
which gives the  bound on the increment of the objective value:
\begin{align}
\left(1+\frac{1}{\theta}\right)\left(1+\frac{1}{C_M}\right)^{2\lceil \sqrt{\Delta A_M C_m}\rceil T_{\text{OFF}}}. \label{eq:case1-final}
\end{align}

\item  If a virtual packet $i$  arrives before $\hat{t}$, meets the condition in Line~\ref{lp-alg-non:for} in slot~$t' \leq t$, but does not activate in iteration~$t'$ of Line~\ref{lp-alg-non:all-update} in slot~$t$, then because the paired $z_i(t')$ was still set to $1 - \sum_{\tau=T_i}^{t'} x(\tau)$ in slot $t'$, it also contributes at most $1$ to the objective value in that iteration.

\item If a virtual packet $i$  arrives in slot~$t$ and activates in slot $t$, then by the proof of Theorem~\ref{theorem:competitive-ratio} the increment from $x(t)$ in Line~\ref{lp-alg-non:else:x} and its paired $z_i(t)$ in Line~\ref{lp-alg-non:z} is $1+(1/\theta)$,  which is also bounded above by Eq.~\eqref{eq:case1-final}.
\end{enumerate}

Next, we show that, for the virtual packets arriving in the period, there are at most $\lceil C_m\rceil$ Case~1 and Case~3 increments since slot~$t_k$, and at most $2\lceil \sqrt{\Delta A_M C_m}\rceil T_{\text{OFF}}$ Case~2 increments since slot~$t_k$. Then, we can  rewrite Eq.~(\ref{eq:theorem:competitive-ratio-1}) as
\begin{align*}
J(k)
&\le \left(1+\frac{1}{\theta}\right)
\left(1+\frac{1}{C_m}\right)^{2 \lceil \sqrt{\Delta A_M  C_m} \rceil T_{\text{OFF}}}
\left(H^*(k)+ \lceil C_m \rceil\right)\\
&\quad + 1 \cdot \left(2 \lceil \sqrt{\Delta A_M  C_m} \rceil T_{\text{OFF}}\right).
\end{align*}
Finally, following the proof of Theorem~\ref{theorem:competitive-ratio} yields the desired result. See Appendix~\ref{appendix:theorem:competitive-ratio-non} for details. 
\end{proof}

Compared with Theorem~\ref{theorem:competitive-ratio}, 
the competitive ratio in Theorem~\ref{theorem:competitive-ratio-non} 
scales that in Theorem~\ref{theorem:competitive-ratio} 
by a factor of $(1+(1/C_m))^{2 \lceil \sqrt{\Delta A_M  C_m} \rceil T_{\text{OFF}}}$ 
(which approaches~1 as $C_u \to \infty$), 
and includes an additional term $(2 \lceil \sqrt{\Delta A_M  C_m} \rceil T_{\text{OFF}})/C_m$ 
(which also approaches~0 as $C_u \to \infty$). 
Thus, even under intermittent update opportunities, 
Alg.~\ref{lp-alg-non} asymptotically achieves the lower bound in Lemma~\ref{lemma:converse}.

Next, we show that if the adversary is so powerful that it can also set  $\Delta A_M$ or $T_{\text{OFF}}$ arbitrarily large, then no competitive ratio can be guaranteed.

\begin{lemma} \label{lemma:inf-t-c}
If either $\Delta A_M$ or $T_{\text{OFF}}$ is unbounded, then no online algorithm can achieve a finite competitive ratio as $C_m \to \infty$.
\end{lemma}

\begin{proof}
See Appendix~\ref{appendix:lemma:inf-t-c} for details.  
\end{proof}

\subsection{With ML advice}  \label{subsection:non-saturated-ml}
This section further incorporates ML advice. 
To this end, we modify Alg.~\ref{lp-alg-ml} by applying the multiple increment mechanism 
for  previous virtual OFF slots, as in Alg.~\ref{lp-alg-non}. 
By extending the proof of Theorem~\ref{theorem:competitive-ratio-non} 
to revise those of Theorems~\ref{theorem:robustness} and~\ref{theorem:consistency}, 
we can obtain the same scaling and additional terms (as in Theorem~\ref{theorem:competitive-ratio-non}), 
yielding the same asymptotic results as in Theorems~\ref{theorem:robustness} and~\ref{theorem:consistency}.

%&\left[ \left(1 + \frac{1}{C_m} \right)^{1+2 \lceil \sqrt{\Delta A_M  C_m} \rceil T_{\text{OFF}}} \left(1 + \frac{1}{(1 + \frac{1}{C_M})^{C_m} - 1}\right)\right.\\
%&\quad \left.+ \frac{2 \lceil \sqrt{\Delta A_M  C_m} \rceil T_{\text{OFF}}}{C_m} \right] \mathrm{OPT}(\mathcal{I}),

\section{Numerical studies}\label{section:sim}

The previous sections established that our proposed algorithms achieve the best
possible competitiveness and the optimal consistency–robustness trade-off in adversarial environments. In this section, we complement the theoretical
results by evaluating the algorithms in stationary stochastic environments
through numerical experiments.

We adopt a setting similar to \cite{hsu2019scheduling}, which derived optimal offline scheduling policies under stationary assumptions. We simulate a horizon of $T=10000$ slots. Update opportunities follow a Bernoulli process with rate~$0.7$. To model update costs with memory, we use a two-state Markov chain (as in the Gilbert–Elliott model) with states $\mathsf{L}$ (low cost) and $\mathsf{H}$ (high cost). Let $S(t) \in \{\mathsf{L},
\mathsf{H}\}$ denote the state in slot~$t$, and let $\mathrm{tr}_p$ and
$\mathrm{tr}_q$ be the transition probabilities from $\mathsf{L}$ to
$\mathsf{H}$ and from $\mathsf{H}$ to $\mathsf{L}$, respectively. 

Following \cite{hsu2019scheduling}, an optimal offline scheduling policy in the stationary setting can be characterized by two age thresholds: if $S(t)=\mathsf{H}$, the device transmits when the age reaches a threshold $T_{\mathsf{H}}$; if $S(t)=\mathsf{L}$, it  does when the age reaches a threshold $T_{\mathsf{L}}$. We compute the optimal pair of $T_{\mathsf{H}}$ and  $T_{\mathsf{L}}$ via exhaustive search to minimize the total cost in Eq.~\eqref{eq:cost}.

Next, we consider both a linear aging function in Section~\ref{subsection:sim-linear} and a nonlinear one in Section~\ref{subsection:sim-nonlinear}.

\subsection{Linear aging function} \label{subsection:sim-linear}
In this section, we consider a constant age increment process with $\Delta A(t)=1$ for all $t$.  We  evaluate the proposed online algorithm without ML in Section~\ref{subsection:sim1} and the ML-augmented version in Section~\ref{subsection:sim2}.

\subsubsection{Online scheduling algorithm without ML}\label{subsection:sim1}
In this section, we validate  Alg.~\ref{online-alg} (modified according to Alg.~\ref{lp-alg-non} to handle intermittent update opportunities). Figs.~\ref{fig:linear-online-0.2} and~\ref{fig:linear-online-0.8}  show the time-average cost (y-axis) for various values of $C_m$ (x-axis). In Fig.~\ref{fig:linear-online-0.2}. we set   $\mathrm{tr}_p=0.2$, $\mathrm{tr}_q=0.8$, resulting in longer stays in state~$\mathsf{L}$; in Fig.~\ref{fig:linear-online-0.8}, we set  $\mathrm{tr}_p=0.8$, $\mathrm{tr}_q=0.2$, resulting in longer stays in
state~$\mathsf{H}$.
Each figure has three
subfigures corresponding to $R=1$, $R=3$, and $R=5$, respectively.  Each subfigure plots five curves: ``Proposed'' (the proposed online algorithm without ML), ``Revised'' (a modified version described later), ``Greedy'' (a baseline policy described later), ``OPT'' (the offline optimum), and ``Theory'' (the upper bound from Theorem~\ref{theorem:online-alg} multiplied by OPT).  We observe that the proposed algorithm
performs significantly better in practice than the worst-case theoretical upper bound, especially as the update cost range increases.

\begin{figure}[!t]
    \centering
    \subfloat[$R=1$]{
        \includegraphics[width=.16\textwidth]{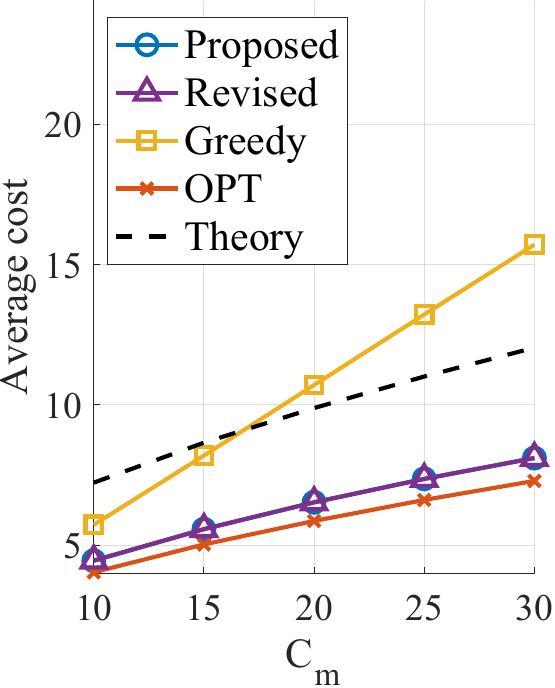}} 
            \subfloat[$R=3$]{
        \includegraphics[width=.16\textwidth]{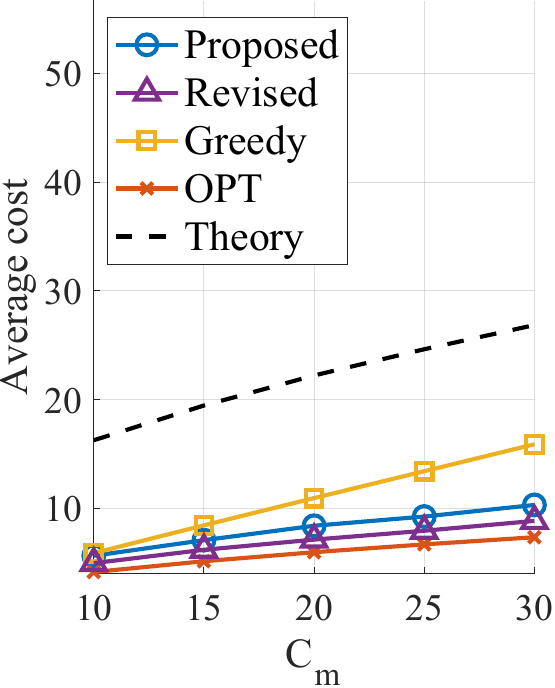}} 
            \subfloat[$R=5$]{
        \includegraphics[width=.16\textwidth]{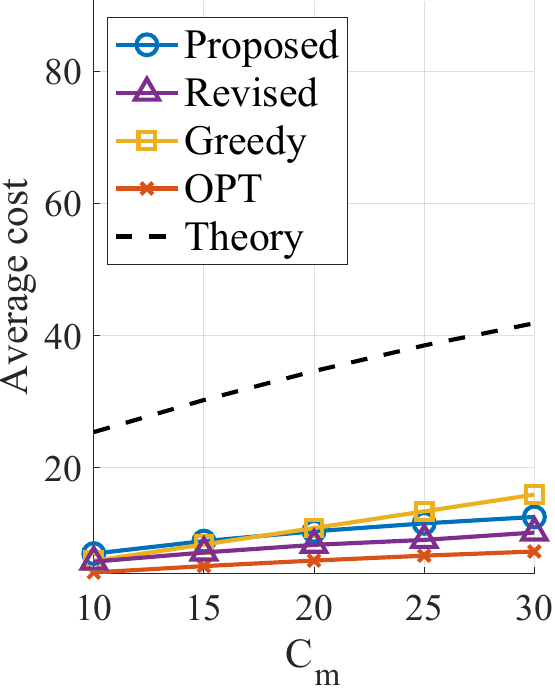}}
   \caption{Performance of online algorithms without ML under a linear aging
function, for transition probabilities $\mathrm{tr}_p = 0.2$ and
$\mathrm{tr}_q = 0.8$.}
    \label{fig:linear-online-0.2}
\end{figure}
\begin{figure}[!t]
    \centering
    \subfloat[$R=1$]{
        \includegraphics[width=.16\textwidth]{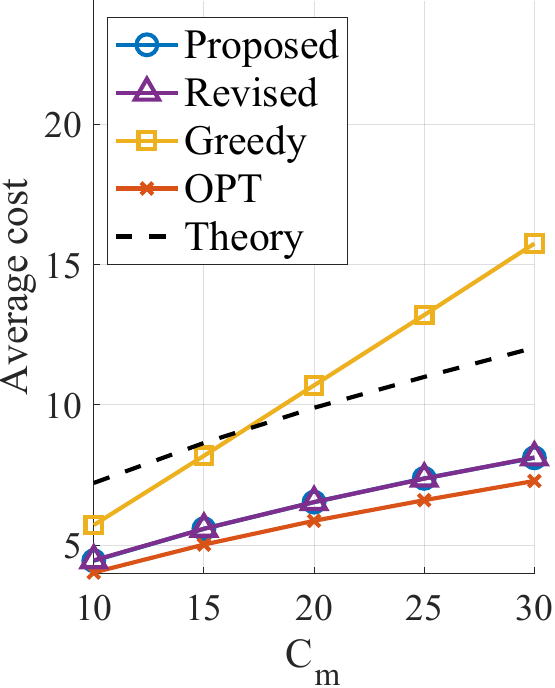}} 
            \subfloat[$R=3$]{
        \includegraphics[width=.16\textwidth]{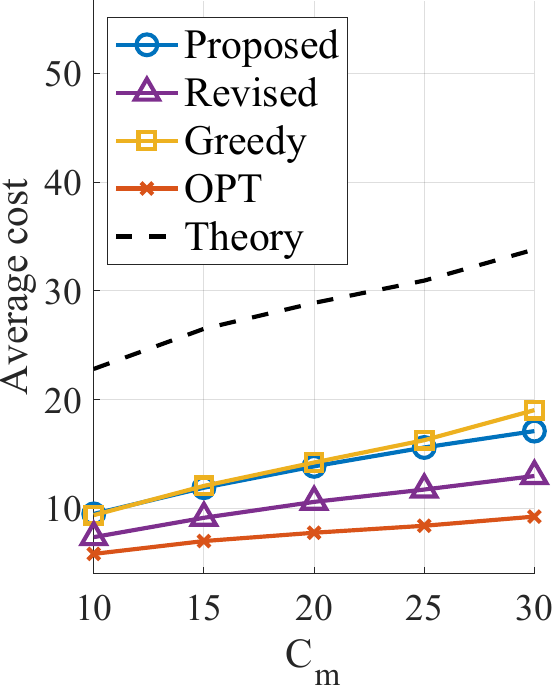}} 
            \subfloat[$R=5$]{
        \includegraphics[width=.16\textwidth]{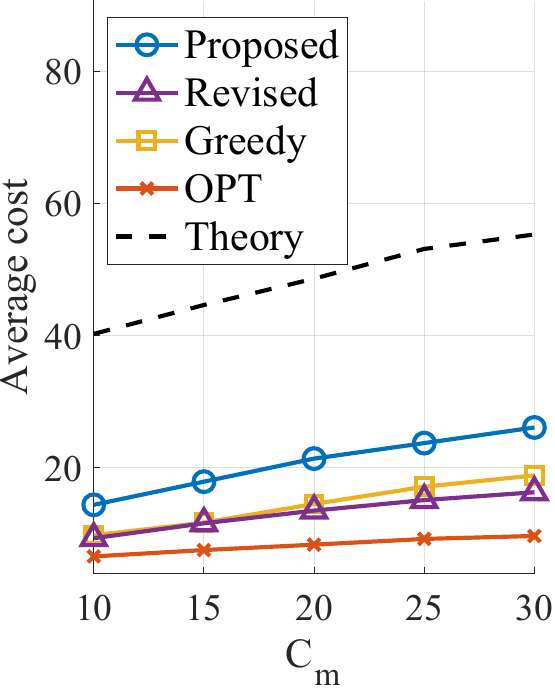}}
    \caption{Performance of online algorithms without ML under a linear aging
function, for transition
    probabilities tr$_p=0.8$ and tr$_q=0.2$.}
    \label{fig:linear-online-0.8}
\end{figure}

For comparison, we also simulate an online greedy algorithm (labeled ``Greedy'' in the figures) that myopically minimizes the current slot cost, i.e.,  it transmits in slot~$t$ if the update cost $C(t)$ is
less than the cost of waiting (i.e., $A(t)+1$). Note that while the greedy algorithm requires knowledge of the current update
cost in each slot, the proposed algorithm does not. From Figs.~\ref{fig:linear-online-0.2} and~\ref{fig:linear-online-0.8}, the proposed algorithm outperforms the greedy baseline except when the
system frequently enters the high cost state with large $C_M$
(i.e., in Fig.~\ref{fig:linear-online-0.8}(c)). This exception would be explained by Lemma~\ref{lemma:max-iteration}: the proposed algorithm must transmit before
activating $\lceil C_m\rceil$ times,  forcing overly frequent updates  when $C_M$ is large and occurs often, as in the environment of
Fig.~\ref{fig:linear-online-0.8}(c).

Although Alg.~\ref{online-alg} asymptotically achieves the optimal competitive ratio and thus serves as an achievability scheme for the lower bound, we observe that it may be too aggressive in such  stochastic environments. To remedy
this, we propose a revised version in which the constant $\theta$ in
Alg.~\ref{online-alg} (and Alg.~\ref{lp-alg}) is replaced by $(1 + (1/C_M))^{C_M} - 1$. By Lemma~\ref{lemma:max-iteration}, the revised algorithm  transmits before
activating $\lceil C_M\rceil$ times,  thereby reducing
the update frequency when $C_M$ is large.  From Figs.~\ref{fig:linear-online-0.2} and~\ref{fig:linear-online-0.8}
(labeled ``Revised'' in the figures), this revised algorithm consistently achieves the best empirical performance.

Given the revised algorithm's superior stochastic performance, we analyze its worst-case guarantees. By Lemma~\ref{lemma:max-iteration}, the virtual packets
arriving in period~$k$ can activate at most $\lceil C_M\rceil$ times from
slot~$t_k$ onward. Let \mbox{$\theta' = (1 + (1/C_M))^{C_M}-1$}. Hence, Eq.~\eqref{eq:theorem:competitive-ratio-1} becomes
\begin{align*}
J(k)
&\le \left(1 + \frac{1}{\theta'}\right)\left( H^*(k) + \lceil C_M\rceil \right) \\
&\le \frac{C_M + 1}{C_m} \left( 1 + \frac{1}{\theta'} \right) J^*(k).
\end{align*}
 As $C_u\to\infty$, we have
$(C_M+1)/C_m \to R$ and $\theta'\to e-1$, yielding an asymptotic competitive
ratio of $(e/(e-1)) R$. This remains order-optimal.

\subsubsection{Online scheduling algorithm with ML}\label{subsection:sim2}

\begin{figure}[!t]
    \centering
    \subfloat[$\epsilon=0$]{
        \includegraphics[width=.16\textwidth]{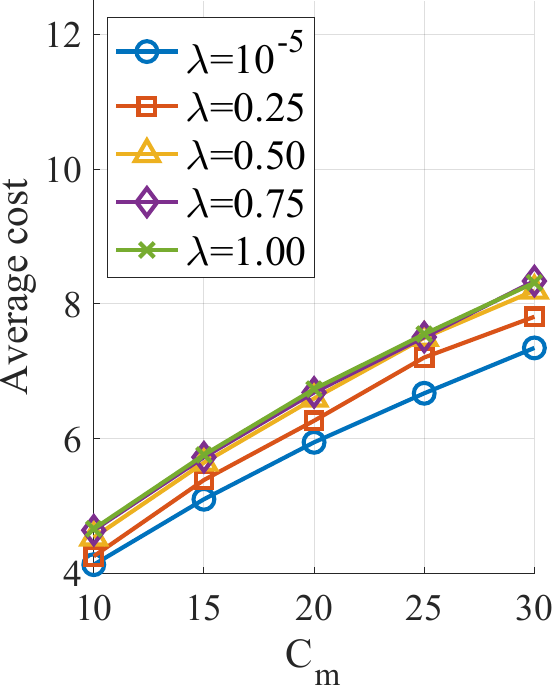 }} 
            \subfloat[$\epsilon=0.5$]{
        \includegraphics[width=.16\textwidth]{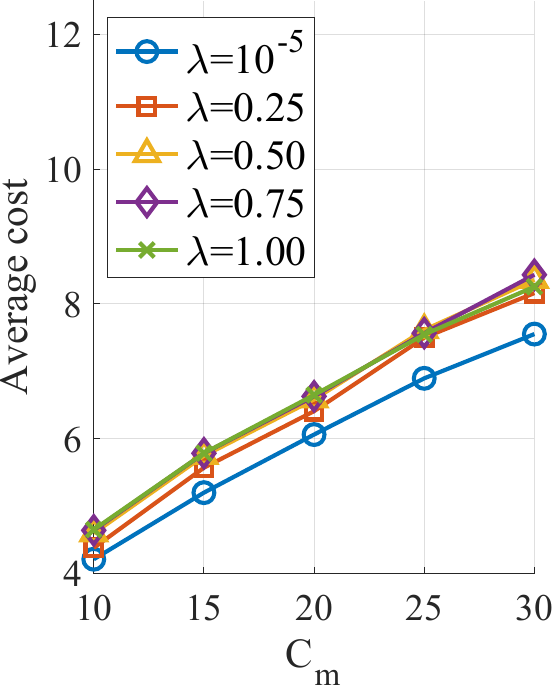 }} 
            \subfloat[$\epsilon=1$]{
        \includegraphics[width=.16\textwidth]{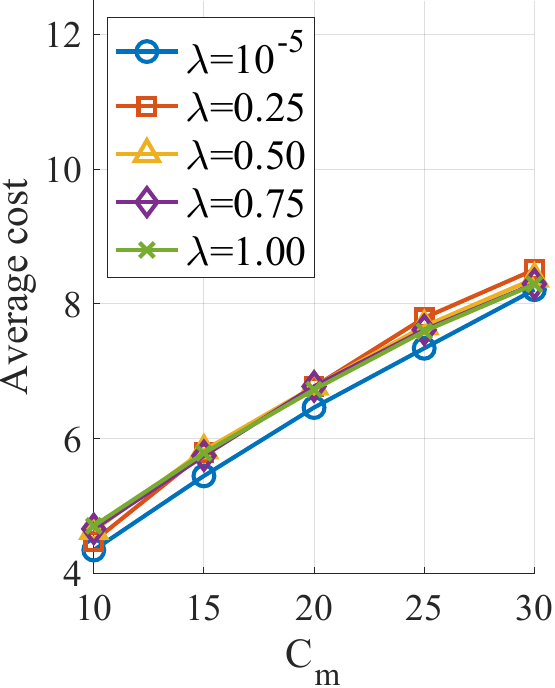 }} 
        \hfill
            \subfloat[$\epsilon=1.5$]{
        \includegraphics[width=.16\textwidth]{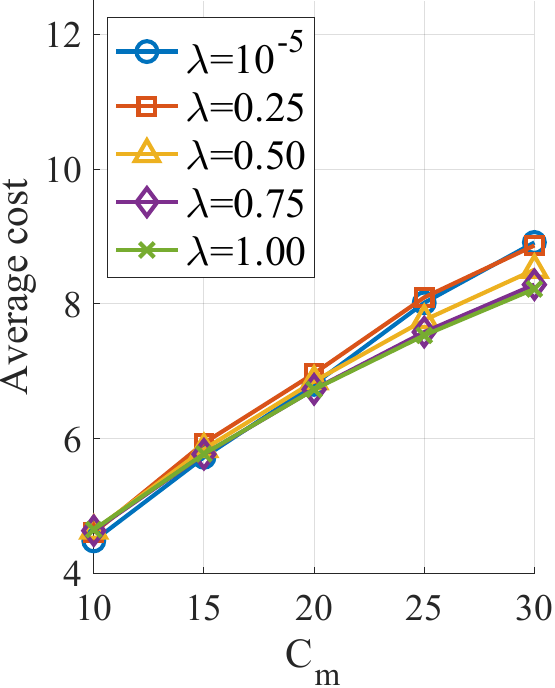 }} 
            \subfloat[$\epsilon=2$]{
        \includegraphics[width=.16\textwidth]{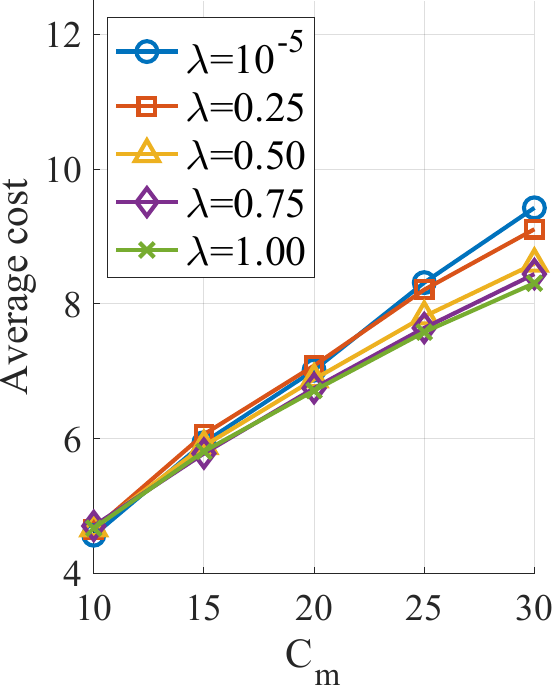 }} 
            \subfloat[$\epsilon=2.5$]{
        \includegraphics[width=.16\textwidth]{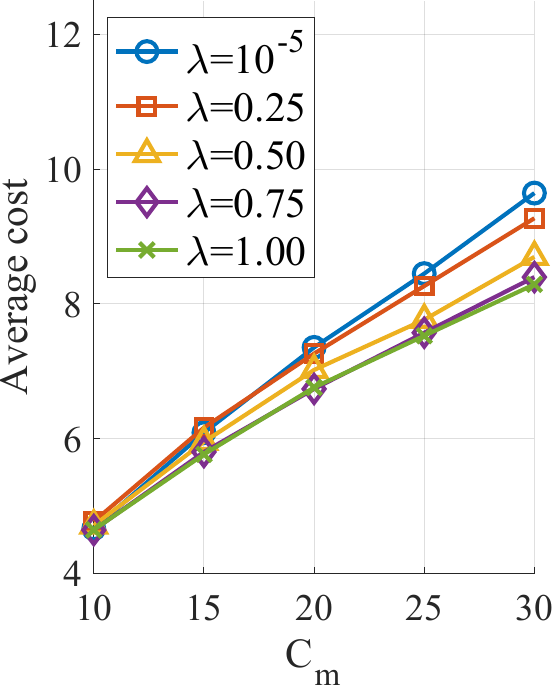 }} 
    \caption{Performance of the proposed ML-augmented online algorithm under a linear aging
function, for transition probabilities tr$_p=0.2$ and tr$_q=0.8$.}
    \label{fig:linear-ml-0.2}
\end{figure}

\begin{figure}[!t]
    \centering
    \subfloat[$\epsilon=0$]{
        \includegraphics[width=.16\textwidth]{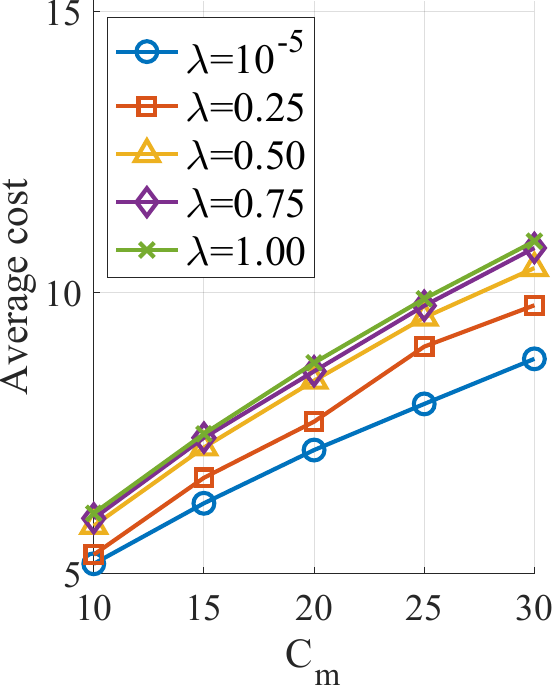 }} 
            \subfloat[$\epsilon=0.5$]{
        \includegraphics[width=.16\textwidth]{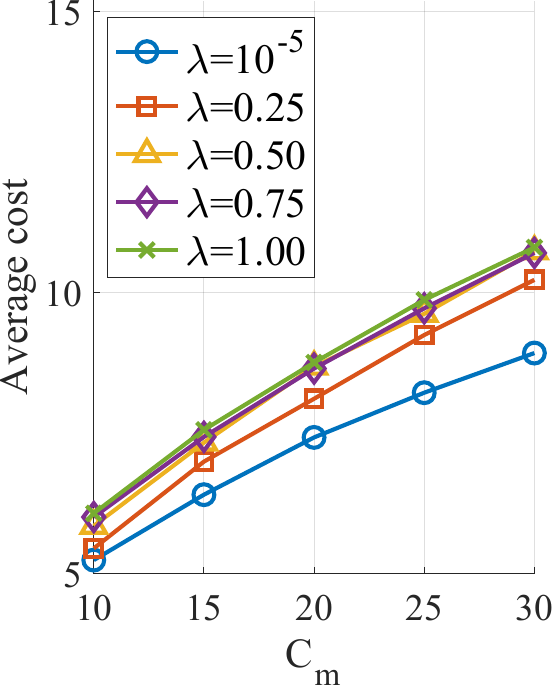 }} 
            \subfloat[$\epsilon=1$]{
        \includegraphics[width=.16\textwidth]{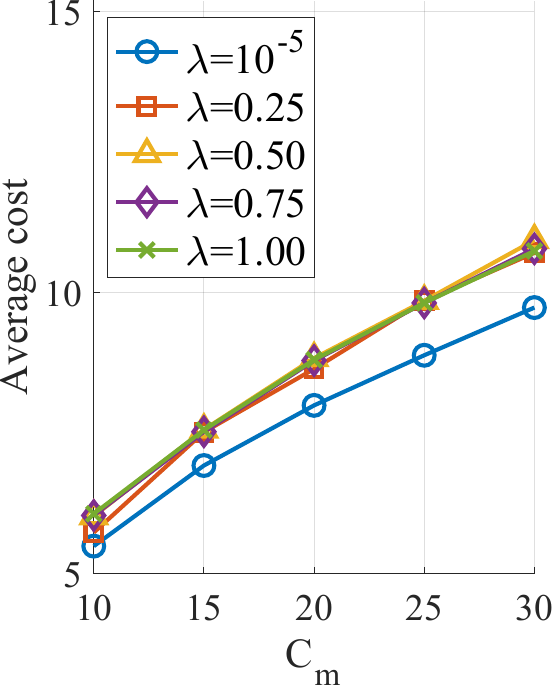 }} 
        \hfill
            \subfloat[$\epsilon=1.5$]{
        \includegraphics[width=.16\textwidth]{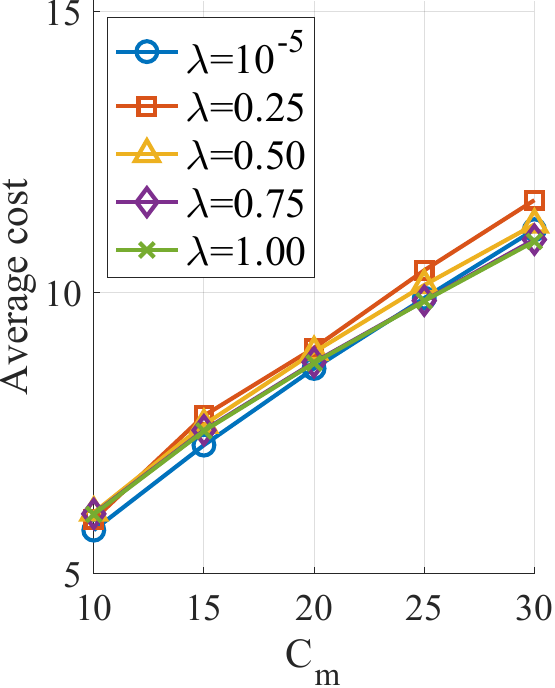 }} 
            \subfloat[$\epsilon=2$]{
        \includegraphics[width=.16\textwidth]{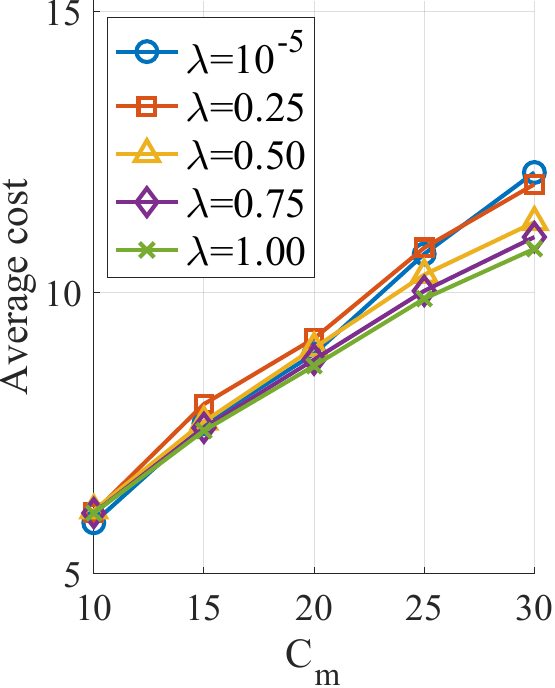 }} 
            \subfloat[$\epsilon=2.5$]{
        \includegraphics[width=.16\textwidth]{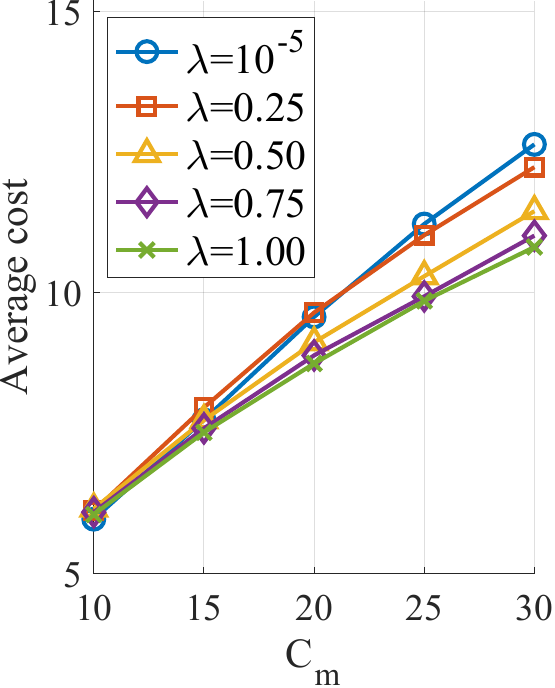 }} 
    \caption{Performance of the proposed ML-augmented online algorithm under a linear aging
function, for transition probabilities tr$_p=0.8$ and tr$_q=0.2$.}
    \label{fig:linear-ml-0.8}
\end{figure}

In this subsection, we evaluate the benefit of incorporating ML advice. Using the same argument as in the previous section, we can show that replacing the constant $\theta$ in Alg.~\ref{lp-alg-ml} with the revised value
$(1 + (1/C_M))^{C_M} - 1$ achieves the robustness of $(e^{\lambda}/(e^{\lambda}-1))R$ and the consistency of $(e^{\lambda}/(e^{\lambda}-1))\lambda R$, as $C_u \to \infty$. The revised algorithm also preserves the order of the optimal consistency–robustness trade-off. Because the revised algorithm performs better in the stochastic environment as shown in the previous section, here we  evaluate its performance when augmented with ML advice.

Let $T^*_{\boldsymbol{\mathcal{M}}}(t)$ denote the offline optimal threshold in slot~$t$. To investigate imperfect ML advice with controllable errors, we model the ML-estimated threshold as $T_{\boldsymbol{\mathcal{M}}}(t) =
T^*_{\boldsymbol{\mathcal{M}}}(t) + \mathcal{N}$, where $\mathcal{N}$ is a
zero-mean Gaussian random variable with variance chosen such that
\[
\Pr\!\left[
\frac{\lvert T_{\boldsymbol{\mathcal{M}}}(t)
      - T^*_{\boldsymbol{\mathcal{M}}}(t)\rvert}
     {T^*_{\boldsymbol{\mathcal{M}}}(t)} \le \epsilon
\right]
= 0.95,
\]
i.e., with $95\%$ probability the relative error is within~$\epsilon$.
The ML advice is then $\mathcal{M}(t)=1$ if $A(t)\ge \max\{T_{\boldsymbol{\mathcal{M}}}(t), 1\}$
and $\mathcal{M}(t)=0$ otherwise.

Figs.~\ref{fig:linear-ml-0.2} (for tr$_p=0.2$ and tr$_q=0.8$) and
\ref{fig:linear-ml-0.8} (for tr$_p=0.8$ and tr$_q=0.2$) plot the time-average cost of the proposed online algorithm with ML for $\epsilon=0,0.5,1,1.5,2,2.5$. Each subfigure shows the results for $\lambda= 10^{-5}, 0.25, 0.5, 0.75, 1$. According to our simulations, the performance of completely following
the ML advice coincides with that of $\lambda = 10^{-5}$; hence, we do not plot it  for clarity. We observe that for small errors ($\epsilon=0,0.5,1$), completely following the ML advice yields the best performance; for large errors
($\epsilon=1.5,2,2.5$), completely ignoring the ML advice with $\lambda=1$
performs best. There also exists  a sharp transition between the two regimes in the stochastic setting, matching the threshold-type behavior by the asymptotic analysis  in the adversarial setting. 
%From the results, we observe that, in the stochastic environments, the ratio between our algorithm with $\lambda=1$ and \mathrm{OPT} is significantly smaller than the worst-case bound $\frac{e^{1/R}}{e^{1/R}-1}$ in Theorem~\ref{theorem:competitive-ratio}. 
%For example, in Fig.~\ref{fig:result0.3}(a), the ratio is $5.3/4.3 \approx 1.23$ for $C_M=15$, which is well below $\frac{e^{10/15}}{e^{10/15}-1} \approx 2.1$;  the ratio is $8.3/4.9 \approx 1.69$  for $C_M=35$, which is also well below $\frac{e^{10/35}}{e^{10/35}-1} \approx 4.02$. 

\begin{figure}[!t]
    \centering
    \subfloat[$R=1$]{
        \includegraphics[width=.16\textwidth]{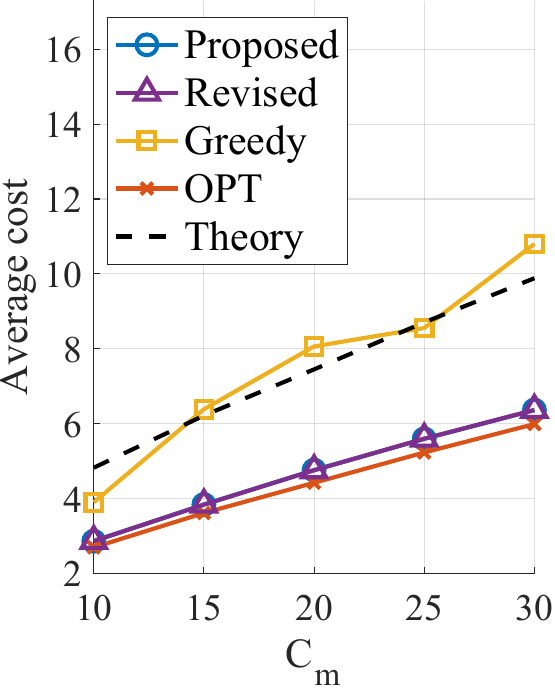 }} 
            \subfloat[$R=3$]{
        \includegraphics[width=.16\textwidth]{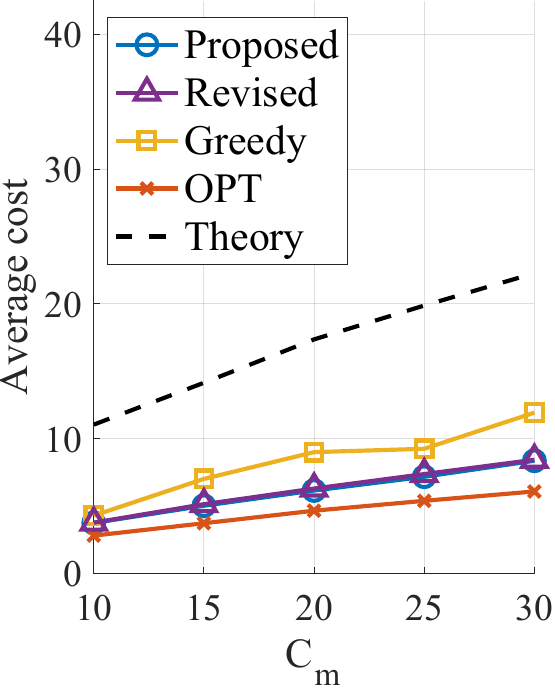 }} 
            \subfloat[$R=5$]{
        \includegraphics[width=.16\textwidth]{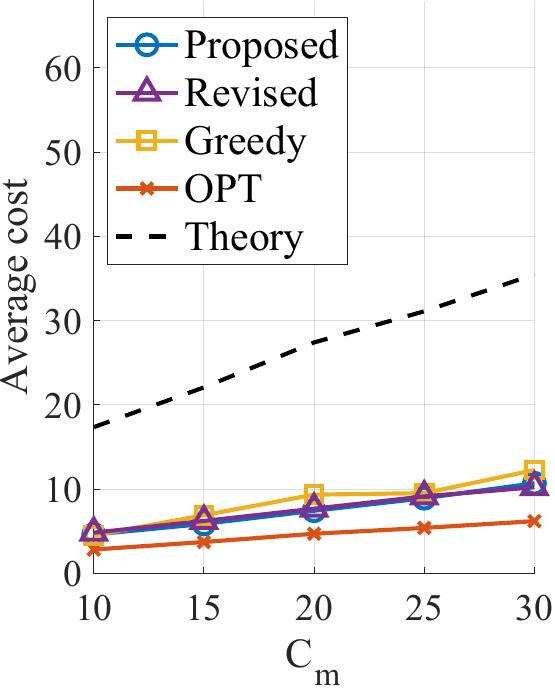 }}
    \caption{Performance of online algorithms without ML under a nonlinear aging
function, for transition
    probabilities tr$_p=0.2$ and tr$_q=0.8$.}
    \label{fig:nonlinear-online-0.2}
\end{figure}
\begin{figure}[!t]
    \centering
    \subfloat[$R=1$]{
        \includegraphics[width=.16\textwidth]{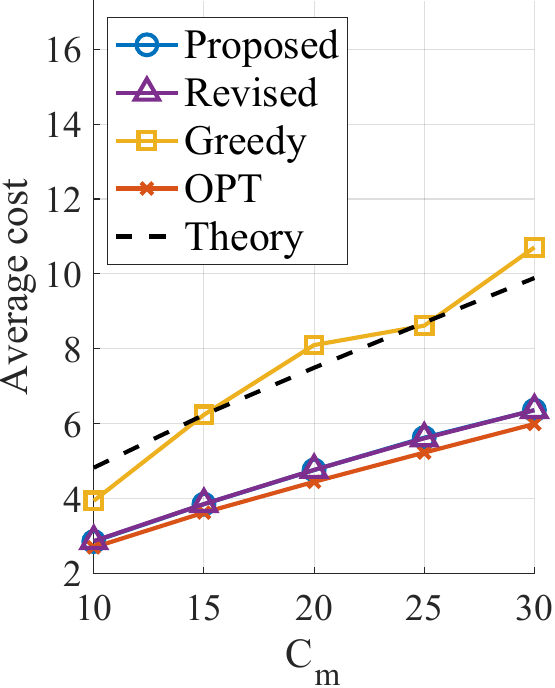}} 
            \subfloat[$R=3$]{
        \includegraphics[width=.16\textwidth]{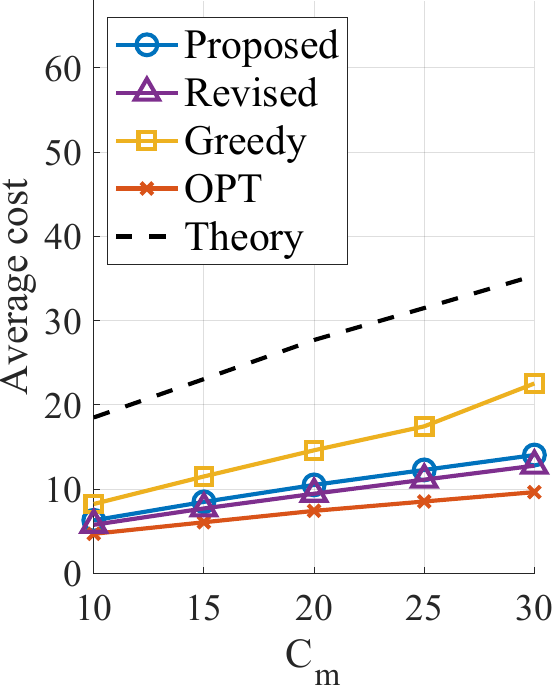}} 
            \subfloat[$R=5$]{
        \includegraphics[width=.16\textwidth]{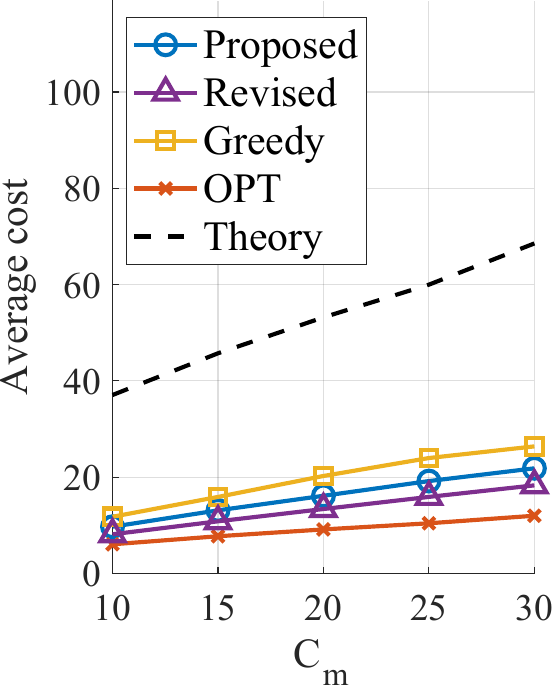}}
    \caption{Performance of online algorithms without ML under a nonlinear aging
function, for transition
    probabilities tr$_p=0.8$ and tr$_q=0.2$.}
    \label{fig:nonlinear-online-0.8}
\end{figure}

\subsection{Nonlinear aging function} \label{subsection:sim-nonlinear}

In this section, we consider a nonlinear aging function similar to that in
\cite{kosta2017age}. Specifically, if $x$ slots have elapsed since the most
recent update, then the age of information in the current slot is given by
$\lfloor e^{0.3x} \rfloor - 1$. Throughout this section, we fix $R = 2$.

Figs.~\ref{fig:nonlinear-online-0.2} (for $\mathrm{tr}_p=0.2$ and
$\mathrm{tr}_q=0.8$) and~\ref{fig:nonlinear-online-0.8} (for
$\mathrm{tr}_p=0.8$ and $\mathrm{tr}_q=0.2$) show the time-average cost of the online algorithms without ML. As in the linear aging case, both the proposed and the
revised algorithms outperform the baseline policy.

We also evaluate the revised algorithm with ML advice in
Figs.~\ref{fig:nonlinear-ml-0.2} (for $\mathrm{tr}_p=0.2$ and
$\mathrm{tr}_q=0.8$) and~\ref{fig:nonlinear-ml-0.8} (for
$\mathrm{tr}_p=0.8$ and $\mathrm{tr}_q=0.2$), for several values of the ML error
parameter~$\epsilon$. In Fig.~\ref{fig:nonlinear-ml-0.2}, when
$\epsilon = 1, 3$, completely following the ML advice yields the best
performance; when $\epsilon = 5, 7$, the performance is nearly identical for
all values of~$\lambda$;  when $\epsilon = 9, 11$, completely ignoring the
ML advice is optimal. Similarly, in Fig.~\ref{fig:nonlinear-ml-0.8}, when
$\epsilon = 1, 2$, completely following the ML advice is  optimal; when
$\epsilon = 3, 4$, the performance is nearly the same for all~$\lambda$; 
when $\epsilon = 5, 6$, completely ignoring the ML advice performs best.

These results exhibit the same qualitative behavior observed under the linear
aging case in the previous section: either blindly following the ML advice or completely ignoring it yields
near-optimal performance, while partially trusting the ML advice provides
little benefit.

\begin{figure}[!t]
    \centering
    \subfloat[$\epsilon=1$]{
        \includegraphics[width=.16\textwidth]{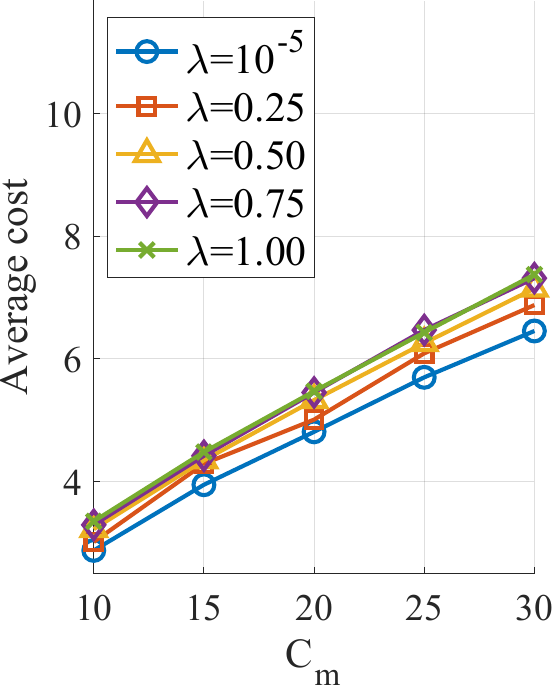 }} 
            \subfloat[$\epsilon=3$]{
        \includegraphics[width=.16\textwidth]{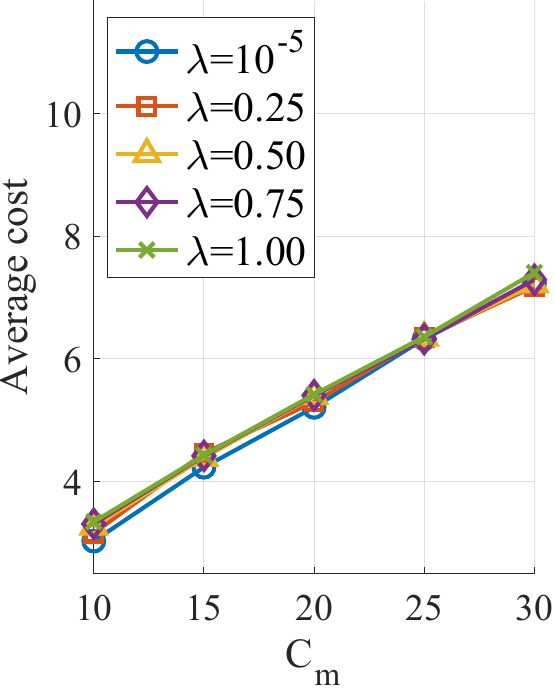 }} 
            \subfloat[$\epsilon=5$]{
        \includegraphics[width=.16\textwidth]{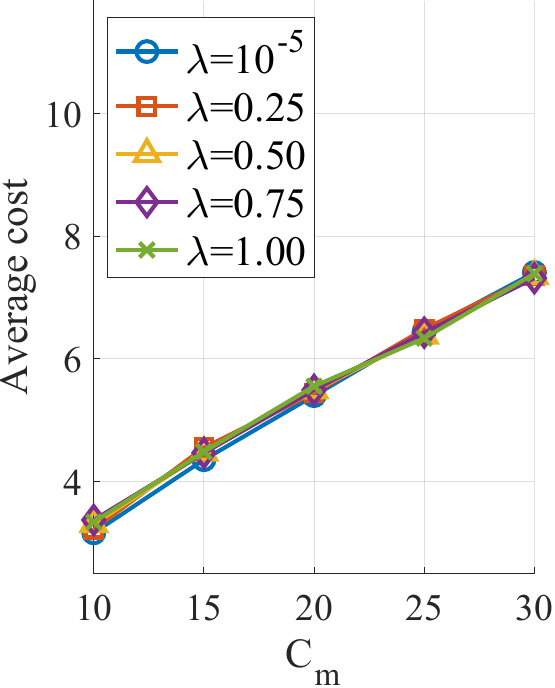 }} 
        \hfill
            \subfloat[$\epsilon=7$]{
        \includegraphics[width=.16\textwidth]{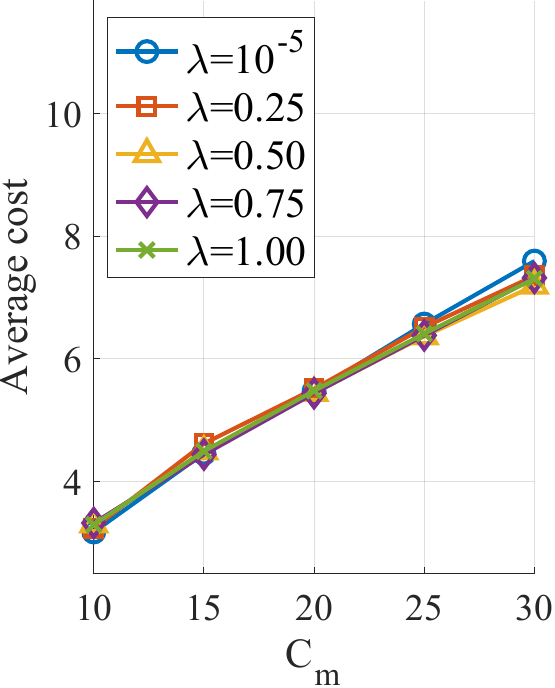 }} 
            \subfloat[$\epsilon=9$]{
        \includegraphics[width=.16\textwidth]{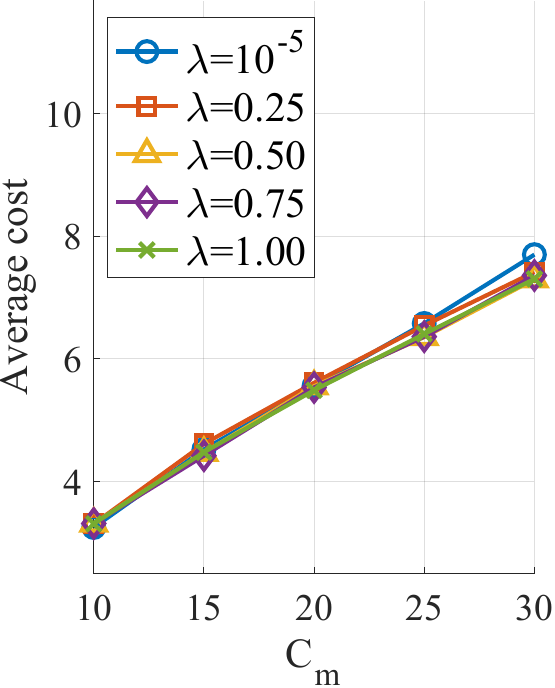 }} 
            \subfloat[$\epsilon=11$]{
        \includegraphics[width=.16\textwidth]{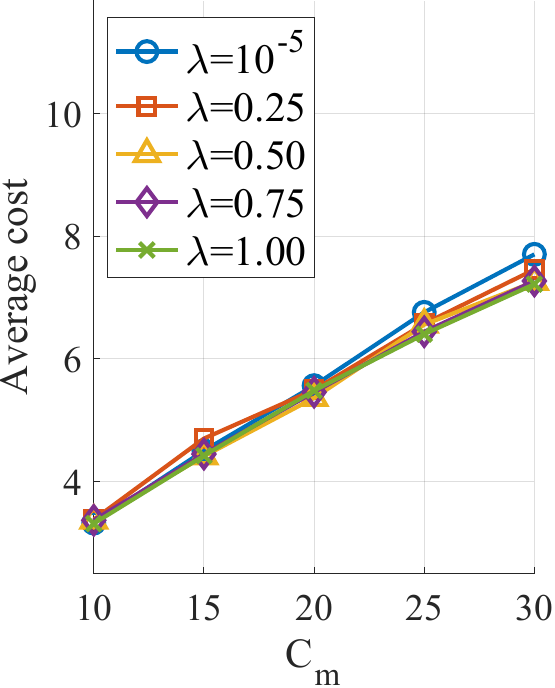 }} 
    \caption{Performance of the proposed ML-augmented online algorithm under a nonlinear aging
function, for transition probabilities tr$_p=0.2$ and tr$_q=0.8$.}
    \label{fig:nonlinear-ml-0.2}
\end{figure}

\begin{figure}[!t]
    \centering
            \subfloat[$\epsilon=1$]{
        \includegraphics[width=.16\textwidth]{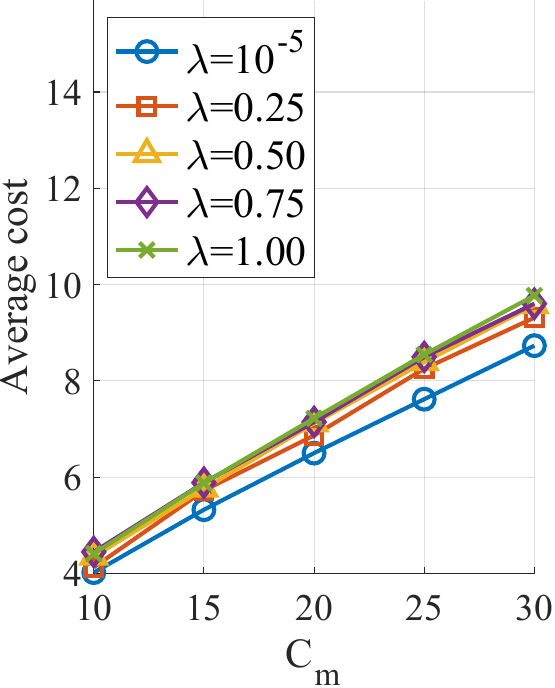 }} 
            \subfloat[$\epsilon=2$]{
        \includegraphics[width=.16\textwidth]{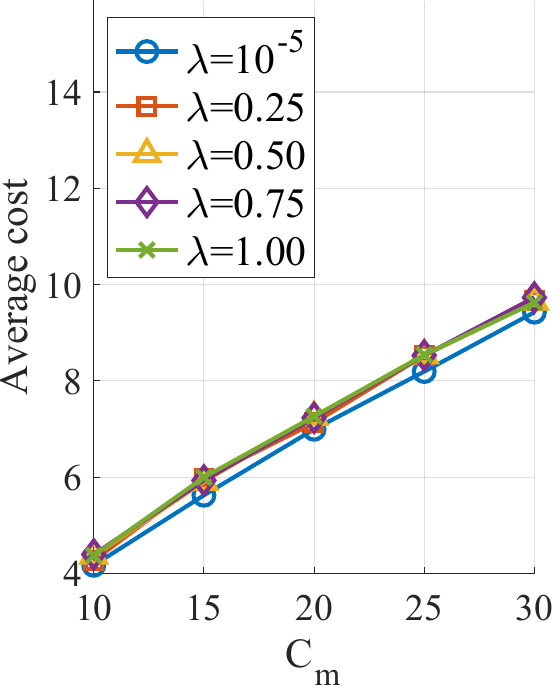 }}
            \subfloat[$\epsilon=3$]{
        \includegraphics[width=.16\textwidth]{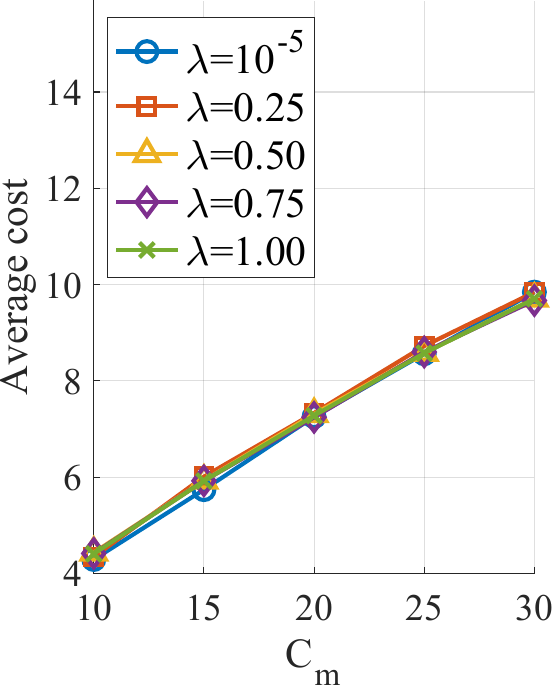 }} 
        \hfill
            \subfloat[$\epsilon=4$]{
        \includegraphics[width=.16\textwidth]{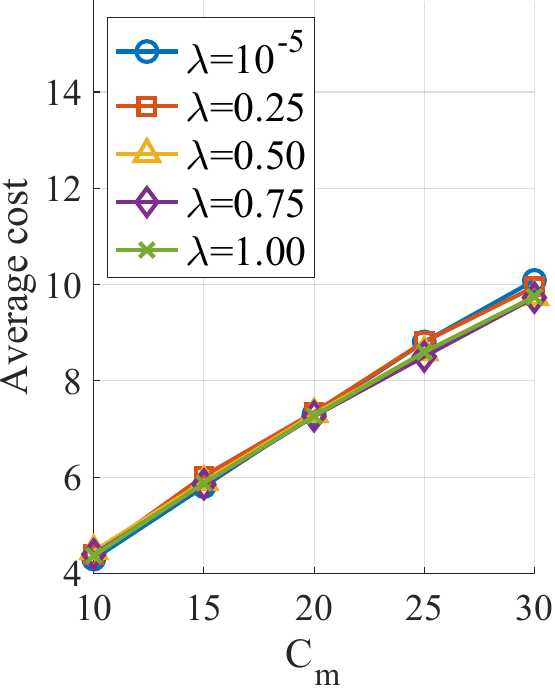 }}
\subfloat[$\epsilon=5$]{
        \includegraphics[width=.16\textwidth]{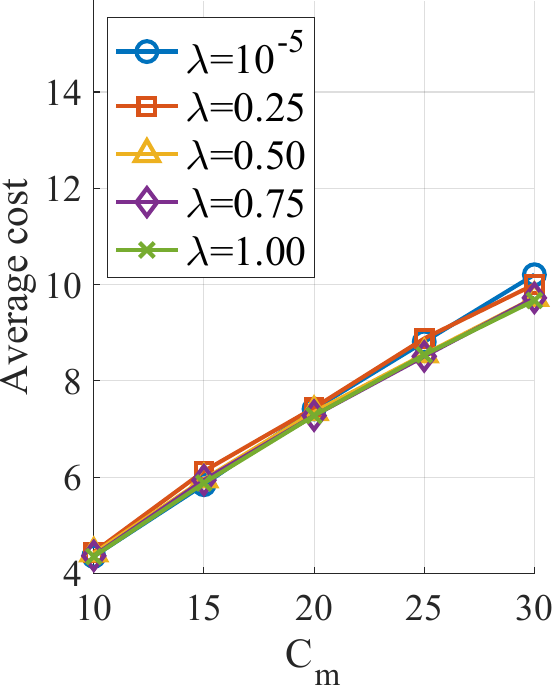 }}  
                    \subfloat[$\epsilon=6$]{
        \includegraphics[width=.16\textwidth]{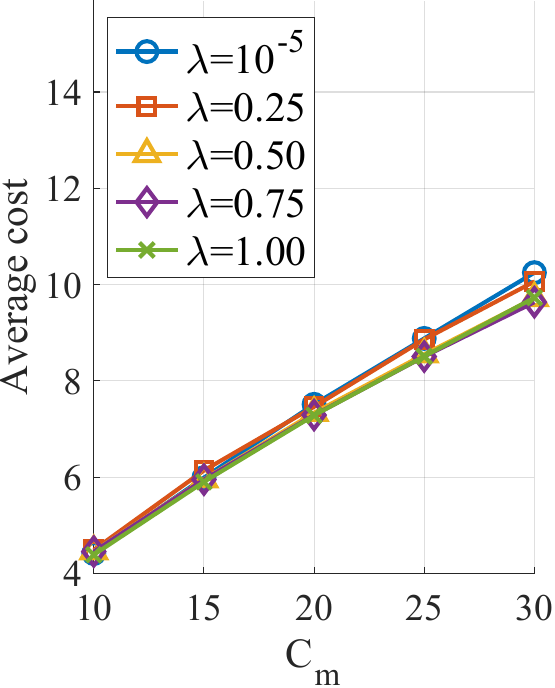 }} 
  \caption{Performance of the proposed ML-augmented online algorithm under a nonlinear aging
function, for transition probabilities tr$_p=0.8$ and tr$_q=0.2$.}
    \label{fig:nonlinear-ml-0.8}
\end{figure}

\section{Conclusion}
This paper investigated a mobile information updating system subject to four
sources of uncertainty. We developed online scheduling algorithms that enable a
mobile device to cost-efficiently maintain fresh information at a central
entity. The proposed online algorithm without ML asymptotically achieves the optimal
competitive ratio, while the ML-augmented version also asymptotically attains the optimal
consistency–robustness trade-off. Moreover, when augmented with ML, we showed
that either blindly following or completely ignoring the ML advice minimizes
the competitive ratio. This work opens several promising research directions for network design under non-stationary uncertainty. Interesting extensions include dynamically adjusting the  threshold (for either blindly following ML or completely ignoring it) because the reliability of the ML advice is unknown in general, developing an optimal algorithm for both the adversarial and stochastic environments, exploring multi-device or networked update systems, and incorporating sampling decisions jointly with update scheduling.

\section{Acknowledgments}
We thank the authors of \cite{liu2025learning} for pointing out mistakes in our  earlier preliminary work \cite{tseng2019online}. 
This research was supported by the National Science and Technology Council, Taiwan, under Grant No.~110-2221-E-305-008-MY3 and  113-2628-E-305-001-MY3.

{
\small
	\bibliographystyle{IEEEtran}
	\bibliography{IEEEabrv,ref}

@STRING{IEEE_J_JSAC       = "{IEEE} J. Sel. Areas Commun."}

@STRING{IEEE_J_WCOM       = "{IEEE} Trans. Wireless Commun."}

@STRING{IEEE_J_IT         = "{IEEE} Trans. Inf. Theory"}

@STRING{IEEE_J_IOT        = "{IEEE} Internet Things J."}

@STRING{IEEE_J_MC         = "{IEEE} Trans. Mobile Comput."}

@STRING{IEEE_J_NET        = "{IEEE/ACM} Trans. Netw."}

@STRING{IEEE_J_NSE        = "{IEEE} Trans. Netw. Sci. Eng."}

@STRING{IEEE_M_COM        = "{IEEE} Commun. Mag."}

@STRING{IEEE_O_CSTO       = "{IEEE} Commun. Surveys Tuts."}

@inproceedings{abolhassani2021fresh,
  title={Fresh caching for dynamic content},
  author={Abolhassani, Bahman and Tadrous, John and Eryilmaz, Atilla and Yeh, Edmund},
  booktitle={Proc. IEEE INFOCOM},
  pages={1--10},
  year={2021}
}

@inproceedings{kosta2017age,
  title={Age and value of information: Non-linear age case},
  author={Kosta, Antzela and Pappas, Nikolaos and Ephremides, Anthony and Angelakis, Vangelis},
  booktitle={Proc. IEEE ISIT},
  pages={326--330},
  year={2017}
}

@inproceedings{kaul2012real,
  title        = {Real-Time Status: How Often Should One Update?},
  author       = {Kaul, Sanjit and Yates, Roy and Gruteser, Marco},
  booktitle    = {Proc. IEEE INFOCOM},
  pages        = {2731--2735},
  year         = {2012}
}

@article{costa2016age,
  title        = {On the Age of Information in Status Update Systems with Packet Management},
  author       = {Costa, Maice and Codreanu, Marian and Ephremides, Anthony},
  journal      = IEEE_J_IT,
  volume       = {62},
  number       = {4},
  pages        = {1897--1910},
  year         = {2016}
}

@article{yates2018age,
  title        = {The Age of Information: Real-Time Status Updating by Multiple Sources},
  author       = {Yates, Roy D. and Kaul, Sanjit K.},
  journal      = IEEE_J_IT,
  volume       = {65},
  number       = {3},
  pages        = {1807--1827},
  year         = {2018}
}

@article{kadota2018scheduling,
  title        = {Scheduling Policies for Minimizing Age of Information in Broadcast Wireless Networks},
  author       = {Kadota, Igor and Sinha, Abhishek and Uysal-Biyikoglu, Elif and Singh, Rahul and Modiano, Eytan},
  journal      = IEEE_J_NET,
  volume       = {26},
  number       = {6},
  pages        = {2637--2650},
  year         = {2018}
}

@article{yates2021age,
  title        = {Age of Information: An Introduction and Survey},
  author       = {Yates, Roy D. and Sun, Yin and Brown, D. Richard and Kaul, Sanjit K. and Modiano, Eytan and Ulukus, Sennur},
  journal      = IEEE_J_JSAC,
  volume       = {39},
  number       = {5},
  pages        = {1183--1210},
  year         = {2021}
}

@article{mankar2021throughput,
  title        = {Throughput and Age of Information in a Cellular-Based {IoT} Network},
  author       = {Mankar, Praful D. and Chen, Zheng and Abd-Elmagid, Mohamed A. and Pappas, Nikolaos and Dhillon, Harpreet S.},
  journal      = IEEE_J_WCOM,
  volume       = {20},
  number       = {12},
  pages        = {8248--8263},
  year         = {2021}
}

@article{park2020centralized,
  title        = {Centralized and Distributed Age of Information Minimization with Nonlinear Aging Functions in the Internet of Things},
  author       = {Park, Taehyeun and Saad, Walid and Zhou, Bo},
  journal      = IEEE_J_IOT,
  volume       = {8},
  number       = {10},
  pages        = {8437--8455},
  year         = {2020}
}

@inproceedings{ornee2019sampling,
  title        = {Sampling for Remote Estimation Through Queues: Age of Information and Beyond},
  author       = {Ornee, Tasmeen Zaman and Sun, Yin},
  booktitle    = {Proc. IEEE WiOpt},
  pages        = {1--8},
  year         = {2019}
}

@inproceedings{tseng2019online,
  title        = {Online Energy-Efficient Scheduling for Timely Information Downloads in Mobile Networks},
  author       = {Tseng, Yi-Hsuan and Hsu, Yu-Pin},
  booktitle    = {Proc. IEEE ISIT},
  pages        = {1022--1026},
  year         = {2019}
}

@inproceedings{nath2018optimum,
  title        = {Optimum Energy Efficiency and Age-of-Information Tradeoff in Multicast Scheduling},
  author       = {Nath, Samrat and Wu, Jingxian and Yang, Jing},
  booktitle    = {Proc. IEEE ICC},
  pages        = {1--6},
  year         = {2018}
}

@inproceedings{lykouris2018competitive,
  title        = {Competitive Caching with Machine Learned Advice},
  author       = {Lykouris, Thodoris and Vassilvitskii, Sergei},
  booktitle    = {Proc. ICML},
  pages        = {3296--3305},
  year         = {2018}
}

@article{liu2025learning,
  title        = {Learning-Augmented Online Minimization of Age of Information and Transmission Costs},
  author       = {Liu, Zhongdong and Zhang, Keyuan and Li, Bin and Sun, Yin and Hou, Y. Thomas and Ji, Bo},
  journal      = IEEE_J_NSE,
  year={2025},
  volume={12},
  number={5},
  pages={3480--3496}
}

@article{sinha2022optimizing,
  title        = {Optimizing Age-of-Information in Adversarial and Stochastic Environments},
  author       = {Sinha, Abhishek and Bhattacharjee, Rajarshi},
  journal      = IEEE_J_IT,
  volume       = {68},
  number       = {10},
  pages        = {6860--6880},
  year         = {2022}
}

@article{lin2025optimal,
  title        = {Optimal Algorithms for Online Age-of-Information Optimization in Energy Harvesting Systems},
  author       = {Lin, Qiulin and Su, Junyan and Chen, Minghua},
  journal      = IEEE_J_NET,
  year         = {2025},
  note         = {{Early Access}}
}

@article{saurav2023online,
  title        = {Online Energy Minimization Under a Peak Age of Information Constraint},
  author       = {Saurav, Kumar and Vaze, Rahul},
  journal      = {IEEE J. Sel. Areas Inf. Theory},
  volume       = {4},
  pages        = {579--590},
  year         = {2023}
}

@inproceedings{tripathi2021online,
  title        = {An Online Learning Approach to Optimizing Time-Varying Costs of {AoI}},
  author       = {Tripathi, Vishrant and Modiano, Eytan},
  booktitle    = {Proc.  ACM MobiHoc},
  pages        = {241--250},
  year         = {2021}
}

@article{buchbinder2009design,
  title        = {The Design of Competitive Online Algorithms via a Primal-Dual Approach},
  author       = {Buchbinder, Niv and Naor, Joseph Seffi},
  journal      = {Found. Trends Theor. Comput. Sci.},
  volume       = {3},
  number       = {2--3},
  pages        = {93--263},
  year         = {2009}
}

@inproceedings{bamas2020primal,
  title        = {The Primal-Dual Method for Learning Augmented Algorithms},
  author       = {Bamas, Etienne and Maggiori, Andreas and Svensson, Ola},
  booktitle    = {Proc. NeurIPS},
  volume       = {33},
  pages= {{20}{083}--{20}{094}},
   year         = {2020}
}

@article{wang2025understanding,
  title        = {Understanding the Fundamental Trade-Off Between Age of Information and Throughput in Unreliable Wireless Networks},
  author       = {Wang, Lin and Hou, I-Hong},
  journal      = {Proc. ACM MobiHoc},
  year         = {2025}
}

@article{gu2019timely,
  title        = {Timely Status Update in Internet of Things Monitoring Systems: An Age-Energy Tradeoff},
  author       = {Gu, Yifan and Chen, He and Zhou, Yong and Li, Yonghui and Vucetic, Branka},
  journal      = IEEE_J_IOT,
  volume       = {6},
  number       = {3},
  pages        = {5324--5335},
  year         = {2019}
}

@book{lattimore2020bandit,
  title        = {Bandit Algorithms},
  author       = {Lattimore, Tor and Szepesv{\'a}ri, Csaba},
  publisher    = {Cambridge Univ. Press},
  address      = {Cambridge, U.K.},
  year         = {2020}
}

@article{lin2022survey,
  title={A survey on drx mechanism: Device power saving from {LTE} and {5G} new radio to {6G} communication systems},
  author={Lin, Kuang-Hsun and Liu, He-Hsuan and Hu, Kai-Hsin and Huang, An and Wei, Hung-Yu},
  journal=IEEE_O_CSTO,
  volume={25},
  number={1},
  pages={156--183},
  year={2022}
}

@article{takeda2020understanding,
  title={Understanding the heart of the {5G} air interface: An overview of physical downlink control channel for {5G} new radio},
  author={Takeda, Kazuki and Xu, Huilin and Kim, Taehyoung and Schober, Karol and Lin, Xingqin},
  journal=IEEE_M_COM,
  volume={4},
  number={3},
  pages={22--29},
  year={2020}
}

@inproceedings{arafa2017age,
  title={Age minimization in energy harvesting communications: Energy-controlled delays},
  author={Arafa, Ahmed and Ulukus, Sennur},
  booktitle={Proc. Asilomar Conf. Signals, Syst., Comput.},
  pages={1801--1805},
  year={2017}
}

@article{gupta2015survey,
  title={Survey of important issues in UAV communication networks},
  author={Gupta, Lav and Jain, Raj and Vaszkun, Gabor},
  journal=IEEE_O_CSTO ,
  volume={18},
  number={2},
  pages={1123--1152},
  year={2015}
}

@inproceedings{karki2020characterizing,
  title={Characterizing power consumption of dual-frequency {GNSS} of smartphone},
  author={Karki, Bikram and Won, Myounggyu},
  booktitle={Proc. IEEE GLOBECOM},
  pages={1--6},
  year={2020}
}

@article{hsu2019scheduling,
  title={Scheduling algorithms for minimizing age of information in wireless broadcast networks with random arrivals},
  author={Hsu, Yu-Pin and Modiano, Eytan and Duan, Lingjie},
  journal=IEEE_J_MC,
  volume={19},
  number={12},
  pages={2903--2915},
  year={2019}
}
}

\clearpage

\setcounter{page}{1}
\twocolumn[
  \begin{center}
    {\Huge Supplementary Material
    }
  \end{center}
  \vspace{1cm}
]

\appendices

\section{Proof of Lemma~\ref{lemma:sum-x-bound}}\label{appendix:lemma:sum-x-bound}
We use the notation $\sum_{\tau=0}^{\infty}x(\tau) \big|_{\text{condition}}$ to represent the value of $\sum_{\tau=T_i}^{\infty}$ under the specific condition. Fix a slot~$t$.  We prove the claim by induction on $n$.  When $n=1$, by Eq.~(\ref{eq:cum-x-update}) we have 
\begin{align*}
	\left(\sum_{\tau=T_i}^{\infty}x(\tau) \bigg|_{n=1}\right) & = \left(1+\frac{1}{C_M}\right)\left(\sum_{\tau=T_i}^{\infty}x(\tau) \bigg|_{n=0}\right) + \frac{1}{\theta C_M} \\
	& \geq  \frac{1}{\theta C_M}= \frac{\left(1 + \frac{1}{C_M}\right)^1 - 1}{\theta}, 
\end{align*}
for all $i\in P(t)$.

Assume that the result holds for $n-1$, i.e., 
\begin{align*}
\left(\sum_{\tau = T_i}^{\infty} x(\tau)\bigg|_{n-1} \right)
\geq \frac{\left(1 + \frac{1}{C_M}\right)^{n-1} - 1}{\theta},	
\end{align*}
for all $i\in P(t)$. 

We show that the result holds for $n$: after the additional step, by Eq.~(\ref{eq:cum-x-update}) we have
\begin{align*}
\left(\sum_{\tau=T_i}^{\infty} x(\tau)\bigg|_{n}\right)& =\left(1+\frac{1}{C_M}\right)\left(\sum_{\tau=T_i}^{\infty} x(\tau)\bigg|_{n-1}\right)+\frac{1}{\theta C_M}	\\
& \geq \left(1+\frac{1}{C_M}\right)\left(\frac{\left(1 + \frac{1}{C_M}\right)^{n-1} - 1}{\theta}\right)\\
& \quad+\frac{1}{\theta C_M}	\\
&= \frac{\left(1 + \frac{1}{C_M}\right)^{n} - 1}{\theta},
\end{align*}
for all $i\in P(t)$. This completes the inductive step and proves the lemma.

%
%\begin{align*}
%\sum_{\tau = T_i}^{\infty} x(\tau) 
%\geq x(t_1) 
%\geq \frac{1}{\theta C_M} 
%= \frac{\left(1 + \frac{1}{C_M}\right)^1 - 1}{\theta},
%\end{align*}
%for all $i$ such that $T_i \leq t$. 
%
%Suppose that, after the virtual packets in $P(t)$ activate  $(n-1)$ times, we have 
%
%After the virtual packets in $P(t)$ activate  $n$ times, by Eq.~(\ref{eq:increment}), we have
%\begin{align*}
%	\sum_{\tau = T_i}^{\infty} x(\tau)  \geq 
%\end{align*}
%
%
% in a slot $t_n \geq t$, 
%Line~\ref{lp-alg:x} increases $x(t_n)$ by at least 
%\begin{align*}
%\frac{1}{C_M} \left[  \frac{\left(1 + \frac{1}{C_M}\right)^{n-1} - 1}{\theta} \right] 
%+ \frac{1}{\theta C_M},
%\end{align*}
%where the bracketed term follows from the induction hypothesis. 
%Adding this increment to the existing value gives
%\begin{align*}
%\sum_{\tau = T_i}^{\infty} x(\tau) 
%\geq & \underbrace{\frac{\left(1 + \frac{1}{C_M}\right)^{n-1} - 1}{\theta}}_{\text{lower bound on the existing value}}\\
%&+ \underbrace{\frac{1}{C_M} \left[ \frac{\left(1 + \frac{1}{C_M}\right)^{n-1} - 1}{\theta} \right] 
%+ \frac{1}{\theta C_M}}_{\text{increment of $x(t_n)$}}  \\
%=& \frac{\left(1 + \frac{1}{C_M}\right)^n - 1}{\theta},
%\end{align*}
%for all $i$ such that $T_i\leq t$. This completes the inductive step and proves the lemma.

\section{Proof of Theorem~\ref{theorem:competitive-ratio}} \label{appendix:theorem:competitive-ratio}
Fix an instance~$\mathcal{I}$. Suppose that an optimal offline scheduling algorithm 
(denoted by $\boldsymbol{\pi}^*$) clears the virtual queue in slots $t_1, \cdots, t_n$, for a total 
of $n$ clearing operations. Let $t_0=0$ and $t_{n+1}=T$. We divide the timeline into $n+1$ periods, 
where period~$k$ consists of slots $t_{k-1}+1$ through $t_k$.  

Let $J^*(k)$ denote the cost incurred by $\boldsymbol{\pi}^*$ in period~$k$. 
The total cost in Eq.~(\ref{eq:cost-2}) incurred by $\boldsymbol{\pi}^*$ is then $\sum_{k=1}^{n+1} J^*(k)$. 
We calculate $J^*(k)$ for two cases.
\begin{enumerate}
	\item For  $k \in \{1, \cdots, n\}$:  For each slot~$\tau$ in period~$k$, the number of virtual packets present in the virtual queue is $\sum_{i=1}^{\infty} \mathbf{1}_{\{t_{k-1}+1 \leq T_i \leq \tau\}}$,
which checks all virtual packets that arrived after the previous clearing in slot~$t_{k-1}$ until slot~$\tau$. Hence, the  holding cost of all the virtual packets that arrive in  period~$k$ is $\sum_{\tau = t_{k-1}+1}^{t_k-1} \sum_{i=1}^{\infty} \mathbf{1}_{\{t_{k-1}+1 \leq T_i \leq \tau\}}$.  We denote this quantity by $H^*(k)$. Adding the clearing cost $C(t_k)$ in slot~$t_k$, we have $J^*(k) = H^*(k) + C(t_k)$.  

\item For $k = n+1$:  
Here, the holding cost has the same form as above, but since no clearing occurs in this 
period, the cost is $J^*(n+1) = H^*(n+1)$.
\end{enumerate}

Next, let $J(k)$ denote the increment of the objective value in Eq.~(\ref{lp:objective}) by Alg.~\ref{lp-alg}, according to the activations of all virtual packets that arrive in period~$k$. This includes the increments of 
$z_i(t)$ in Line~\ref{lp-alg:z} and of $x(t)$ in Line~\ref{lp-alg:x},  for all virtual packets~$i$ with $t_{k-1}+1 \leq T_i \leq t_k$ and for all slots $t$. 
The  objective value computed by Alg.~\ref{lp-alg} is then 
$\sum_{k=1}^{n+1} J(k)$. 
We analyze $J(k)$ in two cases.

\begin{enumerate}
  \item For $k \in \{1,\cdots,n\}$:
 By the condition in Line~\ref{lp-alg:condition}, a virtual packet contributes to the 
objective value only when it activates. If a virtual packet $i$ arriving in period~$k$
activates in some slot $t$, then  
Line~\ref{lp-alg:z} increases  $z_i(t)$  by 
$1-\sum_{\tau=T_i}^t x(\tau)$, and  Line~\ref{lp-alg:x} increases  $x(t)$ by 
$(1/C_M)(\sum_{\tau=T_i}^t x(\tau)) + (1/(\theta C_M))$. 
Hence, one activation  increases the objective value by  
\begin{align}
   & C(t)\left( \frac{1}{C_M} \sum_{\tau=T_i}^t x(\tau) + \frac{1}{\theta C_M} \right) 
     + \left( 1 - \sum_{\tau=T_i}^t x(\tau) \right) \nonumber \\
   \leq \;& 1 + \frac{1}{\theta} \quad \text{(since $C(t) \leq C_M$).}    \label{eq:increment}
\end{align}

Note that $H^*(k)$ exactly counts the number of iterations of Line~\ref{lp-alg:for} 
from slot $t_{k-1}+1$ through slot $t_k-1$  for the virtual packets that arrive  in period~$k$. Thus, the virtual packets can activate at most $H^*(k)$ times before slot~$t_k$. 
Moreover, by Lemma~\ref{lemma:max-iteration}, the virtual packets arriving in period~$k$ 
can activate at most $\lceil C_m \rceil$ additional times from slot~$t_k$ onward. 
Hence, they can activate at most $H^*(k)+\lceil C_m \rceil$ times in total. 
Combining this with Eq.~\eqref{eq:increment}, we obtain
  \begin{align}
    J(k)
      &\le \left(1+\frac{1}{\theta}\right) \big(H^*(k)+\lceil C_m\rceil\big) \label{eq:appendix:theorem:competitive-ratio}\\
      &\le \left(1+\frac{1}{\theta}\right)\frac{\lceil C_m\rceil}{C_m}\,\big(H^*(k)+C(t_k)\big) \nonumber\\
      &\le \left(1+\frac{1}{C_m}\right) \left(1+\frac{1}{\theta}\right) J^*(k). \nonumber
  \end{align}

  \item For $k=n+1$:
  Since Alg.~\ref{lp-alg} terminates in slot $t_{n+1}$ with no  clearing, the virtual packets arriving in this terminal period
  can activate at most $H^*(n+1)$ times. Thus, we have 
  \begin{align*}
    J(n+1)
      &\le \left(1+\frac{1}{\theta}\right) H^*(n+1) \\
      &\le \left(1+\frac{1}{C_m}\right) \left(1+\frac{1}{\theta}\right) J^*(n+1).
  \end{align*}
\end{enumerate}

Combining both cases yields
\[
  \sum_{k=1}^{n+1} J(k)
  \;\le\; \left(1+\frac{1}{C_m}\right) \left(1+\frac{1}{\theta}\right)
  \sum_{k=1}^{n+1} J^*(k).
\]
Substituting $\mathrm{OPT}(\mathcal{I})=\sum_{k=1}^{n+1} J^*(k)$ and $\theta = (1+(1/C_M))^{C_m} - 1$ proves the first part of the theorem.

For the asymptotic result, as $C_u\to\infty$ we have $(1+1/C_m)\to 1$ and $(1+(1/C_M))^{C_m} = (1+(1/C_M))^{C_M/R}\to e^{1/R}$,
so the ratio approaches $e^{1/R}/(e^{1/R}-1)$.

\section{Proof of Lemma~\ref{lemma:max-num-packet}}\label{appendix:lemma:max-num-packet}
Fix a slot~$t$. First, if the total number of virtual packets that have arrived by slot~$t$ is less than $2\left\lceil \sqrt{\Delta A_M  C_m} \right\rceil$, then the result is immediate. Otherwise, suppose that the number of virtual packet arrivals by slot~$t$ is at least $2\left\lceil \sqrt{\Delta A_M  C_m} \right\rceil$. 

Let $i' = \arg\max_i \{ T_i \leq t \}$ denote the most recently arrived virtual packet  by slot~$t$. Since at most $\Delta A_M$ virtual packets can arrive in a single slot, the arrival of $\left\lceil \sqrt{\Delta A_M C_m} \right\rceil$ virtual packets requires at least $\lceil \sqrt{C_m/\Delta A_M} \rceil$ slots. Therefore, virtual packet  $i' - \left\lceil \sqrt{\Delta A_M C_m} \right\rceil + 1$ must have arrived by slot $t - \lceil \sqrt{C_m/\Delta A_M} \rceil + 1$. This implies that the set 
$P(t - \lceil \sqrt{C_m/\Delta A_M} \rceil + 1)$ contains the subset $\{ i \in \mathbb{N} : 
   i' - 2\left\lceil \sqrt{\Delta A_M C_m} \right\rceil + 1 
   \leq i \leq 
   i' - \left\lceil \sqrt{\Delta A_M C_m} \right\rceil
\}$,
which consists of $\left\lceil \sqrt{\Delta A_M C_m} \right\rceil$ virtual packets.  

From slot $t - \lceil \sqrt{C_m/\Delta A_M} \rceil + 1$ through slot~$t$ 
(a total of $\lceil \sqrt{C_m/\Delta A_M} \rceil$ slots), 
these $\left\lceil \sqrt{\Delta A_M C_m} \right\rceil$ virtual packets in 
$P(t - \lceil \sqrt{C_m/\Delta A_M} \rceil + 1)$  have at least $\lceil C_m \rceil$ opportunities 
to activate. By Lemma~\ref{lemma:max-iteration}, the virtual packet 
\mbox{$i' - 2\left\lceil \sqrt{\Delta A_M C_m} \right\rceil + 1$} in this set will no longer satisfy the 
condition in Line~\ref{lp-alg:condition} at the end of slot~$t$. 
Therefore, at most $2\left\lceil \sqrt{\Delta A_M C_m} \right\rceil$ virtual packets can satisfy the 
condition at the end of slot~$t$.

\section{Proof of Lemma~\ref{lemma:converse}} \label{appendix:lemma:converse}
Consider the initial age $A_0=1$, a fixed age increment sequence $\boldsymbol{\Delta A} = (0,  \cdots, 0)$, and a fixed update (or clearing) cost sequence $\mathbf{C} = (C_m, C_m R, C_m R, \cdots, C_m R)$. Only the operation duration $T$ is unknown to the device.    Since the age no longer increases after the first update, the device will not transmit again beyond the first update. Thus, the scheduling problem reduces to deciding when to send a single update (i.e., deciding when to clear the virtual queue once) under uncertainty in $T$.

To simplify the analysis, we rescale the objective function in Eq.~\eqref{eq:cost-2} by dividing it by $C_m$, resulting in 
\begin{align*}
\sum_{t=1}^T \left(\frac{C(t)}{C_m}  d(t) + \sum_{i: T_i \leq t} \frac{1}{C_m} z_i(t)	\right).
\end{align*}
This transformation does not alter the optimal solution. We redefine $C(t)/C_m$ as the new (normalized) clearing cost. Under the given instance, the transformation yields a clearing cost of $1$ in slot~$1$, and a clearing cost of $R$ in all subsequent slots. Moreover, the term $1/C_m$ can be interpreted as the cost of holding a virtual packet for a slot of duration $1/C_m$, under the convention that holding a virtual packet for one unit  time incurs a unit cost. As $C_m \to \infty$, the slot length approaches zero, and the problem transitions to a continuous-time scheduling model, similar to prior studies \cite{bamas2020primal}. In this continuous-time setting, we assume that time starts at $0$. The clearing cost becomes $C(0) = 1$ and $C(t) = R$ for all $t > 0$.  Moreover, we consider a time horizon $T \in (0, 1]$, which is unknown to  the device .

We now establish a lower bound on the competitive ratio of any randomized online scheduling algorithm. Let $p(t)$ denote the probability density function (PDF) describing the randomized clearing time. Since $C(0) = 1$ and $C(t)=R \geq  1$ for all $t\in (0, 1]$, the virtual server optimally clears the virtual queue before time~1. Thus, the PDF $p(t)$ of the randomized clearing time must satisfy the  condition $\int_0^1 p(t)\,dt = 1$.

For a given realization of $T \in (0,1]$, the expected cost incurred by the randomized algorithm is $\int_0^T (R + t) p(t) \,dt + \int_T^1 T  p(t) \,dt$,
where the first term accounts for the cost incurred when the virtual server decides to clear by time~$T$, and the second term accounts for the cost incurred when the virtual server decides to  clear after time~$T$. Moreover, for this instance, since $T \leq 1$, an  optimal offline strategy is to idle for all time, incurring the minimum total cost of $T$. Let $c$ be the competitive ratio of the randomized algorithm. Then, we  have $\int_0^T (R + t) p(t) \,dt + \int_T^1 T  p(t) \,dt \leq c T$ for all  $T \in  (0, 1]$. Thus, we derive the following optimization problem to find the smallest achievable competitive ratio $c$:
\begin{subequations} \label{lower-bound}
\begin{align}
\min_{c, \,p(\cdot) \geq 0} \, & c \label{lower-bound:objective} \\
\text{s.t.} \, &
\int_0^T (R + t) p(t) \,dt + \int_T^1 T  p(t) \,dt \leq c T,\nonumber \\
& \hspace{3.5cm} \text{for all $T \in (0,1]$}; \label{lower-bound:const1} \\
& \int_0^1 p(t) \,dt = 1. \label{lower-bound:const2}
\end{align}
\end{subequations}

We propose a candidate solution of the form $p(t) = K e^{t/R}$ for some constant $K$. Substituting this into the  constraint in Eq.~\eqref{lower-bound:const2}, we can obtain $K = 1/(R(e^{1/R} - 1))$, which yields $p(t) = 1/(R(e^{1/R} - 1)) \cdot e^{t/R}$.
Substituting this density into the left-hand side of constraint~\eqref{lower-bound:const1}, we obtain 
\begin{align*}
\int_0^T (R + t) p(t) \,dt + \int_T^1 T p(t) \,dt = \frac{e^{1/R}}{e^{1/R} - 1}  T.	
\end{align*}
Comparing with the right-hand side $cT$, we conclude that 
\begin{align*}
c \geq \frac{e^{1/R}}{e^{1/R} - 1},	
\end{align*}
which establishes the desired lower bound on the competitive ratio.

\section{Proof of Lemma~\ref{lemma:sum-x-bound-ml}} \label{appendix:lemma:sum-x-bound-ml}
%Substituting Line~\ref{lp-alg-ml:e1} of Alg.~\ref{lp-alg-ml} with $x(t) \leftarrow x(t) + \frac{1}{C_M} \sum_{\tau=T_i}^{t} x(\tau) + \frac{1}{\theta C_M}$, we can reduce the value of  $\sum_{\tau=T_i}^t x(\tau)$ for all $t$. Therefore, it suffice to prove the result  under this modified update.
%

Let $S$ denote the value of $\sum_{\tau = T_i}^{\infty} x(\tau)$ at the beginning of some iteration of Line~\ref{lp-alg-ml:for1} in a  slot. Suppose that a slow step is followed by a fast step, and let $S_{s \rightarrow f}$ represent the resulting value of $\sum_{\tau = T_i}^{\infty} x(\tau)$ after these two steps. By Eq.~(\ref{eq:cum-x-update}), we have
\begin{align*}
S_{s \rightarrow f}
&= \underbrace{\left(1 + \frac{1}{C_M}\right)\underbrace{\left(\left(1 + \frac{1}{C_M}\right)  S + \frac{1}{\theta_s C_M} \right)}_{\sum_{\tau = T_i}^{\infty} x(\tau)\big|_{\text{after the slow step}}} + \frac{1}{\theta_f C_M}}_{\sum_{\tau = T_i}^{\infty} x(\tau)\big|_{\text{after the fast step}}} \\
&= S\left(1 + \frac{1}{C_M}\right)^2 + \frac{1}{\theta_s C_M} + \frac{1}{\theta_s C_M^2} + \frac{1}{\theta_f C_M}.
\end{align*}
Similarly, let $S_{f \rightarrow s}$ denote the resulting value of $\sum_{\tau = T_i}^{\infty} x(\tau)$ after applying a fast step followed by a slow step.  We have
\begin{align*}
S_{f \rightarrow s}
= S\left(1 + \frac{1}{C_M}\right)^2 + \frac{1}{\theta_f C_M} + \frac{1}{\theta_f C_M^2} + \frac{1}{\theta_s C_M}.
\end{align*}
Since $\theta_s \geq \theta_f$, we have
 $S_{s \rightarrow f} \leq S_{f \rightarrow s}$. Therefore, swapping a fast step and a subsequent slow step can reduce the value of $\sum_{\tau = T_i}^{\infty} x(\tau)$. Hence, to prove the desired lower bound, we can assume that all $N_s$ slow steps occur first, followed by all $N_f$ fast steps.

Next, we  prove the bound by induction on $N_f$. When \mbox{$N_f = 0$}, the result follows from Lemma~\ref{lemma:sum-x-bound}, which gives 
\begin{align*}
\left(\sum_{\tau = T_i}^{\infty} x(\tau)\bigg|_{N_f=0}\right) \geq \frac{(1 + \frac{1}{C_M})^{N_s} - 1}{\theta_s},	
\end{align*}
for all $i\in P(t)$. Assume that the result holds for  $N_f-1$, i.e.,
\begin{align*}
\left(\sum_{\tau = T_i}^{\infty} x(\tau)\bigg|_{N_f-1} \right) &\geq  \frac{(1 + \frac{1}{C_M})^{N_s} - 1}{\theta_s}  \left(1 + \frac{1}{C_M} \right)^{N_f-1} \\
& \quad + \frac{(1 + \frac{1}{C_M})^{N_f-1} - 1}{\theta_f},
\end{align*}
for all $i\in P(t)$.	 We show the result also holds for $N_f$: after an additional fast step, by Eq.~(\ref{eq:cum-x-update}) we have
\begin{align*}
&\quad\left(\sum_{\tau = T_i}^{\infty} x(\tau)\bigg|_{N_f}\right)\\
 & = \left(1+\frac{1}{C_M}\right) \left(\sum_{\tau = T_i}^{\infty} x(\tau)\bigg|_{N_f-1}\right)+ \frac{1}{\theta_f C_M}\\
 & \geq   \left(1 + \frac{1}{C_M} \right) \left[ \frac{(1 + \frac{1}{C_M})^{N_s} - 1}{\theta_s}  \left(1 + \frac{1}{C_M} \right)^{N_f-1}\right. \\
& \quad \left.+ \frac{(1 + \frac{1}{C_M})^{N_f-1} - 1}{\theta_f} \right] + \frac{1}{\theta_f C_M} \\
 &= \frac{(1 + \frac{1}{C_M})^{N_s} - 1}{\theta_s} \left(1 + \frac{1}{C_M} \right)^{N_f} + \frac{(1 + \frac{1}{C_M})^{N_f} - 1}{\theta_f},
\end{align*}
for all $i\in P(t)$. This completes the inductive step and proves the lemma.

\section{Proof of Lemma~\ref{lemma:max-iteration-ml}} \label{appendix:lemma:max-iteration-ml}
Fix a slot~$t$. We claim that if $N_s \lambda + N_f \ge C_m$, then 
$\sum_{\tau = T_i}^{\infty} x(\tau) \ge 1$ for all $i\in P(t)$, implying that the condition in 
Line~\ref{lp-alg-ml:condition} no longer holds.

To prove the claim, it suffices to consider the case where the $N_s$ slow steps are followed by the 
$N_f$ fast steps (as in Appendix~\ref{appendix:lemma:sum-x-bound-ml}). Applying 
Lemma~\ref{lemma:sum-x-bound-ml} under $N_s \lambda + N_f \ge C_m$ yields
\begin{align}
\sum_{\tau = T_i}^{\infty} x(\tau) 
& \ge 
\frac{\left(1 + \frac{1}{C_M}\right)^{(C_m - N_f)/\lambda} - 1}{\theta_s} 
\left(1 + \frac{1}{C_M} \right)^{N_f} \nonumber\\
& \quad + \frac{\left(1 + \frac{1}{C_M} \right)^{N_f} - 1}{\theta_f},
\label{eq:iterative-proof}
\end{align}
for all $i\in P(t)$. If $N_f = 0$, then Eq.~(\ref{eq:iterative-proof}) (after the inequality) equals $1$. Next, we show that Eq.~(\ref{eq:iterative-proof}) is nondecreasing in $N_f$. Differentiating it with respect to $N_f$ gives
\begin{align*}
&\ln\left(1 + \frac{1}{C_M}\right) \left(1 + \frac{1}{C_M}\right)^{N_f} \\
& \quad \times \left[
-\frac{\frac{1-\lambda}{\lambda} 
\left(1+\frac{1}{C_{M}}\right)^{(C_m - N_f)/\lambda} + 1}{\theta_{s}}
+\frac{1}{\theta_{f}}
\right].	
\end{align*}
To show the derivative is nonnegative, we examine the bracketed term:
\begin{align}
& -\frac{\frac{1 - \lambda}{\lambda} \left(1 + \frac{1}{C_M} \right)^{(C_m - N_f)/\lambda} + 1}{\theta_s} 
+ \frac{1}{\theta_f} \nonumber \\
\mathop{\ge}^{(a)}\ 
& -\frac{\frac{1 - \lambda}{\lambda} \left(1 + \frac{1}{C_M} \right)^{C_m/\lambda} + 1}{\theta_s} 
+ \frac{1}{\theta_f} \nonumber \\
=\ 
& -\frac{\frac{1 - \lambda}{\lambda} \left(1 + \frac{1}{C_M} \right)^{C_M/(R \lambda)} + 1}{
\left(1 + \frac{1}{C_M} \right)^{C_M/(R \lambda)} - 1}
+ \frac{1}{\left(1 + \frac{1}{C_M} \right)^{(C_M \lambda)/R} - 1},
\label{eq:check-sign}
\end{align}
where $(a)$ sets $N_f=0$ (since $(1 + (1/C_M))^{(C_m - N_f)/\lambda}$ decreases in $N_f$).
From \cite[Page~27]{bamas2020primal}, we can write $(1 + (1/C_M))^{C_M} = e^{x}$ for some $x \in (0,1)$.
Let $x' = x/R \in (0,1)$. Then, Eq.~\eqref{eq:check-sign} becomes
\[
-\frac{\frac{1 - \lambda}{\lambda} e^{x'/\lambda} + 1}{e^{x'/\lambda} - 1}
+ \frac{1}{e^{x' \lambda} - 1},
\]
which is known to be nonnegative \cite[Page~27]{bamas2020primal}. Hence,   Eq.~\eqref{eq:iterative-proof} is nondecreasing as its  derivative  is nonnegative, proving the claim.

\section{Proof of Theorem~\ref{theorem:robustness}}\label{appendix:theorem:robustness}
We follow Appendix~\ref{appendix:theorem:competitive-ratio}. 
Consider a period $k \in \{1, \cdots, n\}$. 
Here, a single slow or fast step activation can increase the objective value in Eq.~(\ref{lp:objective}) by at most $1 + (1/\theta_s)$ or $1 + (1/\theta_f)$, respectively.  
Since $1 + (1/\theta_s) \leq 1 + (1/\theta_f)$, the increment of the objective value 
 due to the activations of all virtual packets arriving in period~$k$, 
from slot $t_{k-1}+1$ up to slot $t_k-1$, is bounded above by
\begin{align}
   H^*(k)\left(1 + \frac{1}{\theta_f}\right). 
   \label{eq:theorem:robustness:bound1}
\end{align}

Moreover, let $N_s(k)$ and $N_f(k)$ denote the numbers of slow and fast steps, respectively, 
performed by the virtual packets arriving in period~$k$, from slot $t_k$ onward. 
Then, the increment of the objective value  due to these $N_s(k)+N_f(k)$ activations is
\begin{align}
   &\quad N_s(k) \left(1 + \frac{1}{\theta_s}\right) + N_f(k) \left(1 + \frac{1}{\theta_f}\right) \nonumber\\
   &\mathop{\leq}^{(a)} \frac{C_m + 1 - N_f(k)}{\lambda} \left(1 + \frac{1}{\theta_s}\right) 
   + N_f(k) \left(1 + \frac{1}{\theta_f}\right),
   \label{eq:theorem:robustness:bound2}
\end{align}
where $(a)$ is because $N_s(k) \lambda + N_f(k) \leq C_m + 1$ from Lemma~\ref{lemma:max-iteration-ml}. 
Differentiating Eq.~(\ref{eq:theorem:robustness:bound2}) (after the inequality) with respect to $N_f(k)$ gives
\begin{align*} &\quad -\frac{1}{\lambda} \left(1+\frac{1}{\theta_s}\right)+ \left(1+\frac{1}{\theta_f}\right)\\ &=-\frac{1}{\lambda} \left(1+\frac{1}{\left(1+\frac{1}{C_M}\right)^{C_m/\lambda}-1}\right) \\
& \quad +\left(1+\frac{1}{\left(1+\frac{1}{C_M}\right)^{C_m \lambda}-1}\right). 
\end{align*}
Following Appendix~\ref{appendix:lemma:max-iteration-ml}, this derivative can be rewritten as
\begin{align*}
   &\quad -\frac{1}{\lambda}\left(1 + \frac{1}{e^{x'/\lambda}-1}\right)  
   + \left(1 + \frac{1}{e^{x'\lambda}-1}\right) \\
   &= -\frac{1}{\lambda} \left(\frac{1}{1 - e^{-x'/\lambda}}\right) 
   + \frac{1}{1 - e^{-x'\lambda}},
\end{align*}
for some $x' \in (0,1)$. This expression is known to be nonnegative 
\cite[Page~27]{bamas2020primal}. Hence, Eq.~(\ref{eq:theorem:robustness:bound2}) is nondecreasing 
in $N_f(k)$, and its maximum is following  attained at $N_f(k)=C_m+1$,
\begin{align}
   \big(C_m+1\big) \left(1 + \frac{1}{\theta_f}\right).
   \label{eq:theorem:robustness:bound3}
\end{align}

Combining Eqs.~\eqref{eq:theorem:robustness:bound1} and \eqref{eq:theorem:robustness:bound3}, 
we obtain
\begin{align*}
   J(k) &\leq \left(1 + \frac{1}{\theta_f}\right) \left(H^*(k) + C_m + 1\right) \\
        &\leq \left(1 + \frac{1}{\theta_f}\right)\frac{C_m+1}{C_m}\left(H^*(k) + C(t_k)\right) \\
        &\leq \left(1 + \frac{1}{C_m}\right) \left(1 + \frac{1}{\theta_f}\right) J^*(k).
\end{align*}
Finally, following the line of Appendix~\ref{appendix:theorem:competitive-ratio} 
and substituting the definition of $\theta_f$ completes the proof.

\section{Proof of Theorem~\ref{theorem:consistency}}\label{appendix:theorem:consistency}
Fix an instance $\mathcal{I}$ and  ML advice~$\boldsymbol{\mathcal{M}}$. 
We follow Appendix~\ref{appendix:theorem:competitive-ratio}, with minor modifications. 
Redefine $t_k$ for \mbox{$k \in \{1, \cdots, n\}$} as the slot when $\boldsymbol{\mathcal{M}}$ clears the virtual queue 
for the $k$-th time. These redefined time points determine a new set of periods, replacing the periods defined in Appendix~\ref{appendix:theorem:competitive-ratio}.

Let $J_{\boldsymbol{\mathcal{M}}}(k)$ denote the cost incurred by $\boldsymbol{\mathcal{M}}$ in period~$k$. 
Then, the total cost in Eq.~(\ref{eq:cost-2}) incurred by $\boldsymbol{\mathcal{M}}$ is 
$\sum_{k=1}^{n+1} J_{\boldsymbol{\mathcal{M}}}(k)$. Let $H_{\boldsymbol{\mathcal{M}}}(k)$ be the holding cost incurred by $\boldsymbol{\mathcal{M}}$ for all virtual packets arriving in period~$k$.
Following Appendix~\ref{appendix:theorem:competitive-ratio}, we have 
$J_{\boldsymbol{\mathcal{M}}}(k)=H_{\boldsymbol{\mathcal{M}}}(k)+C(t_k)$ for all $k \in \{1, \cdots, n\}$, 
and $J_{\boldsymbol{\mathcal{M}}}(n+1)=H_{\boldsymbol{\mathcal{M}}}(n+1)$.

Let $J(k)$ be the increment of the objective value in Eq.~(\ref{lp:objective}) by Alg.~\ref{lp-alg-ml}, according to the slow and fast step 
activations of all virtual packets arriving in period~$k$. Consider a fixed $k\in \{1, \cdots, n\}$. 
The virtual packets arriving in period~$k$ activate only slow steps from slot~$t_{k-1}+1$ until slot~$t_k-1$ (before advising clearing). 
Each slow step activation increases the objective value by at most $1+(1/\theta_s)$, 
so the total increment of the objective value in this interval is at most $(1+(1/\theta_s))H_{\boldsymbol{\mathcal{M}}}(k)$. Moreover, after  advising clearing in slot~$t_k$, the same virtual packets activate only fast steps. 
Following the proof of Lemma~\ref{lemma:max-iteration}, there are at most $\lceil C_m \lambda \rceil$ such activations, 
each increasing the objective value by at most $1+(1/\theta_f)$. 
Thus, the total increment of the objective value after slot~$t_k$ is at most $(1+(1/\theta_f))\lceil C_m \lambda \rceil$. 

Hence, we have
\begin{align*}
J(k)  
&\leq \left(1+\frac{1}{\theta_s}\right)H_{\boldsymbol{\mathcal{M}}}(k) + \left(1+\frac{1}{\theta_f}\right)\lceil C_m \lambda \rceil\\
&= \left(1+\frac{1}{\theta_s}\right)H_{\boldsymbol{\mathcal{M}}}(k) + \left(1+\frac{1}{\theta_f}\right)\frac{\lceil C_m \lambda \rceil}{C_m} \,C_m\\
&\leq \max \left\{1+\frac{1}{\theta_s}, \frac{\lceil C_m \lambda \rceil}{C_m}\left(1+\frac{1}{\theta_f}\right)\right\}\big(H_{\boldsymbol{\mathcal{M}}}(k)+ C_m \big)\\
&\leq \max \left\{1+\frac{1}{\theta_s}, \frac{\lceil C_m \lambda \rceil}{C_m}\left(1+\frac{1}{\theta_f}\right)\right\}\big(H_{\boldsymbol{\mathcal{M}}}(k)+ C(t_k)\big)\\
&= \max \left\{1+\frac{1}{\theta_s}, \frac{\lceil C_m \lambda \rceil}{C_m}\left(1+\frac{1}{\theta_f}\right)\right\} J_{\boldsymbol{\mathcal{M}}}(k).
\end{align*}
Similarly, we also have
\[
J(n+1) \leq\max \left\{1+\frac{1}{\theta_s}, \frac{\lceil C_m \lambda \rceil}{C_m}\left(1+\frac{1}{\theta_f}\right)\right\} 
J_{\boldsymbol{\mathcal{M}}}(n+1).
\]

Combining the two cases yields
\[
\sum_{k=1}^{n+1} J(k)\leq\max \left\{1+\frac{1}{\theta_s}, \frac{\lceil C_m \lambda \rceil}{C_m}\left(1+\frac{1}{\theta_f}\right)\right\} 
\sum_{k=1}^{n+1} J_{\boldsymbol{\mathcal{M}}}(k).
\]
Substituting $J(\mathcal{I}, \boldsymbol{\mathcal{M}})=\sum_{k=1}^{n+1} J_{\boldsymbol{\mathcal{M}}}(k)$ 
and the definitions of $\theta_s$ and $\theta_f$ proves the first part of the theorem.

Finally,  following the derivation in Appendix~\ref{appendix:theorem:robustness}, 
we obtain the asymptotic behavior of the bound as $C_u \to \infty$.

\section{Proof of Lemma~\ref{lemma:converse-ml}} \label{appendix:lemma:converse-ml}
We use the same continuous-time instance as in Appendix~\ref{appendix:lemma:converse}. 
Consider a $\big(\lambda e^{ \lambda/R}\big)/(e^{ \lambda/R}-1)$-consistent scheduling algorithm that 
chooses a random update (or equivalently clearing) time $t\in[0,1]$ with PDF $p(t)$ (so $\int_0^1 p(t)\,dt=1$).  Let $c$ denote the robustness factor.
Following Appendix~\ref{appendix:lemma:converse}, we have 
\begin{align}
\int_0^T (R+t)p(t)\,dt + \int_T^1 T p(t)\,dt \le cT,
\label{eq:robust-constraint}	
\end{align}
 for all $T \in (0,1]$.

Setting $T=1$ and assuming perfect ML advice (updating at time $0$), the ML advice incurs  a cost of $1$, 
while the online algorithm incurs a cost of $\int_0^1 (R+t)\,p(t)\,dt$. 
By $\big(\lambda e^{ \lambda/R}\big)/(e^{ \lambda/R}-1)$-consistency, we have 
\begin{equation}
\int_0^1 (R+t)\,p(t)\,dt \;\le\; \frac{\lambda e^{\lambda/R}}{e^{\lambda/R}-1}\cdot 1. 
\label{eq:consistency-constraint}
\end{equation}

We now lower bound the optimal robustness $c$ subject to 
Eqs.~\eqref{eq:robust-constraint} and~\eqref{eq:consistency-constraint}, and $\int_0^1 p(t)\,dt=1$:
\begin{subequations}\label{lower-bound2}
\begin{align}
\min_{c,\, p(\cdot)} \, & c \\
\text{s.t. }\,
& \int_0^T (R+t)\,p(t)\,dt + \int_T^1 Tp(t)\,dt \le cT, \nonumber\\
& \hspace{3.5cm}\text{for all $T \in (0,1]$}; \label{lower-bound2:const1}\\
& \int_0^1 (R+t)\,p(t)\,dt \le \frac{\lambda e^{\lambda/R}}{e^{\lambda/R}-1}; 
\label{lower-bound2:const2}\\
& \int_0^1 p(t)\,dt = 1. \label{lower-bound2:const3}
\end{align}
\end{subequations}

By weak duality, any feasible solution to the dual of Eq.~\eqref{lower-bound2} yields a lower bound on $c$. 
Let $\eta(T)$, $\mu$, and $\nu$ be the dual variables for 
Eqs.~\eqref{lower-bound2:const1}, \eqref{lower-bound2:const2}, and \eqref{lower-bound2:const3}, respectively. Then, the dual program can be written as follows:
\begin{subequations}\label{lower-bound2-dual}
\begin{align}
\max_{\eta(\cdot),\mu,\nu}\quad 
& \nu - \mu \,\frac{\lambda e^{\lambda/R}}{e^{\lambda/R}-1} 
\label{lower-bound2-dual:objective}\\
\text{s.t.}\quad 
& \int_0^1 T\eta(T)\,dT = 1; \label{lower-bound2-dual:const1}\\
& \nu - (R+t)\mu  \le \int_0^t T\eta(T)\,dT \nonumber\\
&\hspace{2.8cm}+ (R+t)\int_t^1 \eta(T)\,dT,\nonumber\\
&\hspace{3cm} \text{for all $t\in[0,1]$}.
\label{lower-bound2-dual:const2}
\end{align}
\end{subequations}

Next, we propose a feasible solution to the optimization problem in Eq.~(\ref{lower-bound2-dual}). We propose $\eta(T) = K e^{-T/R} \mathbf{1}_{\{T \leq \lambda\}}$ for some constant $K$ to be determined, where $\mathbf{1}_{\{\cdot\}}$ is the indicator function. Substituting this form into Eq.~(\ref{lower-bound2-dual:const1})  yields $K = 1/(R^{2} - R^{2} e^{-\lambda/R} - R \lambda e^{-\lambda/R})$. 

We further propose $\nu = aK$ and $\mu = bK e^{-\lambda/R}$ for some constants $a$ and $b$ to be determined. Substituting these into the objective in Eq.~\eqref{lower-bound2-dual:objective} gives
\begin{align} 
&\nu - \mu \frac{\lambda e^{\lambda/R}}{e^{\lambda/R} - 1} \nonumber\\ =\,& K\left(a - b e^{-\lambda/R}  \frac{\lambda e^{\lambda/R}}{e^{\lambda/R} - 1} \right) \nonumber\\ =\,& \frac{1}{R^{2} - R^{2} e^{-\lambda/R} - R \lambda e^{-\lambda/R}} \left(a - a e^{-\lambda/R} - b \lambda e^{-\lambda/R}\right)\nonumber\\
&  \times \left(\frac{e^{\lambda/R}}{e^{\lambda/R} - 1}\right). 
\label{eq:lower-bound-cr} 
\end{align}
To make the objective equal to $(e^{\lambda/R})/(e^{\lambda/R} - 1)$ (the claimed robustness bound), we choose $a = R^{2}$ and $b = R$.

Next, we verify that the chosen values satisfy Eq.~\eqref{lower-bound2-dual:const2}.  
We consider two cases:
\begin{enumerate}
	\item For $t \leq \lambda$: The left-hand side of Eq.~\eqref{lower-bound2-dual:const2} (before the inequality) is
	\[
	\nu - (R+t)\mu 
	= K\big(R^2 - R^2 e^{-\lambda/R} - tR e^{-\lambda/R}\big).
	\]
	The right-hand side is
	\begin{align*}
	& \quad \int_{0}^{t} T \eta(T) \,dT + (R+t) \int_{t}^{1} \eta(T)\,dT \\
	&= K\left(\int_{0}^{t} T e^{-T/R} \, dT + (R+t) \int_{t}^{\lambda} e^{-T/R}  \,dT\right) \\
	&= K\big(R^{2} - R^{2} e^{-\lambda/R} - t R e^{-\lambda/R}\big),		
	\end{align*}
	which matches the left-hand side.
	\item For  $t > \lambda$: The left-hand side is  
	\begin{align*}
	\nu - (R+t)\mu 
	&\leq \nu - (R+\lambda)\mu \\
	&= K\big(R^{2} - R^{2} e^{-\lambda/R} - \lambda R e^{-\lambda/R}\big) \\
	&= 1,
	\end{align*}
	The right-hand side is
	\begin{align*}
	&\quad	\int_{0}^{t} T \eta(T)\,dT + (R+t) \int_{t}^{1} \eta(T)\,dT\\
	&= \int_{0}^{\lambda} T \eta(T)\,dT = 1,	
	\end{align*}
	so the inequality holds.
\end{enumerate}
Therefore, Eq.~\eqref{lower-bound2-dual:const2} is satisfied in both cases.  
By the weak duality theorem, the minimum possible robustness $c$ is at least 
\[
\frac{e^{\lambda/R}}{e^{\lambda/R}-1},
\]
as stated in Eq.~\eqref{eq:lower-bound-cr}.

\section{Proof of Lemma~\ref{lemma:x-upper-bound}} \label{appendix:lemma:x-upper-bound}

Fix a slot~$t$. We prove the result by induction on $n$. 
When $n=1$, by Eq.~(\ref{eq:cum-x-update}) we have
\begin{align*}
&\quad \left(\sum_{\tau=T_i}^{\infty} x(\tau) \Big|_{n=1}\right) 
- \left(\sum_{\tau=T_i}^{\infty} x(\tau) \Big|_{n=0}\right)\\
&= \frac{1}{C_M}\left(\sum_{\tau=T_i}^{\infty} x(\tau) \Big|_{n=0}\right)
+ \frac{1}{\theta C_M} \\
&\le \frac{1}{C_M}\cdot 1 + \frac{1}{\theta C_M} \\
&=\left(1+\frac{1}{\theta}\right)\!\left[\left(1+\frac{1}{C_M}\right)^1-1\right],
\end{align*}
for all $i \in P(t)$.

Assume the result holds for $n-1$, i.e.,
\begin{align*}
& \quad \left(\sum_{\tau=T_i}^{\infty} x(\tau) \Big|_{n-1}\right) 
- \left(\sum_{\tau=T_i}^{\infty} x(\tau) \Big|_{0}\right)\\
&\le \left(1+\frac{1}{\theta}\right)\!\left[\left(1+\frac{1}{C_M}\right)^{n-1}-1\right],
\end{align*}
for all $i \in P(t)$.

We show that the result holds for $n$:
\begin{align}
&\quad \left(\sum_{\tau=T_i}^{\infty} x(\tau) \Big|_{n}\right) 
- \left(\sum_{\tau=T_i}^{\infty} x(\tau) \Big|_{0}\right) \nonumber\\
&= \left(\sum_{\tau=T_i}^{\infty} x(\tau) \Big|_{n}\right) 
- \left(\sum_{\tau=T_i}^{\infty} x(\tau) \Big|_{n-1}\right)\nonumber\\
&\quad+ \left(\sum_{\tau=T_i}^{\infty} x(\tau) \Big|_{n-1}\right) 
- \left(\sum_{\tau=T_i}^{\infty} x(\tau) \Big|_{0}\right) \nonumber\\
&\le \frac{1}{C_M}\left(\sum_{\tau=T_i}^{\infty} x(\tau) \Big|_{n-1}\right)
+ \frac{1}{\theta C_M}\nonumber\\
&\quad + \left(1+\frac{1}{\theta}\right)\!\left[\left(1+\frac{1}{C_M}\right)^{n-1}-1\right],
\label{eq:lemma:x-upper-bound-1}
\end{align}
where the inequality uses Eq.~(\ref{eq:cum-x-update}) and the inductive hypothesis; moreover, the term $\sum_{\tau=T_i}^{\infty} x(\tau) \big|_{n-1}$ can be further calculated as follows: 
\begin{align*}
&\quad \left(\sum_{\tau=T_i}^{\infty} x(\tau) \Big|_{n-1}\right)\\
&= \left(\sum_{\tau=T_i}^{\infty} x(\tau) \Big|_{0}\right)
+ \Bigg[\left(\sum_{\tau=T_i}^{\infty} x(\tau) \Big|_{n-1}\right)
- \left(\sum_{\tau=T_i}^{\infty} x(\tau) \Big|_{0}\right)\Bigg] \\
&\le 1 + \left(1+\frac{1}{\theta}\right)\!\left[\left(1+\frac{1}{C_M}\right)^{n-1}-1\right].
\end{align*}
Plugging this into Eq.~\eqref{eq:lemma:x-upper-bound-1} yields
\begin{align*}
&\quad \left(\sum_{\tau=T_i}^{\infty} x(\tau) \Big|_{n}\right) 
- \left(\sum_{\tau=T_i}^{\infty} x(\tau) \Big|_{0}\right)\\
&\le \left(1+\frac{1}{\theta}\right)\!\left[\left(1+\frac{1}{C_M}\right)^n-1\right],
\end{align*}
for all $i \in P(t)$. This completes the inductive step and proves the lemma.

\section{Proof of Theorem~\ref{theorem:competitive-ratio-non}}\label{appendix:theorem:competitive-ratio-non}
For any scheduling algorithm, a virtual packet that arrives during a virtual OFF slot must remain in the virtual queue until the next virtual ON slot. 
Hence, the holding cost accrued by the virtual packets that arrive in virtual OFF slots until the slot immediately before the next virtual ON slot is identical across all algorithms (including the offline optimum). We therefore remove this constant from the objective, which is equivalent to deferring any virtual packet arrival in a virtual OFF slot to the next virtual ON slot. Then, it  suffices to consider a fixed instance~$\mathcal{I}$ in which no virtual packet arrives in a virtual OFF slot.

We follow Appendix~\ref{appendix:theorem:competitive-ratio}. 
Fix a period $k\in\{1,\dots,n\}$ and a virtual ON slot $t$ in that period. We bound the increment of the objective value in Eq.~\eqref{lp-non:objective} in slot~$t$  by considering how $x(t)$ and the matching $z$-variables are adjusted by Alg.~\ref{lp-alg-non}. There are three mutually exclusive cases:
\begin{enumerate}
	\item If a virtual packet $i$  arrives before $\hat{t}$ and activates in iteration $t'$ of Line~\ref{lp-alg-non:all-update} in slot~$t$, then Line~\ref{lp-alg-non:x} increases $x(t)$ by
\begin{align*}
\frac{1}{C_M}\left(\sum_{\tau=T_i}^t x(\tau)\bigg|_{\text{before the activation}}\right) + \frac{1}{\theta C_M}.	
\end{align*}
However, the paired $z_i(t')$ was already set to be \mbox{$1-\sum_{\tau=T_i}^{t'}x(\tau)$} in a slot \mbox{$t' \leq t$}. Because $x(\tau)$ does not change (for all possible $\tau$) over the virtual OFF period until slot~$t$,  we have 
\begin{align*}
z_i(t')=1-\left(\sum_{\tau=T_i}^t x(\tau)\bigg|_{\text{start of slot~$t$}}\right).	
\end{align*}

Hence, the increment of the objective value  due to the adjustment of $x(t)$ from virtual packet~$i$ in iteration~$t'$ and the paired $z_i(t')$ is
\begin{align}
&C(t)\left[\frac{1}{C_M}\left(\sum_{\tau=T_i}^t x(\tau)\Big|_{\text{before the activation}}\right) + \frac{1}{\theta C_M}\right]\nonumber\\
&\quad+ 1-\left(\sum_{\tau=T_i}^t x(\tau)\bigg|_{\text{start of slot~$t$}}\right)\nonumber\\
&\le 1+\frac{1}{\theta} \nonumber\\
&+ \left(\sum_{\tau=T_i}^t x(\tau)\bigg|_{\text{before the activation}}\right) - \left(\sum_{\tau=T_i}^t x(\tau)\bigg|_{\text{start of slot~$t$}}\right).
\label{eq:non-inc-1}
\end{align}

Next, we analyze the total number of activations performed from the start of slot~$t$ until the considered activation. 
By Lemma~\ref{lemma:max-num-packet}, at most $2\lceil \sqrt{\Delta A_M C_m}\rceil$ virtual packets can satisfy the  condition in Line~\ref{lp-alg-non:for} at the end of the previous ON slot $\hat{t}-1$. Because  no virtual packets arrive during the virtual OFF period, also at most $2 \left\lceil \sqrt{\Delta A_M C_m} \right\rceil$ virtual packets  can satisfy the  condition at the beginning of slot~$t$. Moreover, since each such virtual packet can be iterated at most \(T_{\text{OFF}}\) times in Line~\ref{lp-alg-non:all-update}, the total number of activations performed from the start of slot~$t$ until the considered activation is bounded by $2 \left\lceil \sqrt{\Delta A_M C_m} \right\rceil T_{\text{OFF}}$.

Then, by Lemma~\ref{lemma:x-upper-bound}, we have
\begin{align*}
&\left(\sum_{\tau=T_i}^t x(\tau)\big|_{\text{before the activation}}\right) - \left(\sum_{\tau=T_i}^t x(\tau)\bigg|_{\text{start of slot~$t$}}\right)\\
& \le\left(1+\frac{1}{\theta}\right)\left[\left(1+\frac{1}{C_M}\right)^{2 \left\lceil \sqrt{\Delta A_M C_m} \right\rceil T_{\text{OFF}}}-1\right].	
\end{align*}

Substituting this in Eq.~\eqref{eq:non-inc-1} gives the  bound on the increment of the objective value:
\begin{align}
\left(1+\frac{1}{\theta}\right)\left(1+\frac{1}{C_M}\right)^{2\lceil \sqrt{\Delta A_M C_m}\rceil T_{\text{OFF}}}.
\label{eq:appendix:case1-final}
\end{align}

\item If a virtual packet $i$  arrives before $\hat{t}$, meets the condition in Line~\ref{lp-alg-non:for} in slot~$t' \leq t$, but does not activate in iteration~$t'$ of Line~\ref{lp-alg-non:all-update} in slot~$t$, then because the paired $z_i(t')$ was still set to $1 - \sum_{\tau=T_i}^{t'} x(\tau)$ in slot $t'$, it  contributes at most $1$ to the objective value in that iteration.

\item If a virtual packet $i$ arrives in slot~$t$ and activates in slot $t$, then by Appendix~\ref{appendix:theorem:competitive-ratio} the increment from $x(t)$ in Line~\ref{lp-alg-non:else:x} and its paired $z_i(t)$ in Line~\ref{lp-alg-non:z} is $1+(1/\theta)$,  which is also bounded above by Eq.~\eqref{eq:appendix:case1-final}.
\end{enumerate}

Next, we count the occurrences of the three cases since slot~$t_k$ for the virtual packets that arrive in the period. By Lemma~\ref{lemma:max-iteration}, the virtual packets that arrive in period $k$ can activate (in Cases~1 and~3 together) at most $\lceil C_m\rceil$ times since $t_k$. 
In addition, by Lemma~\ref{lemma:max-num-packet}, at most $2\lceil \sqrt{\Delta A_M C_m}\rceil$ virtual packets meet the condition in Line~\ref{lp-alg-non:for} at the end of slot $t_k$. Once such a virtual packet stops activating, because of at most $T_{\text{OFF}}$ iterations in Line~\ref{lp-alg-non:all-update}, the virtual packet can contribute at most $T_{\text{OFF}}$ Case~2 increments. Thus, there are at most \(2\lceil \sqrt{\Delta A_M C_m}\rceil T_{\text{OFF}}\) Case~2 increments since slot \(t_k\).

Then, we can  revise Eq.~(\ref{eq:appendix:theorem:competitive-ratio}) 
in Appendix~\ref{appendix:theorem:competitive-ratio} as
\begin{align*}
J(k) & \leq \underbrace{\left(1+\frac{1}{\theta}\right)
\left(1+\frac{1}{C_m}\right)^{2 \lceil \sqrt{\Delta A_M  C_m} \rceil T_{\text{OFF}}} H^*(k)}_{(a)}\\
&\quad + \underbrace{\left(1+\frac{1}{\theta}\right)
\left(1+\frac{1}{C_m}\right)^{2 \lceil \sqrt{\Delta A_M  C_m} \rceil T_{\text{OFF}}}\lceil C_m \rceil}_{(b)}\\
&\quad + \underbrace{1 \cdot \left(2 \lceil \sqrt{\Delta A_M  C_m} \rceil T_{\text{OFF}}\right)}_{(c)},
\end{align*}
where (a) corresponds to all cases before slot~$t_k$; (b)
corresponds to Cases~1 and~3 from slot~$t_k$ onward; (c) corresponds to Case~2 from slot~$t_k$ onward. These terms
 can be further simplified as
\begin{align*}
(a) + (b) &\le \left(1+\frac{1}{\theta}\right)
\left(1+\frac{1}{C_m}\right)^{1+2 \lceil \sqrt{\Delta A_M  C_m} \rceil T_{\text{OFF}}} J^*(k),
\end{align*}
and also
\begin{align*}
(c) &\le \frac{2 \lceil \sqrt{\Delta A_M  C_m} \rceil T_{\text{OFF}}}{C_m}\,(H^*(k)+C_m)\\
&\le \frac{2 \lceil \sqrt{\Delta A_M  C_m} \rceil T_{\text{OFF}}}{C_m} J^*(k).
\end{align*}

Finally, following Appendix~\ref{appendix:theorem:competitive-ratio} yields the desired result.

\section{Proof of Lemma~\ref{lemma:inf-t-c}}\label{appendix:lemma:inf-t-c}
First, we consider the initial age $A_0 = 1$ and a fixed update cost sequence
\mbox{$\mathbf{C} = (C_m, \cdots, C_m)$}. Consider an online algorithm that updates in slot~$t$ with probability $p(t)$. Depending on $p(1)$, the adversary constructs the operation duration~$T$, the age increment sequence $\boldsymbol{\Delta A}$, and the update opportunity sequence~$\mathbf{U}$ as follows:
\begin{enumerate}
    \item If $p(1) < 1$: Set $T = 2$, $\boldsymbol{\Delta A}=(0, C_m^2-2)$, and $\mathbf{U}=(1, 0)$. In this case, the online algorithm incurs expected total cost
    \begin{align*}
        &\quad p(1)\, C_m + (1 - p(1)) \bigl(A(1)+A(2)\bigr)\\
        &= p(1)\, C_m + (1 - p(1)) C_m^2,
    \end{align*}
    while the offline optimum updates in slot~1 and incurs total cost $C_m$.
    Hence, the competitive ratio is $p(1) + (1 - p(1)) C_m$, which diverges as $C_m \to \infty$.

    \item If $p(1) = 1$: Set $T = 1$, $\Delta A(1)=0$, and $U(1)=1$. In this case, the online algorithm incurs cost $C_m$, while the offline optimum chooses not to update and incurs total cost $1$. Hence, the competitive ratio is $C_m$, which again diverges as $C_m \to \infty$.
\end{enumerate}
In both cases, if $\Delta A_M$ can be arbitrarily large, the adversary can construct an instance that forces the competitive ratio of any online algorithm to diverge.

Second, we consider the initial age $A_0 = 1$, a fixed age increment sequence
$\boldsymbol{\Delta A} = (0, \cdots, 0)$, and a fixed update cost sequence
\mbox{$\mathbf{C} = (C_m, \cdots, C_m)$}. Depending on $p(1)$, the adversary constructs the operation duration~$T$ and the update opportunity sequence $\mathbf{U}$ as follows:
\begin{enumerate}
    \item If $p(1) < 1$: Set $T = C_m^2 + 1$ and $\mathbf{U}=(1, 0, 0, \cdots, 0)$.
    In this case, the online algorithm incurs expected total cost
    \begin{align*}
        &\quad p(1)\, C_m + (1 - p(1))(T - 1)\\
        &= p(1)\, C_m + (1 - p(1)) C_m^2,
    \end{align*}
    while the offline optimum updates in slot~1 and incurs total cost $C_m$.
    Hence, the competitive ratio is $p(1) + (1 - p(1)) C_m$, which diverges as $C_m \to \infty$.

    \item If $p(1) = 1$: Set $T = 1$ and $U(1)=1$. In this case, the online algorithm incurs cost $C_m$, while the offline optimum chooses not to update and incurs total cost $1$. Hence, the competitive ratio is $C_m$, which again diverges as $C_m \to \infty$.
\end{enumerate}
In both cases, if $T_{\text{OFF}}$ can be arbitrarily large, the adversary can also construct an instance that forces the competitive ratio of any online algorithm to diverge, completing the proof.

\end{document}